\documentclass[a4paper,reqno]{amsart}

\usepackage[backend=biber,style=numeric-comp,isbn=false,url=false,doi=false]{biblatex}
\AtEveryBibitem{\clearfield{issue}\clearfield{number}\clearfield{month}\clearfield{address}}
\renewbibmacro*{name:family}[4]{%
  \ifuseprefix
    {\usebibmacro{name:delim}{#3#1}%
     \usebibmacro{name:hook}{#3#1}%
     \mkbibcompletenamefamily{%
       \ifdefvoid{#3}
         {}
         {\mkbibnameprefix{\MakeCapital{#3}}\isdot%
          \ifprefchar{}{\bibnamedelimc}}%
       \mkbibnamefamily{#1}\isdot}}
    {\usebibmacro{name:delim}{#1}%
     \usebibmacro{name:hook}{#1}%
     \mkbibcompletenamefamily{%
       \mkbibnamefamily{#1}\isdot}}}%
\renewbibmacro*{name:family-given}[4]{%
  \ifuseprefix
    {\usebibmacro{name:delim}{#3#1}%
     \usebibmacro{name:hook}{#3#1}%
     \mkbibcompletenamefamilygiven{%
       \ifdefvoid{#3}
         {}
         {\mkbibnameprefix{\MakeCapital{#3}}\isdot
          \ifprefchar{}{\bibnamedelimc}}%
       \mkbibnamefamily{#1}\isdot
       \ifdefvoid{#4}
         {}
         {\bibnamedelimd\mkbibnamesuffix{#4}\isdot}%
       \ifdefvoid{#2}
         {}
         {\revsdnamepunct\bibnamedelimd\mkbibnamegiven{#2}\isdot}}}
    {\usebibmacro{name:delim}{#1}%
     \usebibmacro{name:hook}{#1}%
     \mkbibcompletenamefamilygiven{%
       \mkbibnamefamily{#1}\isdot
       \ifdefvoid{#4}
         {}
         {\bibnamedelimd\mkbibnamesuffix{#4}\isdot}%
       \ifboolexpe{%
         test {\ifdefvoid{#2}}
         and
         test {\ifdefvoid{#3}}}
         {}
         {\revsdnamepunct}%
       \ifdefvoid{#2}
         {}
         {\bibnamedelimd\mkbibnamegiven{#2}\isdot}%
       \ifdefvoid{#3}
         {}
         {\bibnamedelimd\mkbibnameprefix{#3}\isdot}}}}

\usepackage[T1]{fontenc}
\usepackage[british]{babel}
\usepackage{csquotes}
\usepackage{amssymb,amsthm,bm}
\usepackage{mathtools}
\usepackage{nicefrac}
\mathtoolsset{mathic}
\allowdisplaybreaks
\usepackage{courier,mathptmx,stmaryrd}
\usepackage{mathrsfs}
\DeclareMathAlphabet{\mathcal}{OMS}{cmsy}{m}{n}
\usepackage{thmtools,thm-restate}
\usepackage{enumitem}
\usepackage{xcolor}
\usepackage{longtable}
\usepackage{tabu}
\newlength{\labelind}
\newlength{\leftmarg}
\usepackage{marvosym}
\errorcontextlines\maxdimen
\usepackage{hyperref,url}
\hypersetup{  
bookmarksnumbered
}
\hypersetup{  
breaklinks=false,
pdfborderstyle={/S/U/W 0},  
citebordercolor=.235 .702 .443,
urlbordercolor=.255 .412 .882,
linkbordercolor=.804 .149 .19,
}
\hypersetup{ 
pdfauthor={Keano De Vos, Gert de Cooman, Alexander Erreygers and Jasper De Bock},
pdftitle={A decision-theoretic approach to dealing with uncertainty in quantum mechanics},
pdfkeywords={quantum mechanics, decision theory, uncertainty, desirability, lower prevision, prevision, credal set, density operator, utility}
}
\usepackage[all]{hypcap} 
\usepackage[capitalise]{cleveref}
\crefformat{equation}{Eq.~(#2#1#3)}
\Crefformat{equation}{Eq.~(#2#1#3)}
\crefrangeformat{equation}{Eqs.~(#3#1#4)--(#5#2#6)}
\crefformat{definition}{\text{Definition~}#2#1#3}
\crefformat{item}{#2#1#3}
\newtheoremstyle{ownTheorem}
{.5 em}
{1 em}
{\itshape}
{0pt}
{\bfseries}
{}
{\newline}
{}

\newtheoremstyle{ownpostulateTheorem}
  {.5 em}
  {1 em}
  {\itshape}
  {0pt}
  {\bfseries}
  {}
  {\newline}
  {\thmname{#1}\thmnumber{ #2} (\thmnote{#3}\thmnumber{#2})}

\theoremstyle{ownTheorem}
\newtheorem{theorem}{Theorem}

\newtheorem{lemma}[theorem]{Lemma}
\newtheorem{corollary}[theorem]{Corollary}
\newtheorem{proposition}[theorem]{Proposition}
\newtheorem*{DMsumm}{Decision-theoretic summary}
\theoremstyle{ownpostulateTheorem}
\newtheorem{QMpost}{Quantum mechanical postulate}
\crefformat{QMpost}{#2\upshape{QM#1}#3}
\crefrangeformat{QMpost}{#3\upshape{QM#1}#4--#5\upshape{QM#2}#6}
\crefmultiformat{QMpost}{#2\upshape{QM#1}#3}{ and~#2\upshape{QM#1}#3}{,~#2\upshape{QM#1}#3}{ and~#2\upshape{QM#1}#3}
\crefrangemultiformat{QMpost}{#3\upshape{QM#1}#4--#5\upshape{QM#2}#6}{ and~#3\upshape{QM#1}#4--#5\upshape{QM#2}#6}{,~#3\upshape{QM#1}#4--#5\upshape{QM#2}#6}{ and~#3\upshape{QM#1}#4--#5\upshape{QM#2}#6}
\newtheorem{DMbackass}{Decision-theoretic background assumption}
\crefformat{DMbackass}{#2\upshape{DTB#1}#3}
\crefrangeformat{DMbackass}{#3\upshape{DTB#1}#4--#5\upshape{DTB#2}#6}
\crefmultiformat{DMbackass}{#2\upshape{DTB#1}#3}{ and~#2\upshape{DTB#1}#3}{,~#2\upshape{DTB#1}#3}{ and~#2\upshape{DTB#1}#3}
\crefrangemultiformat{DMbackass}{#3\upshape{DTB#1}#4--#5\upshape{DTB#2}#6}{ and~#3\upshape{DTB#1}#4--#5\upshape{DTB#2}#6}{,~#3\upshape{DTB#1}#4--#5\upshape{DTB#2}#6}{ and~#3\upshape{DTB#1}#4--#5\upshape{DTB#2}#6}
\newtheorem{RFpost}{Reward function postulate}
\crefformat{RFpost}{#2\upshape{RF#1}#3}
\crefrangeformat{RFpost}{#3\upshape{RF#1}#4--#5\upshape{RF#2}#6}
\crefmultiformat{RFpost}{#2\upshape{RF#1}#3}{ and~#2\upshape{RF#1}#3}{,~#2\upshape{RF#1}#3}{ and~#2\upshape{RF#1}#3}
\crefrangemultiformat{RFpost}{#3\upshape{RF#1}#4--#5\upshape{RF#2}#6}{ and~#3\upshape{RF#1}#4--#5\upshape{RF#2}#6}{,~#3\upshape{RF#1}#4--#5\upshape{RF#2}#6}{ and~#3\upshape{RF#1}#4--#5\upshape{RF#2}#6}

\newtheoremstyle{ownDef}
  {.5 em}
  {1 em}
  {}
  {0pt}
  {\bfseries}
  {}
  {\newline}
  {}
\theoremstyle{ownDef}

\newtheoremstyle{ownEx}
  {.5 em}
  {1 em}
  {}
  {0pt}
  {\bfseries}
  {}
  {\newline}
  {}
\theoremstyle{ownEx}
\newtheorem{examplex}{Qubit running example: instalment}
\newenvironment{example}
  {\pushQED{\qed}\examplex}
  {\popQED\endexamplex}

\usepackage{textcomp}
\usepackage{xpatch}
\usepackage{siunitx}
\usepackage{nicefrac}
\makeatletter
\newcounter{proofpart}
\xpretocmd{\proof}{\setcounter{proofpart}{0}}{}{}
\newcommand{\proofpart}[1]{%
  \par
  \addvspace{\medskipamount}%
  \stepcounter{proofpart}%
  \noindent\emph{Part \theproofpart: #1}\par\nobreak\smallskip
  \@afterheading
}
\makeatother
\usepackage{eurosym}

\usepackage{tikz}
\usetikzlibrary{arrows.meta,positioning}

\DeclarePairedDelimiter{\group}{(}{)}
\DeclarePairedDelimiter{\sqgroup}{[}{]}
\DeclarePairedDelimiter{\set}{\{}{\}}

\DeclarePairedDelimiter{\norm}{\Vert}{\Vert}
\DeclarePairedDelimiter{\abs}{\vert}{\vert}
\DeclarePairedDelimiter{\modulus}{\vert}{\vert}

\DeclarePairedDelimiter\bra{\langle}{\rvert}
\DeclarePairedDelimiter\ket{\lvert}{\rangle}
\DeclarePairedDelimiterX\braket[2]{\langle}{\rangle}{#1\delimsize\vert#2}
\DeclarePairedDelimiterX\braketwithop[3]{\langle}{\rangle}{#1\delimsize\vert#2\delimsize\vert#3}
\DeclarePairedDelimiterXPP{\btrace}[2]{\tr_{#1}}{(}{)}{}{{#2}}
\DeclarePairedDelimiterXPP{\trace}[1]{\tr}{(}{)}{}{{#1}}
\DeclarePairedDelimiterXPP{\inprod}[2]{}{(}{)}{}{{#1,#2}}
\DeclareMathOperator{\spec}{spec}
\DeclareMathOperator{\range}{rng}
\DeclareMathOperator{\posi}{posi}
\DeclareMathOperator{\interior}{Int}
\DeclareMathOperator{\closure}{Cl}
\DeclareMathOperator{\spa}{span}

\DeclareMathOperator{\tr}{Tr}
\DeclareMathOperator{\expe}{E}

\newcommand{\hilbertspace}{\mathscr{X}}
\newcommand{\althilbertspace}{\mathscr{Y}}

\newcommand{\allhilbertspaces}{\mathbb{X}}
\newcommand{\statespace}{\bar{\mathscr{X}}}

\newcommand{\subspace}{\mathscr{V}}

\newcommand{\linprevs}[1][]{\mathbb{P}}

\newcommand{\supnorm}[2][]{\norm[#1]{#2}_\infty}
\newcommand{\frobnorm}[2][]{\norm[#1]{#2}_\mathrm{F}}
\newcommand{\opnorm}[2][]{\norm[#1]{#2}_{\mathrm{op}}}

\newcommand{\complexes}{\mathbb{C}}
\newcommand{\reals}{\mathbb{R}}
\newcommand{\rationals}{\mathbb{Q}}

\newcommand{\naturals}{\mathbb{N}}
\newcommand{\nonnegreals}{\reals_{\geq0}}
\newcommand{\posreals}{\reals_{>0}}

\newcommand{\posrationals}{\rationals_{>0}}

\newcommand{\linspanof}[1]{\spa(#1)}
\newcommand{\basis}{\mathscr{B}}
\newcommand{\densities}{\mathscr{R}}

\newcommand{\posmeasurements}{\measurements_{\alwaysbetterthan\zero}}
\newcommand{\strictlyposmeasurements}{\measurements_{\alwaysstrictlybetterthan\zero}}

\newcommand{\negmeasurements}{\measurements_{\alwaysworsethan\zero}}
\newcommand{\strictlynegmeasurements}{\measurements_{\alwaysstrictlyworsethan\zero}}
\newcommand{\possemidefmeasurements}{\measurements_{\possemidef\zero}}
\newcommand{\negsemidefmeasurements}{\measurements_{\negsemidef\zero}}

\newcommand{\desirset}[1][\,]{\mathscr{D}_{#1}}

\newcommand{\assessment}{\mathscr{A}}

\newcommand{\utilities}{\mathscr{W}}
\newcommand{\theutilities}{\mathscr{U}}

\newcommand{\ext}{\mathscr{C}}

\newcommand{\fket}[1][\,]{\smash{\ket{\phi_{#1}}}}
\newcommand{\indfket}[2]{\smash{\ket{\phi^{#1}_{#2}}}}
\newcommand{\fbra}[1][\,]{\smash{\bra{\phi_{#1}}}}
\newcommand{\fbraket}[1][\,]{\braket{\phi_{#1}}{\phi_{#1}}}

\newcommand{\gket}[1][\,]{\smash{\ket{\psi_{#1}}}}
\newcommand{\gbra}[1][\,]{\smash{\bra{\psi_{#1}}}}

\newcommand{\gbrakets}[2]{\braket{\psi_{#1}}{\psi_{#2}}}
\newcommand{\uket}{\ket{\Psi}}
\newcommand{\ubra}{\bra{\Psi}}

\newcommand{\basisket}[1][\,]{\smash{\ket{b_{#1}}}}
\newcommand{\basisbra}[1][\,]{\smash{\bra{b_{#1}}}}

\newcommand{\basisbrakets}[2]{\smash{\braket{b_{#1}}{b_{#2}}}}

\newcommand{\zeroket}{\ket{0}}
\newcommand{\oneket}{\ket{1}}
\newcommand{\zerobra}{\bra{0}}
\newcommand{\onebra}{\bra{1}}

\newcommand{\adjoint}[1]{{#1}^\dag}

\newcommand{\operator}[1]{\hat{#1}}
\newcommand{\measurement}[1]{\operator{#1}}
\newcommand{\indmeasurement}[2]{\operator{#1}_{#2}}
\newcommand{\measurements}{\mathscr{H}}

\newcommand{\projection}[1][\,]{\measurement{P}_{{#1}}}
\newcommand{\unitary}[1][\,]{\measurement{U}_{{#1}}}
\newcommand{\density}[1][\,]{\measurement{\rho}_{#1}}

\newcommand{\identity}{\measurement{I}}
\newcommand{\zero}{\measurement{0}}

\newcommand{\uval}{u}
\newcommand{\vval}{v}
\newcommand{\wval}{w}
\newcommand{\utility}[1]{\wval_{\measurement{#1}}}
\newcommand{\indutility}[2]{\wval_{\measurement{#1}_{#2}}}
\newcommand{\theutility}[1]{\uval_{\measurement{#1}}}

\newcommand{\utilitymaps}[1][\hilbertspace]{\mathbb{W}_{#1}}
\newcommand{\utilitymap}{\wval_{\bolleke}}
\newcommand{\theutilitymap}{\uval_{\bolleke}}
\newcommand{\assignation}{W}
\newcommand{\theassignation}{U}
\newcommand{\permutation}{\pi}

\newcommand{\pauli}{\measurement{\sigma}}
\newcommand{\paulix}{\pauli_x}
\newcommand{\pauliy}{\pauli_y}
\newcommand{\pauliz}{\pauli_z}

\newcommand{\commutator}[3][]{\sqgroup[#1]{#2,#3}}

\newcommand{\eigval}{\lambda}
\newcommand{\alteigval}{\mu}
\newcommand{\eigket}[1][\,]{\smash{\ket{a_{#1}}}}
\newcommand{\eigbra}[1][\,]{\smash{\bra{a_{#1}}}}
\newcommand{\alteigket}[1][\,]{\smash{\ket{b_{#1}}}}

\newcommand{\alteigbra}[1][\,]{\smash{\bra{b_{#1}}}}

\newcommand{\eigbraket}[1][\,]{\smash{\braket{a_{#1}}{a_{#1}}}}
\newcommand{\eigbrakets}[2]{\smash{\braket{a_{#1}}{a_{#2}}}}
\newcommand{\eigspace}[1][\,]{\smash{\mathscr{E}_{{#1}}}}
\newcommand{\alteigspace}[1][\,]{\smash{\mathscr{F}_{{#1}}}}

\newcommand{\condprob}[2]{p\group[\big]{#1\big\vert#2}}
\newcommand{\prob}[1]{p\group{#1}}

\newcommand{\expec}[2][\density]{\expe_{#1}\group{#2}}
\newcommand{\condexpec}[2][\gket]{\expe\group[\big]{#2\big\vert#1}}

\newcommand{\betterthan}{\mathrel{\rhd}}
\newcommand{\betterthanorequal}{\mathrel{\unrhd}}
\newcommand{\notbetterthan}{\not\betterthan}
\newcommand{\alwaysbetterthan}{\gneq}

\newcommand{\alwaysstrictlybetterthan}{>}
\newcommand{\alwaysworsethan}{\lneq}
\newcommand{\alwaysstrictlyworsethan}{<}

\newcommand{\possemidef}{\geq}
\newcommand{\negsemidef}{\leq}

\newcommand{\gambset}{{\mathscr{G}}}
\newcommand{\mapset}{{\mathscr{L}}}

\newcommand{\linprev}[1][\,]{\Lambda_{#1}}
\newcommand{\lowprev}[1][\,]{\underline{\Lambda}_{#1}}
\newcommand{\uppprev}[1][\,]{\overline{\Lambda}_{#1}}

\newcommand{\tensortimes}{\otimes}
\newcommand{\con}{\alpha}
\newcommand{\altcon}{\beta}
\newcommand{\im}{i}

\newcommand{\SDM}{set of desirable measurements}
\newcommand{\SDMs}{sets of desirable measurements}

\newcommand*{\instantiateas}{\rightsquigarrow}

\newcommand{\cset}[3][]{\set[#1]{#2\colon#3}}
\newcommand{\then}{\Rightarrow}

\newcommand{\ifandonlyif}{\Leftrightarrow}

\newcommand{\game}[2][A]{\smash{\mathrm{game}^{\,\measurement{#1}}_{\ket{#2}}}}
\newcommand{\act}[1]{\smash{\mathrm{act}_{\measurement{#1}}}}
\newcommand{\valuationsymbol}{\mathscr{V}}
\newcommand{\valuation}[1][\measurement{A}]{\valuationsymbol_{#1}}
\newcommand{\spectrum}[1]{\spec(\measurement{#1})}
\newcommand{\spectrumpure}[1]{\spec(#1)}
\newcommand{\utilitypure}[1]{\smash{\wval_{#1}}}
\newcommand{\theutilitypure}[1]{\smash{\uval_{#1}}}
\newcommand{\gambles}[1][\statespace]{\gambset({#1})}
\newcommand{\realmaps}[1][\statespace]{\mapset({#1})}
\newcommand{\bolleke}{\vcenter{\hbox{\scalebox{0.7}{\(\bullet\)}}}}
\newcommand{\dotket}{\ket{\bolleke}}
\newcommand{\dotbra}{\bra{\bolleke}}
\newcommand{\dotbraket}{\braket{\bolleke}{\bolleke}}
\newcommand{\gbraketwithop}[2][]{\braketwithop[#1]{\psi}{#2}{\psi}}
\newcommand{\fbraketwithop}[2][]{\braketwithop[#1]{\phi}{#2}{\phi}}
\newcommand{\solp}{\mathscr{M}}
\newcommand{\todensity}{\hat{R}}

\begin{document}
\title{A decision-theoretic approach to dealing with uncertainty in quantum mechanics}
\author{Keano De Vos \and Gert de Cooman \and Alexander Erreygers \and Jasper De Bock}
\address{Ghent University, Foundations Lab for imprecise probabilities, Technologiepark-Zwijnaarde 125, 9052 Zwijnaarde, Belgium}
\begin{abstract}
We provide a decision-theoretic framework for dealing with uncertainty in quantum mechanics.
This uncertainty is two-fold: on the one hand there may be uncertainty about the state the quantum system is in, and on the other hand, as is essential to quantum mechanical uncertainty, even if the quantum state is known, measurements may still produce an uncertain outcome.
In our framework, measurements therefore play the role of acts with an uncertain outcome and our postulates ensure that Born's rule is encapsulated in the utility functions associated with such acts.
This approach allows us to uncouple (precise) probability theory from quantum mechanics, in the sense that it leaves room for a more general, so-called imprecise probabilities approach.
We discuss the mathematical implications of our findings, which allow us to give a decision-theoretic foundation to recent seminal work by \citeauthor{benavoli2016:quantum}, and we compare our approach to earlier and different approaches by Deutsch and Wallace.
\end{abstract}
\maketitle

\section{Introduction and overview}\label{sec:introduction}
In dealing with a quantum system, there are various reasons why some subject, whom we'll call You, might be uncertain about the present state it's in: the state may have been prepared by performing a measurement, or You might be uncertain about the dynamics that brought it to its present state, or about the state it started out from, to name a few of them.

This uncertainty is typically represented, and reasoned with, using (quantum) probabilities, in the form of density operators.
This type of probabilistic model, pervasive in the literature, has a characteristic that it shares with the probabilistic models that are commonly used in more classical --- non-quantum --- contexts: it has a very high informational content.
That it does, is exemplified by the fact that when such models are used in a decision-theoretic (expected utility) framework, they leave no room for indecision: they allow You to choose between any two options and provide a definite answer to every yes-or-no question.
While this isn't a problem or a drawback in itself, it does call into question why such decisiveness is always justified in all contexts where uncertainty is present.
And this question is made all the more incisive by the following observation: while probabilities are often seen as constituting a superstructure built on a foundation of classical (propositional) logic, they don't share, in the way they're commonly represented, with this logic one of its most crucial properties: its ability to represent and deal with \emph{partial information} through its conservative \emph{deductive inference} mechanism.

This `problem' has been recognised for quite some time in classical (non-quantum) uncertainty contexts, and the research field of \emph{imprecise probabilities} (see, for instance, Refs.~\cite{walley1991,walley2000,augustin2013:itip,troffaes2013:lp}) has devoted quite some time and effort to allowing the mutually related aspects of indecision, imprecision and partial information to also play their important constitutive role in probability theory.
It was, arguably and to the best of our knowledge, \citeauthor{walley1991} who first drew attention to the conservative inference mechanism that is hidden behind classical (non-quantum) probability theory \cite{walley1991,walley2000} and one of us (\citeauthor{cooman2003a}) who first drew attention to the analogy between that conservative inference mechanism and deductive inference in classical propositional logic \cite{cooman2003a}.
In a nutshell, the so-called \emph{imprecise probability models} (such as there are: sets of probability measures, coherent lower and upper previsions, coherent sets of desirable gambles and coherent partial preference orderings) represent partial information states that have the same role as deductively closed and logically consistent sets of propositions do in classical propositional logic; and it's the precise models amongst them (such as there are: single probability measures, coherent previsions, maximal coherent sets of desirable gambles and coherent total preference orderings, respectively) that correspond to the maximal (so-called \emph{complete}) deductively closed and logically consistent sets of propositions, to which no further information can be added without producing inconsistencies.
On this way of looking at things, insisting that all uncertainty must be modelled by probabilities is very much like claiming that the only reasonable sets of propositions are the complete ones.
We'll have ample occasion to explain, and provide pointers to, the relevant technical imprecise probabilities literature dealing with such models as coherent sets of desirable gambles, coherent lower previsions and sets of probability measures further on in the paper, but as far as their connection with conservative inference is concerned, we want to take the opportunity here to also refer readers to Refs.~\cite{cooman2021:archimedean:choice,debock2018,decooman2024:things:logic,cooman2003a}, which are relevant explorations of this issue.

We believe that \citeauthor{benavoli2016:quantum} were the first to bring the imprecise probabilities framework to bear on dealing with uncertainty in quantum mechanics in a  finite-dimensional context.
In a number of ground-breaking papers \cite{benavoli2016:quantum,benavoli2019:computational}, they essentially associated specific types of desirable gambles with measurements, and showed how coherent sets of such desirable gambles can be connected to sets of density operators.

One of the things we want to achieve in this paper, is to provide their uncertainty models with a more solid decision-theoretic foundation and to show that the exact `Born rule'-like form they give to the desirable gamble that is associated with a measurement on a quantum-mechanical system can be derived from decision-theoretic principles.

We're well aware that we aren't the only ones --- and definitely also not the first --- to try and provide a decision-theoretic foundation for dealing with uncertainty in quantum mechanics.
To name only a few others, \citeauthor{deutsch1999:quantum:decisions} \cite{deutsch1999:quantum:decisions} and \citeauthor{wallace2003:defending:deutsch} \cite{wallace2003:defending:deutsch,wallace2007:improving:deutsch,wallace2009:born:arxiv} have tried to show that in an Everettian (many-worlds) approach to quantum mechanics, a few simple decision-theoretic principles are enough to derive the probabilistic quantum mechanical postulates from the non-probabilistic ones.
Their aim in this has been to justify the use of probabilities (Born's rule) in a quantum mechanical context, or to use Deutsch's phrase, to justify
\begin{quote}
[\dots] deriving a probability statement from a factual statement.
This could be called deriving a `tends to' from a `does'.
\end{quote}
Furthermore, \citeauthor{gleason} famously derived Born's rule from the non-probabilistic postulates of quantum mechanics, but his argument relies on the assumption that probabilities are already present in the form of probability measures on the projection operators \cite{gleason}.
Other work, under the umbrella of \emph{generalised probabilistic theories} \cite{hardy2001,hardy2002,barrett2007,Chiribella2011}, makes no assumptions about the mathematical structure of the state space or measurements, but derives these from a number of operationally motivated postulates.
However, all this work starts from the assumption that probabilities are fundamental and takes their existence as a given.

In this paper, we don't necessarily want to repeat or improve on these existing arguments and we most decidedly aren't looking for ways to justify using (precise) probabilities in a quantum context.
But we are looking for ways to deal with the existing uncertainty in quantum mechanics using a decision-theoretic toolbox.
In doing so, we also want to show that more general models than probabilities are reasonable and useful for dealing with epistemic uncertainty in quantum mechanics.
As we'll explain in much more detail further on in \cref{sec:our:approach}, we'll place the decision-theoretic argument in a de Finetti rather than a Savage framework, which will then allow us to deal more easily with partial preference, imprecision and indecision.
Indeed, rather than postulate Born's rule (and the concomitant existence of probabilities) directly, we propose (like Deutsch and Wallace) a conglomerate of decision-theoretic postulates that seem less invasive, and which allow us to derive Born's rule as a special case, while allowing (unlike Deutsch and Wallace) for more general (imprecise or partially specified probabilities) models in less restrictive cases.
This will allow us to put into perspective the --- to our mind unjustifiably strong --- uniquely central role that precise probabilities often seem to play in quantum mechanics.

How have we structured the discussion?
We begin in \cref{sec:quantum:mechanics} with a brief review of the basics of quantum mechanics, and in particular of its (non-probabilistic and probabilistic) postulates, which we'll have occasion to refer to and build on later.

In \cref{sec:our:approach}, we describe the decision-theoretic foundations for our approach to dealing with uncertainty in quantum mechanics, we formulate and motivate our (decision-theoretic) postulates, and we announce the important conclusions that can be drawn from them; we postpone proving these conclusions until \cref{sec:proofs}.
At the same time, we sketch the decision-theoretic background that will help readers place our subsequent discussion in its proper context and compare it in \cref{sec:deutsch,sec:wallace} to the approaches followed by Deutsch and Wallace.
In a nutshell, we show in this section that when measurements~\(\measurement{A}\) are interpreted as acts with an uncertain reward~\(\utility{A}(\uket)\), which depends on the (possibly uncertain) quantum state~\(\uket\) and which is expressed in units of some predetermined linear utility scale, a number of simple postulates fix the form of this reward, as they guarantee that \(\utility{A}(\uket)=\ubra\measurement{A}\uket\) --- we'll explain the notations in \cref{sec:quantum:mechanics}.

The decision-theoretic upshot of our results in \cref{sec:our:approach} (and their proofs in~\cref{sec:proofs}) is that uncertainty about the quantum state~\(\uket\) of a system under consideration can be described by a strict partial vector ordering on the uncertain rewards~\(\utility{A}\), or equivalently, on the measurements~\(\measurement{A}\).
We describe the mathematical consequences of using such uncertainty models, and alternative representations for them, in \cref{sec:desirability,sec:lower:upper:previsions,sec:coherent:previsions}.
This is where we show that the models introduced by \citeauthor{benavoli2016:quantum} \cite{benavoli2016:quantum,benavoli2019:computational} can be derived within our decision-theoretic framework.
We also provide proofs for our claims there, even though some of them can be found elsewhere in some form or other, mainly in the interest of making this discussion as self-contained as possible.

In~\cref{sec:desirability}, we show how these strict partial vector orderings can be represented by so-called \emph{coherent sets of desirable measurements}, and that such models provide a means to perform conservative inference based on partial desirability statements about various measurements.
In~\cref{sec:lower:upper:previsions} we go on to introduce a slightly simpler and less general type of model that has the advantage of being Archimedean, meaning that the preferences can be expressed using the real number scale; these are the coherent lower (and upper) prevision functionals, or equivalently, convex closed sets of density operators.
When lower and upper prevision functionals coincide, they turn into coherent previsions; equivalently, when the convex closed sets of density operators collapse to singletons, we recover the more classical case of quantum probability --- working with density operators --- as a special case, as discussed in \cref{sec:coherent:previsions}.
We'll also discover, incidentally, that the mathematical \smash{(almost-)}equivalence between working with sets of measurements on the one hand in \cref{sec:desirability} and working with sets of density operators on the other in \cref{sec:coherent:previsions}, allows us to recover the well-known duality between the Heisenberg and Schrödinger pictures in our decision-theoretic approach as well.
Throughout \cref{sec:our:approach,sec:desirability,sec:lower:upper:previsions,sec:coherent:previsions}, we illustrate the various models and arguments using a series of simple examples involving qubits.

\cref{sec:born} is of special importance, because we show there that our general sets of desirable measurements approach, based on the representation result in \cref{sec:our:approach}, collapses to Born's rule when the state~\(\uket\) is known.
This allows us to compare in some detail our present approach to the earlier ones by Deutsch (in \cref{sec:deutsch}) and Wallace (in \cref{sec:wallace}), and to point out the relevant differences in \cref{sec:our:approach:is different}.

Finally, we spend some time in \cref{sec:conclusion} on rehearsing the main themes in this paper and drawing conclusions from them.
We also point to useful and interesting ways that the uncertainty models we're about to discuss in detail, might be used advantageously in practical computational tasks, which also opens up paths for future research.

We've also added two appendices.
The first contains an overview of the many notations and mathematical concepts in the paper, which we believe could be helpful for readers who might otherwise lose their way.
The second appendix shows how our theory could be extended from the projective measurements that are Hermitian operators to POVMs (positive operator-valued measures), the more general types of measurements considered in quantum information theory and quantum statistical decision theory \cite{helstrom1976,holevo2001}, which all in their foundations take precise probabilities for granted.
Again, our main aim here is to broaden those foundations.

\section{The quantum-mechanical basics}\label{sec:quantum:mechanics}
To lay the foundations for the coming discussion and to sketch the necessary context, we start with a concise account of those foundational tenets of quantum mechanics that we'll need, and which we'll present here in the form of seven basic postulates.
For a more thorough account of the foundations and more details about what we leave unexplained, we refer to the basic textbooks by \citeauthor{cohen1977:quantum:1} \cite{cohen1977:quantum:1}  and \citeauthor{nielsen2010:quantum} \cite{nielsen2010:quantum}, where readers can also find a discussion of those aspects of linear algebra that are necessary for understanding the mathematical underpinnings of these postulates, and which we'll assume them to be familiar with.

We begin with those postulates that don't mention, or rely on, probabilistic notions.

\subsection{The non-probabilistic postulates}\label{sec:quantum:mechanics:non-probabilistic}
With every quantum system, we associate a complex \emph{Hilbert space}~\(\hilbertspace\), with an \emph{inner product}~\(\inprod{\bolleke_1}{\bolleke_2}\) and associated norm~\(\norm{\bolleke}\coloneqq\sqrt{\inprod{\bolleke}{\bolleke}}\).

We'll use Dirac's notation and terminology \cite{dirac1981:principles} throughout, so we'll call an element~\(\gket\) of~\(\hilbertspace\) a \emph{ket}.
The corresponding \emph{bra}, or dual ket, is then generically denoted by~\(\gbra\): it's the continuous\footnote{In the topology induced by the norm~\(\norm{\bolleke}\), the inner product and the norm are continuous.} linear functional on~\(\hilbertspace\) defined by \(\gbra(\fket)\coloneqq\inprod{\gket}{\fket}\) for all~\(\fket\in\hilbertspace\).
The inner product of two kets~\(\gket\) and \(\fket\), or in other words, the image of the ket~\(\fket\) under the bra~\(\gbra\), is then conveniently denoted by~\(\braket{\psi}{\phi}\).

For any subset~\(\assessment\) of~\(\hilbertspace\), its \emph{linear span}
\begin{equation*}
\linspanof{\assessment}
\coloneqq\cset[\bigg]{\sum_{k=1}^n\lambda_k\fket[k]}{n\in\naturals,\fket[1],\dots,\fket[n]\in\assessment,\lambda_1,\dots,\lambda_n\in\complexes}
\end{equation*}
is the smallest linear subspace of~\(\hilbertspace\) that includes~\(\assessment\).
Here and in what follows, we denote by \(\naturals\) the set of all natural numbers (without zero), and by \(\complexes\) the set of all complex numbers.

A \emph{state} is a normalised, or normal, ket~\(\gket\in\hilbertspace\), which means that \(\braket{\psi}{\psi}=1\).
We'll denote the set of possible states --- the \emph{state space} --- by~\(\statespace\coloneqq\cset{\gket\in\hilbertspace}{\braket{\psi}{\psi}=1}\).

\begin{QMpost}[QM]\label{post:qm:kets}
At any fixed time, the state of a physical system is represented by a normalised ket~\(\gket\), which is an element of the system's state space~\(\statespace\).\footnote{Strictly speaking, the state of a system is represented by a \emph{ray} of kets, typically characterised by one of its elements: a normalised one. This implies that if the normalised ket~\(\gket\) represents the state of a system, then so does \(e^{i\alpha}\gket\) for any \(\alpha\in\reals\); a normalised ket represents a system's state uniquely, up to a phase factor. We'll always implicitly disregard this phase factor when considering the state of a system.}
\end{QMpost}

To deal with certain aspects of quantum-mechanical systems, such as location, infinite-dimensional Hilbert spaces are essential.
But to keep the discussion as simple as possible, \emph{we'll restrict ourselves in this paper to the case of finite-dimensional Hilbert spaces}, which can for instance be used to model such aspects as the spin or (with some extra assumptions, such as ignoring the higher energy levels) the energy of a bounded electron.
Such finite-dimensional spaces are particularly useful in quantum computing and quantum cryptography \cite{nielsen2010:quantum}.

One way to interact with a quantum system is through measurements, which are represented by Hermitian operators.
The \emph{Hermitian conjugate} of a linear operator~\(\operator{O}\) on~\(\hilbertspace\) is the unique linear operator~\(\adjoint{\operator{O}}\) on~\(\hilbertspace\) such that \(\inprod{\adjoint{\operator{O}}\gket}{\fket}=\inprod{\gket}{\operator{O}\fket}\) for all~\(\gket,\fket\in\hilbertspace\).
A linear operator~\(\measurement{A}\) on~\(\hilbertspace\) is called \emph{Hermitian} if it's equal to its Hermitian conjugate~\(\adjoint{\measurement{A}}\), and we can then use Dirac's notation~\(\braketwithop{\psi}{\measurement{A}}{\phi}\coloneqq\inprod{\gket}{\measurement{A}\fket}=\inprod{\measurement{A}\gket}{\fket}\) without any possible confusion.
Observe, by the way, that it follows from the properties of an inner product that\footnote{We denote the complex conjugate of~\(a\in\complexes\) by \(a^*\).}
\begin{equation}\label{eq:hermitian:inner:product}
\braketwithop{\psi}{\measurement{A}}{\phi}
=\inprod{\gket}{\measurement{A}\fket}
=\inprod{\measurement{A}\fket}{\gket}^*
=\braketwithop{\phi}{\measurement{A}}{\psi}^*
\text{ for all~\(\gket,\fket\in\hilbertspace\)},
\end{equation}
and therefore also that \(\braketwithop{\phi}{\measurement{A}}{\phi}\) is real for all~\(\fket\in\hilbertspace\).
We denote the set of all Hermitian operators on~\(\hilbertspace\) by~\(\measurements\), or by~\(\measurements(\hilbertspace)\) should we want to avoid confusion about the Hilbert space~\(\hilbertspace\) they're defined on.

\begin{QMpost}[QM]\label{post:qm:operator}
Every measurable physical quantity of a quantum system with state space~\(\statespace\) can be described by a Hermitian operator~\(\measurement{A}\in\measurements(\hilbertspace)\).
\end{QMpost}

\noindent In fact, in finite-dimensional Hilbert spaces it's also often tacitly assumed that all Hermitian operators correspond to physical measurements, and we'll also make that assumption here.

When we measure a physical quantity associated with a system, this measurement produces a real number, which we'll call the \emph{outcome} of the measurement; we'll denote the set of all real numbers by~\(\reals\).
But it turns out that, typically, for a given measurement~\(\measurement{A}\), not all real numbers are possible outcomes for the measurement.
This is where the third postulate comes in, for the formulation of which we recall the notion of eigenvalues and eigenkets.

A complex number~\(\eigval\) is an \emph{eigenvalue} of a linear operator~\(\operator{O}\) on~\(\hilbertspace\) if there's some non-null ket~\(\eigket\in\hilbertspace\) such that \(\operator{O}\eigket=\eigval\eigket\).
Any ket~\(\eigket\) for which \(\operator{O}\eigket=\eigval\eigket\) is then called an \emph{eigenket} of~\(\operator{O}\) corresponding to the eigenvalue~\(\eigval\); we'll denote the set of all such eigenkets by~\(\eigspace[\eigval]\).
A normal eigenket is called an \emph{eigenstate}.
It's clear that any linear combination of eigenkets in~\(\eigspace[\eigval]\) is still an eigenket, so \(\eigspace[\eigval]\) is a linear subspace of~\(\hilbertspace\), called the \emph{eigenspace} of~\(\operator{O}\) corresponding to the eigenvalue~\(\eigval\).
The set of all the eigenvalues of the operator~\(\operator{O}\) is called its \emph{spectrum} and denoted by~\(\spectrum{O}\).

\begin{QMpost}[QM]\label{post:qm:eigenvalues}
The only possible outcomes of a measurement are the eigenvalues of the corresponding Hermitian operator~\(\measurement{A}\in\measurements\).
\end{QMpost}

The fact that all eigenvalues of any Hermitian operator are always real and that eigenstates corresponding to different eigenvalues are always orthogonal \cite[Exercise~2.17]{nielsen2010:quantum}, gives some intuition as to why the measurement operators are (taken to be) Hermitian.
Since the outcome of a physical measurement must be a real number, we must therefore consider measurement operators with real eigenvalues, and this is a property that Hermitian operators have.
In fact, the knowledge of both the eigenvalues and their corresponding eigenstates is enough to determine the corresponding Hermitian operator.
To explain this in the following proposition, we first introduce a convenient new notation.
For any ket~\(\gket\), we'll denote by \(\gket\gbra\) the linear operator on~\(\hilbertspace\) that maps any ket~\(\fket\) to the ket~\((\gket\gbra)(\fket)\coloneqq\gket\braket{\psi}{\phi}\), which is a scalar multiple of the ket~\(\gket\) --- in fact, it's the orthogonal projection of~\(\fket\) on the linear subspace spanned by \(\gket\).
It's a trivial exercise to show that \(\gket\gbra\) is Hermitian.
Indeed, any Hermitian operator is a real linear combination of specific operators of this type.

\begin{proposition}[\protect{\cite[Box~2.2]{nielsen2010:quantum}}]\label{prop:basis:and:eigenvalues}
Let \(\eigval_1\), \dots, \(\eigval_n\) be any real numbers and let \(\set{\eigket[1],\dots,\eigket[n]}\) be any orthonormal basis for the \(n\)-dimensional Hilbert space~\(\hilbertspace\).
Then \(\measurement{A}\coloneqq\sum_{k=1}^n\lambda_k\eigket[k]\eigbra[k]\) is the only operator that has \(\eigval_1\), \dots, \(\eigval_n\) as eigenvalues with respective eigenstates~\(\eigket[1]\), \dots, \(\eigket[n]\), and this operator is Hermitian.
Conversely, every Hermitian operator~\(\measurement{A}\in\measurements\) can be written as \(\measurement{A}=\sum_{k=1}^n\lambda_k\eigket[k]\eigbra[k]\), with~\(\set{\eigket[1],\dots,\eigket[n]}\) an orthogonal basis of eigenstates\footnote{Throughout, we'll use the term `orthogonal basis of eigenstates' rather than `orthonormal basis of eigenstates', because the normality is already implied by our using the term `eigen\emph{states}'.} for~\(\measurement{A}\) with corresponding respective real eigenvalues~\(\eigval_1,\dots,\eigval_n\in\reals\).
\end{proposition}

A \emph{projection operator}~\(\projection\) is a Hermitian operator such that \(\projection\projection=\projection\).
Due to its hermicity, a projection operator is an orthogonal projection: \(\inprod{\projection\fket}{\fket-\projection\fket}=\inprod{\projection\fket}{\fket}-\inprod{\projection\fket}{\projection\fket}=\inprod{\projection\fket}{\fket}-\inprod{\adjoint{\projection}\projection\fket}{\fket}=0\) for all~\(\fket\in\hilbertspace\).
Moreover, since for any eigenvalue~\(\lambda\in\spectrum{P}\) with corresponding eigenket~\(\fket\), \(\lambda\fket=\projection\fket=\projection^2\fket=\lambda^2\fket\), we see that \(\spectrum{P}\subseteq\set{0,1}\).

For any linear subspace~\(\eigspace\) of~\(\hilbertspace\), we'll denote by \(\projection[\eigspace]\) the unique projection operator~\(\projection\) whose range~\(\range(\projection)\coloneqq\cset{\projection\fket}{\fket\in\hilbertspace}\) is equal to~\(\eigspace\).
Given an \(m\)-dimensional linear subspace~\(\eigspace\) of~\(\hilbertspace\) and any orthonormal basis~\(\set{\basisket[1],\dots,\basisket[m]}\) for~\(\eigspace\), this projection operator can be written as \(\projection[\eigspace]=\sum_{k=1}^m\basisket[k]\basisbra[k]\); see \Cref{prop:basis:and:eigenvalues}.
If \(\eigspace[\eigval]\) is the eigenspace corresponding to the eigenvalue~\(\eigval\) of some Hermitian operator, then we'll also use the notation~\(\projection[\eigval]\) for the projection operator~\(\projection[{\eigspace[\eigval]}]\) with range~\(\eigspace[\eigval]\).

Projection operators allow us to represent any Hermitian operator in terms of its eigenvalues and eigenspaces.

\begin{corollary}[\protect{\cite[Box~2.2]{nielsen2010:quantum}}]\label{cor:basis:and:eigenvalues}
Consider \(m\leq n\) distinct real numbers~\(\eigval_1\), \dots, \(\eigval_m\) and let \(\eigspace[1]\), \dots, \(\eigspace[m]\) be orthogonal subspaces of an \(n\)-dimensional Hilbert space~\(\hilbertspace\) that span \(\hilbertspace\), then the Hermitian operator~\(\measurement{A}\coloneqq\sum_{k=1}^m\lambda_k\projection[{\eigspace[k]}]\) is the unique operator with eigenvalues~\(\eigval_1\), \dots, \(\eigval_m\) and corresponding eigenspaces~\(\eigspace[\eigval_1]=\eigspace[1]\), \dots, \(\eigspace[\eigval_m]=\eigspace[m]\).
Conversely, any Hermitian operator~\(\measurement{A}\in\measurements\) can be uniquely written as \(\measurement{A}=\sum_{k=1}^m\eigval_k\projection[{\eigspace[\eigval_k]}]\), with~\(\eigval_1\), \dots, \(\eigval_m\) its distinct real eigenvalues and with orthogonal corresponding eigenspaces~\(\eigspace[\eigval_1]\), \dots, \(\eigspace[\eigval_m]\) that span \(\hilbertspace\).
\end{corollary}

Postulate~\cref{post:qm:eigenvalues} fixes~\(\spectrum{A}\) as the set of all possible outcomes of a measurement~\(\measurement{A}\), but says nothing about which of these outcomes will actually be obtained.
The following postulate resolves this uncertainty in an important particular case.

\begin{QMpost}[QM]\label{post:qm:deterministic}
Performing a measurement with Hermitian operator~\(\measurement{A}\in\measurements\) on a quantum system in an eigenstate~\(\eigket\) of~\(\measurement{A}\) results in the corresponding eigenvalue~\(\lambda\) being measured with certainty.
\end{QMpost}

In quantum mechanics, measurements seldom leave the state untouched, as is made clear by the following postulate.
We won't rely on it for our main argument in \Cref{sec:our:approach} (and \Cref{sec:desirability,sec:lower:upper:previsions,sec:coherent:previsions,sec:proofs}).
But since it allows for ways to make sure that a quantum system is in a known state --- by performing a measurement with an outcome that corresponds to a one-dimensional eigenspace --- it will enable us in \Cref{sec:born} to consider a special case where we can compare our approach to earlier discussions by \citeauthor{deutsch1999:quantum:decisions} \cite{deutsch1999:quantum:decisions} and \citeauthor{wallace2003:defending:deutsch} \cite{ wallace2003:defending:deutsch,wallace2007:improving:deutsch,wallace2009:born:arxiv}.

\begin{QMpost}[QM]\label{post:qm:after:measurement}
If a measurement with Hermitian operator~\(\measurement{A}\in\measurements\) on a quantum system in a state~\(\gket\in\statespace\) results in the eigenvalue~\(\eigval\), then the state~\(\fket\) of the system immediately after the measurement is given by\footnote{This \(\fket\) is the normalised version of~\(\projection[\eigval]\gket\).}
\[
\fket
=\frac{\projection[\eigval]\gket}{\sqrt{\gbraketwithop{\projection[\eigval]}}}.
\]
\end{QMpost}

The last non-probabilistic postulate describes the dynamics of a quantum-mechanical system --- how its state changes over time.

\begin{QMpost}[QM]
The temporal evolution of the state~\(\ket{\psi(t)}\) as a function of the time \(t\), is governed by the Schrödinger equation:\footnote{Here, \(\hbar=\nicefrac{h}{2\pi}\), where \(h\) is Planck's constant.}
\[
\im\hbar\frac{\mathrm{d}}{\mathrm{d}t}\ket{\psi(t)}
=\operator{H}(t)\ket{\psi(t)},
\]
where \(\operator{H}(t)\) is the (possibly time-dependent) Hamiltonian of the system.
\end{QMpost}

\noindent As this paper is focused on the uncertainty related to measurements in quantum mechanics, the dynamical aspects will have little or no bearing on the present discussion, but we still include this postulate for the sake of completeness.\footnote{Of course, the dynamics will become relevant when we want to account for how the uncertainty changes with time, which we'll leave for future work.}

\subsection{Probabilities in quantum mechanics}\label{sec:quantum:mechanics:probabilistic}
As we've already mentioned, Postulate~\cref{post:qm:eigenvalues} tells us what the possible outcomes of a measurement are, but tells us nothing about which of these possible outcomes will actually be observed.
Postulate~\cref{post:qm:deterministic} resolves the remaining uncertainty whenever the quantum state belongs to an eigenspace of the measurement, but fails to do so in more general cases, where the quantum state is a \emph{superposition} --- a linear combination --- of kets belonging to different eigenspaces.
In classical accounts of quantum mechanics, this is where probabilities come into play.
The final postulate, also known as \emph{Born's rule}, attaches a specific probability to each possible outcome.

\begin{QMpost}[QM]\label{post:qm:born}
When a measurement with corresponding Hermitian operator~\(\measurement{A}\) on~\(\hilbertspace\) is executed on a system in state~\(\gket\in\statespace\), then the probability~\(\condprob{\eigval}{\gket}\) of measuring the eigenvalue~\(\eigval\) of~\(\measurement{A}\) is given by
\begin{equation}\label{eq:qm:born}
\condprob{\eigval}{\gket}\coloneqq\gbraketwithop{\projection[\eigval]}.
\end{equation}
\end{QMpost}
\noindent This is in similar spirit to Postulate~\cref{post:qm:deterministic}, where if the system resides in the eigenstate~\(\eigket\), the probability of observing the corresponding eigenvalue~\(\eigval\) is \(1\).

As an immediate consequence of this postulate, the linearity properties of the inner product and \cref{cor:basis:and:eigenvalues}, the \emph{expected outcome}~\(\condexpec{\measurement{A}}\) of the measurement~\(\measurement{A}\) is then given by
\begin{align}
\condexpec{\measurement{A}}
&=\smashoperator[r]{\sum_{\eigval\in\spectrum{A}}}\lambda\condprob{\eigval}{\gket}
=\smashoperator[r]{\sum_{\eigval\in\spectrum{A}}}\lambda\gbraketwithop{\projection[\eigval]}\notag
=\gbraketwithop[\bigg]{\smashoperator[r]{\sum_{\eigval\in\spectrum{A}}}\lambda\projection[\eigval]}\\
&=\gbraketwithop{\measurement{A}}.
\label{eq:qm:conditional:expectation}
\end{align}
Born's rule, in the guise of \cref{eq:qm:born,eq:qm:conditional:expectation},\footnote{\Cref{eq:qm:born} can be seen as a special case of \cref{eq:qm:conditional:expectation}, for the specific choice~\(\measurement{A}\instantiateas\projection[\eigval]\). We'll come across an even more general version, or formulation, further on in \cref{eq:qm:expectation}.} postulates the existence of such probabilities (and expectations), but it leaves open the question of how they should be interpreted.
In the many interpretations of quantum mechanics, they tend to acquire a different meaning and/or justification.

Some interpretations of quantum mechanics, such as the Copenhagen interpretation \cite{faye2019:copenhagen}, interpret the Born probabilities as physical and frequentist: they tend to insist that these probabilities, and their interpretation as limit frequencies, have been corroborated by numerous experiments, and that they're physical, objective properties attached to the system; on such views, the universe is indeterministic.

Other interpretations, such as QBism \cite{fuchs2014:introduction}, insist that these probabilities are epistemic and that they characterise some subject's knowledge, or the lack thereof, about the quantum system.
This view is compatible with an agnostic stance about whether the universe is deterministic.

The Everettian world view is deterministic, which might lead us to suspect that the Born probabilities there must necessarily have some epistemic interpretation.
Indeed, Deutsch \cite{deutsch1999:quantum:decisions} and later Wallace \cite{wallace2003:defending:deutsch,wallace2007:improving:deutsch,wallace2009:born:arxiv}, have advanced the claim that under certain (rationality) assumptions, and taking into account the non-probabilistic quantum mechanical postulates, any rational decision-maker must bet on the possible outcomes of a measurement using the betting rates supplied by Born's rule.
In this sense, on the Everettian view, their claim is that Born's rule follows from the non-probabilistic postulates and basic tenets of rational decision-making.
We'll come back to Deutsch and Wallace's arguments in later sections.

Leaving aside the different possible interpretations of the Born probabilities for now, we want to stress that the uncertainty about the outcome of a measurement when the system is in a given state~\(\gket\), isn't the only type of uncertainty that can be considered in quantum mechanics: it's also possible that we have imperfect knowledge about the actual state that the system resides in.
In the traditional formulation of quantum mechanics, such so-called \emph{epistemic uncertainty} about the system state is described by so-called \emph{mixed states}.
It's essential for the later discussion that we recall here in some detail how it can be represented mathematically and how this mathematical representation can be combined with Born's rule in the shape of \cref{eq:qm:conditional:expectation}.

A quantum system is said to be in a \emph{mixed state} if it's believed to be in one of a finite number of possible quantum states~\(\gket[1],\gket[2],\dots,\gket[m]\in\statespace\), with respective probabilities~\(p_1,p_2,\dots,p_m\), where~\(m\in\naturals\).
These probabilities~\(p_k\) therefore describe epistemic uncertainty about the actual state the system is in.

The outcome of a measurement on a system in a mixed state is uncertain because there is, on the one hand, this epistemic uncertainty about which of the possible states that system is in, and because in each of these possible states, the outcome of the measurement will be uncertain according to Postulate~\cref{post:qm:born}.
To distinguish between these two types of uncertainty, we'll also refer to the uncertainty about the outcome of a measurement on a system in a given pure state as \emph{quantum-mechanical}.\footnote{Depending on the interpretation of quantum mechanics, quantum-mechanical uncertainty may be seen as physical, or also as epistemic; see the discussion above.}

It turns out that, as we'll explain below, the epistemic uncertainty associated with a mixed state can be very conveniently described using its so-called density operator
\begin{equation}\label{eq:qm:density:operator}
\density\coloneqq\sum_{k=1}^mp_k\gket[k]\gbra[k].
\end{equation}
Generally speaking, a \emph{density operator}~\(\density\) is defined as a Hermitian operator such that \(\trace{\density}=1\)  and \(\density\geq\zero\).
We'll denote the convex set of all density operators on~\(\hilbertspace\) by~\(\densities\), a subset of~\(\measurements\).

To explain the notations we've just used, recall that the \emph{trace}~\(\trace{\measurement{A}}\) of a Hermitian operator~\(\measurement{A}\) is the sum of its eigenvalues.
Given any orthonormal basis~\(\set{\basisket[1],\dots,\basisket[n]}\) for the \(n\)-dimensional Hilbert space~\(\hilbertspace\), the trace can also be written as \(\trace{\measurement{A}}=\sum_{k=1}^n\basisbra[k]\measurement{A}\basisket[k]\).
Moreover, a Hermitian operator~\(\measurement{A}\in\measurements\) is called \emph{positive semidefinite}, which we'll write as \(\measurement{A}\geq\zero\), if \(\gbraketwithop{\measurement{A}}\geq0\) for all~\(\gket\in\hilbertspace\), or equivalently, if all its eigenvalues are non-negative.

It isn't hard to see that the operator~\(\density\) in~\cref{eq:qm:density:operator} is indeed a density operator.
As a matter of fact, any density operator can be rewritten as some such finite convex mixture of projection operators based on states.
But, such finite convex mixtures aren't unique in determining the density operator~\(\density\).

\begin{proposition}[\protect{\cite[Problem 10.1]{kitaev2002:classical}}]\label{prop:qm:density}
An operator~\(\density\) is a density operator if and only if there are states~\(\gket[1],\gket[2],\dots,\gket[m]\in\statespace\) and real numbers~\(\con_1,\con_2,\dots,\con_m\in[0,1]\) such that \(\sum_{k=1}^m\con_k=1\) and \(\density=\sum_{k=1}^m\con_k\gket[k]\gbra[k]\).
\end{proposition}
\noindent The so-called \emph{pure states} are then (degenerate) special cases: the density matrix~\(\density\) corresponding to a pure state~\(\gket\in\statespace\) is the projection operator~\(\density=\gket\gbra\) on that state.

Let's now argue briefly why the concept of a density operator~\(\density\) associated with a mixed state in~\cref{eq:qm:density:operator}, is so useful.

Consider any Hermitian measurement operator~\(\measurement{A}\in\measurements\) with eigenvalues \(\eigval_1,\dots,\eigval_m\) and corresponding respective eigenspaces~\(\eigspace[\eigval_1],\dots,\eigspace[\eigval_m]\) and corresponding respective projection operators~\(\projection[\eigval_1],\dots,\projection[\eigval_m]\), then we gather from \cref{cor:basis:and:eigenvalues} that \(\measurement{A}=\sum_{\ell=1}^m\eigval_\ell\projection[\eigval_\ell]\).
We now perform the measurement~\(\measurement{A}\) on the quantum system in its mixed state.

First, it follows from a straightforward application of the Law of Total Probability that the probability \(\prob{\eigval_\ell}\) of observing eigenvalue~\(\eigval_\ell\) is given by
\begin{align}
\prob{\eigval_\ell}
&=\sum_{k=1}^m\condprob{\eigval_\ell}{\gket[k]}p_k
=\sum_{k=1}^mp_k\gbra[k]\projection[\eigval_\ell]\gket[k]
=\sum_{k=1}^mp_k\trace[\big]{\gket[k]\gbra[k]\projection[\eigval_\ell]}\notag\\
&=\trace{\density\projection[\eigval_\ell]},
\label{eq:qm:probabilities}
\end{align}
where the second equality follows from Born's rule [\cref{post:qm:born}] in its simplest form~\eqref{eq:qm:born}; the third equality follows from the cyclic nature of the trace;\footnote{We'll have occasion to use this idea a number of times. It's based on the observation that \(\trace{\measurement{A}\measurement{B}}=\trace{\measurement{B}\measurement{A}}\); see \cite[Sec.~2.1.8]{nielsen2010:quantum} for more details.} and the final equality follows from the linearity of the trace.

As we now know the probabilities~\(\prob{\eigval_\ell}\) of the different possible outcomes~\(\eigval_\ell\) of the measurement~\(\measurement{A}\), it's a simple matter to calculate the corresponding \emph{expected outcome}:
\begin{align}
\expec[\density]{\measurement{A}}
\coloneqq&\sum_{\ell=1}^m\eigval_\ell\prob{\eigval_\ell}
=\sum_{\ell=1}^m\eigval_\ell\trace{\density\projection[\eigval_\ell]}
=\trace[\bigg]{\sum_{\ell=1}^m\eigval_\ell\density\projection[\eigval_\ell]}
=\trace[\bigg]{\density\sum_{\ell=1}^m\eigval_\ell\projection[\eigval_\ell]}\notag\\
=&\trace{\density\measurement{A}},
\label{eq:qm:expectation}
\end{align}
where the second equality follows from~\cref{eq:qm:probabilities} and the third equality from the linearity of the trace.
This expression for the expectation operator on measurements associated with a mixed state, or a density operator, is also often referred to as \emph{Born's rule}, and it will provide an anchor point for connecting our developments starting in \cref{sec:our:approach} with the more traditional approach in quantum mechanics.
This is the version of Born's rule in the presence of \emph{epistemic uncertainty} about the state of the system, as captured by the density operator~\(\density\).
The version of \cref{eq:qm:conditional:expectation} in the presence of a specific state~\(\gket\), can be recovered from it by letting~\(\density\instantiateas\gket\gbra\), as \(\trace{\gket\gbra\measurement{A}}=\trace{\gbra\measurement{A}\gket}=\gbraketwithop{\measurement{A}}\).\footnote{Here and further on, we use the symbol `\(\instantiateas\)' to denote `instantiates as' or `is replaced by'.}
So, \cref{eq:qm:expectation} combines the \emph{quantum-mechanical uncertainty} about the outcome of a measurement on a system in a pure state, as expressed in \cref{eq:qm:conditional:expectation}, with the \emph{epistemic uncertainty} about the state of the system, captured by the density operator~\(\density\).

\subsection{A few elementary results about Hermitian operators}\label{sec:hermitian:operators:basics}
We conclude this introductory section with a number of basic results about Hilbert spaces and Hermitian operators that will come in useful in later sections, and whose formulations (and proofs) we want to separate off in the interest of didactic clarity.

We begin with a number of elementary lemmas.
The first lemma provides a necessary and sufficient condition for kets to be normalised.

\begin{lemma}[\protect{\cite[p.~67]{nielsen2010:quantum}}]\label{lem:normalisation}
Let \(\eigket[1],\eigket[2],\dots,\eigket[m]\) be any collection of mutually orthogonal states in a Hilbert space~\(\hilbertspace\).
Then any~\(\gket\in\hilbertspace\) is normal if (and only if when \(m=n\)) there are \(\con_1,\dots,\con_m\in\complexes\) such that \(\sum_{k=1}^m\vert\con_k\vert^2=1\) and \(\gket=\sum_{k=1}^m\con_k\eigket[k]\).
\end{lemma}

\noindent Our second lemma shows that it's possible to write any state as an equal-amplitude superposition of some appropriate collection of orthogonal states.

\begin{lemma}\label{lem:uniform:expansion:in:basis}
If \(\gket\) is a state in an \(n\)-dimensional Hilbert space~\(\hilbertspace\), then for all~\(m\in\set{1,2,\dots,n}\), there are mutually orthogonal states~\(\gket[1]\), \(\gket[2]\), \dots, \(\gket[m]\) in~\(\statespace\) such that \(\gket=\frac{1}{\sqrt{m}}\sum_{k=1}^m\gket[k]\).
\end{lemma}

\begin{proof}
Let \(\set{\basisket[1],\basisket[2],\dots,\basisket[n]}\) be any orthonormal basis of~\(\hilbertspace\) with \(\basisket[m]\coloneqq\gket\), and let
\[
\gket[s]
\coloneqq\frac{1}{\sqrt{m}}\sum_{k=1}^me^{2\pi\im\frac{sk}{m}}\basisket[k],
\text{ for~\(s\in\set{1,\dots,m}\).}
\]
Then by \cref{lem:normalisation}, \(\gket[s]\) has norm~\(1\) for all~\(s\in\set{1,2,\dots,m}\).
For all~\(r,s\in\set{1,\dots,m}\), we use the conjugate symmetry of the inner product to find that
\begin{align*}
\gbrakets{r}{s}
&=\group[\bigg]{\frac{1}{\sqrt{m}}\sum_{k=1}^me^{-2\pi\im\frac{rk}{m}}\basisbra[k]}
\group[\bigg]{\frac{1}{\sqrt{m}}\sum_{\ell=1}^me^{2\pi\im\frac{s\ell}{m}}\basisket[\ell]}
=\frac{1}{m}\sum_{k=1}^m\sum_{\ell=1}^m
e^{-2\pi\im\frac{rk}{m}}e^{2\pi\im\frac{s\ell}{m}}\basisbrakets{k}{\ell}\\
&=\frac{1}{m}\sum_{k=1}^m\sum_{\ell=1}^m
e^{-2\pi\im\frac{rk}{m}}e^{2\pi\im\frac{s\ell}{m}}\delta_{k\ell}
=\frac{1}{m}\sum_{k=1}^me^{-2\pi\im\frac{rk}{m}}e^{2\pi\im\frac{sk}{m}}
=\frac{1}{m}\sum_{k=1}^me^{2\pi\im\frac{(s-r)k}{m}}
=\delta_{rs},
\end{align*}
where the second equality is explained by the linearity properties of the inner product, the third equality follows from the orthonormality of the basis~\(\set{\basisket[1],\basisket[2],\dots,\basisket[n]}\) and the last equality identifies the sum of the terms of a finite geometric sequence.
Therefore, \(\gket[r]\) and \(\gket[s]\) are orthogonal for~\(r\neq s\), and since they all have norm~\(1\), they're normal.

Finally, after identifying once again the sum of the terms of a finite geometric sequence, we find that
\begin{align*}
\frac{1}{\sqrt{m}}\sum_{s=1}^m\gket[s]
&=\frac{1}{m}\sum_{s=1}^m\sum_{k=1}^me^{2\pi\im\frac{sk}{m}}\basisket[k]
=\frac{1}{m}\sum_{k=1}^m\group[\bigg]{\sum_{s=1}^me^{2\pi\im\frac{sk}{m}}}\basisket[k]
=\frac{1}{m}\sum_{k=1}^mm\delta_{km}\basisket[k]
=\basisket[m]\\
&=\gket.
\qedhere
\end{align*}
\end{proof}

\noindent The final lemma provides a useful alternative expression for~\(\gbra\measurement{A}\gket\).

\begin{lemma}\label{lem:quadratic:form:in:basis}
Consider any collection~\(\eigket[1],\dots,\eigket[m]\) of mutually orthogonal eigenstates of a Hermitian operator~\(\measurement{A}\) on a Hilbert space~\(\hilbertspace\) with corresponding eigenvalues \(\eigval_1,\dots,\eigval_m\), and any linear combination~\(\gket\coloneqq\sum_{k=1}^m\con_k\eigket[k]\) of these eigenstates with~\(\con_1,\dots,\con_m\in\complexes\) and such that \(\sum_{k=1}^m\modulus{\con_k}^2=1\).
Then \(\gket\) is normal and \(\gbra\measurement{A}\gket=\sum_{k=1}^m\modulus{\con_k}^2\eigval_k\).
\end{lemma}

\begin{proof}
That \(\gket\) is normal follows from \cref{lem:normalisation}.
Observe that
\[
\gbra=\sum_{k=1}^m\con_k^*\eigbra[k]
\text{ and }
\measurement{A}\gket
=\measurement{A}\sum_{\ell=1}^m\con_\ell\eigket[\ell]
=\sum_{\ell=1}^m\con_\ell\measurement{A}\eigket[\ell]
=\sum_{\ell=1}^m\con_\ell\eigval_\ell\eigket[\ell],
\]
where the last equality holds because \(\eigket[\ell]\) is an eigenstate of~\(\measurement{A}\) with eigenvalue~\(\eigval_\ell\); so if we use the linearity properties of the inner product, we find that, indeed,
\[
\gbra\measurement{A}\gket
=\sum_{k=1}^m\sum_{\ell=1}^m\con_k^*\con_\ell\eigval_\ell\eigbrakets{k}{\ell}
=\sum_{k=1}^m\sum_{\ell=1}^m\con_k^*\con_\ell\eigval_\ell\delta_{k\ell}
=\sum_{k=1}^m\modulus{\con_k}^2\eigval_k,
\]
where the second equality follows from the orthonormality of the~\(\eigket[1],\dots,\eigket[m]\).
\end{proof}

It will also be useful to remember that the set~\(\measurements\) of all Hermitian operators on a finite-dimensional complex Hilbert space~\(\hilbertspace\) constitutes a finite-dimensional real linear space: if \(\hilbertspace\) has dimension~\(n\), then \(\measurements\) has dimension~\(n^2\).
The so-called \emph{Frobenius inner product} \cite[Exercise 2.39]{nielsen2010:quantum}, defined by
\[
\inprod{\measurement{A}}{\measurement{B}}\coloneqq\trace{\measurement{A}\measurement{B}}
\text{ for all~\(\measurement{A},\measurement{B}\in\measurements\)}
\]
turns \(\measurements\) into a real Hilbert space.\footnote{This is because any finite-dimensional linear space with an inner product is a Hilbert space.}
The so-called \emph{Frobenius norm}~\(\frobnorm{\bolleke}\) that's associated with this inner product, is defined by
\[
\frobnorm{\measurement{A}}
\coloneqq\sqrt{\inprod{\measurement{A}}{\measurement{A}}}
=\sqrt{\trace{\measurement{A}\measurement{A}}}
\text{ for all~\(\measurement{A}\in\measurements\)}.
\]
It turns \(\measurements\) into a normed linear space, with which we can associate a topology of open sets, a topological interior operator~\(\interior(\bolleke)\) and a topological closure operator~\(\closure(\bolleke)\).

\section{Our decision-theoretic approach}\label{sec:our:approach}
Let's now describe our approach to deriving a generalisation of Born's rule, in the more general version of \cref{eq:qm:expectation}, from a number of assumptions, inspired by \cref{post:qm:kets,post:qm:operator,post:qm:eigenvalues,post:qm:deterministic}, the relevant non-probabilistic postulates of quantum mechanics.

\subsection{The set-up}\label{sec:setup}
First, we consider a subject, whom we'll call You.
You are concerned with the value that the state of a quantum system under consideration takes in the set of all its possible values, which in principle, by \cref{post:qm:kets}, is given by the set~\(\statespace\) of all states, or normalised kets.
\emph{You are uncertain about this state}: You may have some knowledge about the actual value that this state takes in the set~\(\statespace\), but this knowledge may not be sufficient for You to determine this actual value with certainty.
To honour a time-tested tradition, we'll denote this possibly unknown state by a \emph{capital letter}~\(\uket\).
On the other hand, lower case letters will be used to denote, amongst other things, the states~\(\gket\in\statespace\) that are the candidate values for~\(\uket\).

In a second step, we recall that we can get information about the unknown state~\(\uket\in\statespace\) through performing measurements, described by Hermitian operators~\(\measurement{A}\in\measurements\), as postulated by \cref{post:qm:operator}.
If You were to perform a measurement~\(\measurement{A}\in\measurements\) on the system, its outcome would be uncertain too, for two reasons: firstly because You are (or may be) uncertain about the state the system is in, and secondly because even in a perfectly known state, the Postulate~\cref{post:qm:eigenvalues} typically leaves unspecified which of the eigenvalues~\(\eigval\in\spectrum{A}\) of the Hermitian operator~\(\measurement{A}\) the measurement would actually yield.

We'll consider that after You've performed the measurement~\(\measurement{A}\), what You'll get as a pay-off will be the actual outcome~\(\eigval\in\spectrum{A}\) of the measurement, expressed in units of some linear utility scale --- also called \emph{utiles}.
Your entertaining beliefs about which value~\(\uket\) assumes in~\(\statespace\) might cause You to prefer performing some measurements over others, because You believe they will lead to better pay-offs.

This is a simple idea, but important for what is to come, so let's look at a simple example to better explain where it comes from.
It's the first in a series of instalments of a running example that we'll use to illustrate various ideas throughout the text.

\begin{example}\label{example:preference:between:measurements}
Consider a qubit \cite{nielsen2010:quantum}, which is a quantum system with a two-dimensional Hilbert space, spanned by two orthogonal basis states~\(\zeroket\) and~\(\oneket\), so
\[
\hilbertspace
=\cset{a\zeroket+b\oneket}{a,b\in\complexes},
\]
and the set of all possible states is then, due to \Cref{lem:normalisation},
\[
\statespace
=\cset{a\zeroket+b\oneket}{a,b\in\complexes\text{ and }\modulus{a}^2+\modulus{b}^2=1}.
\]
We consider the unknown ket~\(\uket\coloneqq A\zeroket+B\oneket\), where~\(A\) and~\(B\) are uncertain variables in~\(\complexes\) such that \(\modulus{A}^2+\modulus{B}^2=1\).
For any given~\(\alpha,\beta\in\reals\), we consider the Hermitian operators~\(\indmeasurement{C}{\alpha},\indmeasurement{D}{\beta}\in\measurements\) defined by
\[
\indmeasurement{C}{\alpha}\coloneqq\alpha\zeroket\zerobra-\oneket\onebra
\text{ and }
\indmeasurement{D}{\beta}\coloneqq-\zeroket\zerobra+\beta\oneket\onebra,
\]
so \(\indmeasurement{C}{\alpha}\) has eigenstates~\(\zeroket\) and~\(\oneket\) with respective eigenvalues~\(\alpha\) and~\(-1\), and \(\indmeasurement{D}{\beta}\) has eigenstates~\(\zeroket\) and~\(\oneket\) with respective eigenvalues~\(-1\) and~\(\beta\) [see \Cref{prop:basis:and:eigenvalues}].

Now, suppose that You strongly believe that the qubit is in the \(\zeroket\) state, so in other words, that \(\modulus{B}=0\) and~\(\modulus{A}=1\).
Because You accept \cref{post:qm:deterministic}, this implies a strong belief, on Your part, that the outcome of the measurement~\(\indmeasurement{C}{\alpha}\) will be~\(\alpha\), and correspondingly, that the outcome of the measurement~\(\indmeasurement{D}{\beta}\) will be~\(-1\).
You'll therefore prefer performing the measurement~\(\indmeasurement{C}{\alpha}\) to performing the measurement~\(\indmeasurement{D}{\beta}\) as long as \(\alpha\) is sufficiently larger than~\(-1\).
\end{example}

It's this idea, namely that \emph{Your having beliefs about what value the (possibly) unknown state~\(\uket\) assumes in the possibility space~\(\statespace\) may lead You to having preferences between performing different measurements}, on which we'll build our theory.
It brings us squarely into a traditional context of decision-making under uncertainty \cite{walley1991,finetti19745,aumann1962,aumann1964,anscombe1963,nau2006,zaffalon2017:incomplete:preferences,cooman2021:archimedean:choice}, for which we're about to spell out the more important details, in the form of a number of background assumptions and postulates.

\subsection{Decision-theoretic background}\label{sec:decision:theoretic:background}
There are quite a number of ways in which so-called rational decision-making can be approached.
Many, if not most, of them start with the basic set-up where, on the one hand, there are \emph{acts}~\(a\) that You're invited to choose --- express a preference --- between, and where, on the other hand, the so-called \emph{consequence}~\(c\) of choosing an act~\(a\) may depend on what is often referred to as the \emph{state of the world}~\(\omega\), which is considered to be unknown, or uncertain.

In the simplest case, where there are only a finite number~\(n\) of acts \(a_k\), \(k=1,\dots,n\) and a finite number~\(m\) of possible states of the world \(\omega_\ell\), \(\ell=1,\dots, m\), the decision problem can be summarised in an act-state table, where each entry \(c_{k\ell}\) is the consequence of choosing act~\(a_k\) when the actual state of the world turns out to be~\(\omega_\ell\):
\begin{center}
\begin{tabular}{c|cccc}
& \(\omega_1\) & \(\omega_2\) & \(\cdots\) & \(\omega_m\)\\
\hline
\(a_1\) & \(c_{11}\) & \(c_{12}\) & \(\cdots\) & \(c_{1m}\)\\
\(a_2\) & \(c_{21}\) & \(c_{22}\) & \(\cdots\) & \(c_{2m}\)\\
\vdots & \vdots & \vdots & \(\cdots\) & \vdots \\
\(a_n\) & \(c_{n1}\) & \(c_{n2}\) & \(\cdots\) & \(c_{nm}\)
\end{tabular}
\end{center}

If we denote the set of all consequences by~\(C\), then the acts~\(a\) can be seen as, or identified with, maps from the set of all states~\(\Omega\) to~\(C\); in the simple example above, we then have that \(a_k(\omega_\ell)=c_{k\ell}\).
Savage \cite{savage1972} considers specific preference relations~\(\betterthanorequal\) on the (set of all) acts~\(A\), and his approach consists in providing axioms for the sets of acts~\(A\) and consequences~\(C\) and for the preference relation~\(\betterthanorequal\) on~\(A\) that guarantee that there's some probability on the state space~\(\Omega\) with corresponding expectation operator~\(E\) and some so-called \emph{utility function}~\(U\colon C\to\reals\) such that
\begin{equation}\label{eq:decision:via:expected:utility}
a\betterthanorequal b
\ifandonlyif E(U\circ a)\geq E(U\circ b)
\text{ for all~\(a,b\in A\).}
\end{equation}
A number of his axioms have the effect of ensuring that the sets of acts and consequences are \emph{sufficiently rich} and closed under convex mixtures.
We see that, in other words, probabilities, utilities and expectations on this way of thinking are tools that can be used to conveniently represent Your preferences between acts; and they can be constructed from the acts, consequences and preferences.

One specific aspect of Savage's approach stands out in the light of what we want to come to next.
He only allows for a \emph{total}, or \emph{linear}, \emph{ordering} of acts: it must be that \(a\betterthanorequal b\) or \(b\betterthanorequal a\), for all~\(a,b\in A\).
This allows him to define a \emph{strict preference}~\(\betterthan\) by letting \(a\betterthan b\ifandonlyif b\not\betterthanorequal a\), for all~\(a,b\in A\).
In other words, if \(a\not\betterthan b\) and \(b\not\betterthan a\), then it must be that both \(b\betterthanorequal a\) and \(a\betterthanorequal b\), which means that You have no option but to consider \(a\) and \(b\) equivalent: there's no room for incomparability or indecision in Savage's set-up.

We find Savage's blanket totality requirement too strong, generally speaking, for a number of reasons.
The first is that, as Savage himself indicates \cite[Sec.~2.7]{savage1972}, there seems to be no {\itshape a priori} reason not to allow for incomparability or indecision when trying to model Your rational preferences: indecision can be perfectly rational, even though allowing for it tends to complicate things.
A reluctance to complicate matters seems to have been Savage's prime reason for not pursuing this idea of allowing for indecision:
\begin{quote}
There is some temptation to explore the possibilities of analysing preference amongst acts as a {\bfseries partial ordering}, that is, in effect to replace part 1 of the definition of simple ordering by the very weak proposition~\(\mathrm{f}\leq\mathrm{f}\), admitting that some pairs of acts are incomparable.
This would seem to give expression to introspective sensations of indecision or vacillation, which we may be reluctant to identify with indifference.
My own conjecture is that it would prove a blind alley losing much in power and advancing little, if at all, in realism; but only an enthusiastic exploration could shed real light on the question.
\end{quote}
The second reason has a more positive flavour: not allowing for incomparability or indecision tends to hide or ignore important \emph{inferential} aspects of specifying preferences.
Indeed, there's a very strong analogy with propositional logic, where it's generally accepted to be of crucial importance to allow for sets of propositions that are deductively closed but \emph{not complete} in the sense that they leave the truth states of some propositions undecided; see for instance the detailed discussion in Ref.~\cite{cooman2003a}.

It's interesting to mention already here that Wallace \cite{wallace2003:defending:deutsch,wallace2007:improving:deutsch}, as we'll discuss in \cref{sec:wallace}, uses Savage's approach to refine Deutsch's argumentation \cite{deutsch1999:quantum:decisions} (discussed in \cref{sec:deutsch}) in a multiverse context; it's therefore not very surprising that his argument leads to a model for the uncertainty in quantum mechanics that describes Your preferences by a probability--utility pair, as alluded to above.

Here, we want to follow a different route, for the reasons already mentioned, and not impose the totality of preference orderings as a foundational part of our set-up.
If we do assume totality of preference relations in certain places and instances, we will be clear about it, and mention it explicitly.
In all other cases and instances, we'll assume partially ordered preferences to be the fallback models.
To see where this different route could lead, we have a brief look at some of the relevant literature.

\citeauthor{anscombe1963} \cite{anscombe1963} and \citeauthor{aumann1962} \cite{aumann1962,aumann1964}, as well as a few authors in their wake, amongst whom \citeauthor{seidenfeld1995} \cite{seidenfeld1995} and \citeauthor{nau2006} \cite{nau2006}, have followed a different approach than Savage to modelling decision-making, by (essentially) considering as the set of consequences~\(C\) the simplex of all probability mass functions --- so-called \emph{lotteries} --- on some set of rewards~\(R\).
Taking the set of acts~\(A\) to be all maps from the set of states of the world~\(\Omega\) to these lotteries makes sure that \(A\) is sufficiently rich: it contains all constant maps --- the lotteries --- and is closed under convex mixtures.
The state-dependent lotteries that constitute the acts in~\(A\) are also called \emph{horse lotteries}.
Two differences with Savage's approach stand out: (i) in most of the above-mentioned papers \cite{aumann1962,aumann1964,seidenfeld1995,nau2006}, the preference relations~\(\betterthanorequal\) and~\(\betterthan\) needn't reflect totality, and therefore still allow for incomparability between acts; and (ii) in some of them \cite{seidenfeld1995}, the authors also let go of conditions of Archimedeanity that allow the preferences to be represented by (sets of) real-valued utilities.

Bruno de Finetti \cite{finetti1937,finetti1970}, on the other hand, followed a simpler --- and therefore less general --- route, where contrary to the above-mentioned approaches, the utility function~\(U\) isn't derived from Your preferences, but is instead assumed to exist extraneously to the decision problem.
This then essentially implies that the consequences~\(c\in C\) in the decision problem can now themselves be identified with utilities, that is, assumed to be expressed in units of some predetermined linear utility scale.
These could, for instance, be tickets in a lottery for a single desirable prize, and in this sense, this simpler take on preference modelling can also be viewed as a special case of the above-mentioned horse lottery approach, where the doubleton set of rewards consists of winning and not winning the prize; see also Refs.~\cite{cooman2021:archimedean:choice,nau2006,zaffalon2017:incomplete:preferences}.
The acts can now be seen as \emph{uncertain rewards}: they're maps~\(a\colon\Omega\to\reals\) that associate with each state of the world~\(\omega\) the reward~\(a(\omega)\) for taking action~\(a\), expressed in utiles.
The preference relations~\(\betterthanorequal\) and~\(\betterthan\) now express preferences between such uncertain rewards.
We'll see further on in \cref{sec:deutsch} that Deutsch's take \cite{deutsch1999:quantum:decisions} on framing quantum uncertainty in a decision-theoretic context can be seen as inspired by de Finetti's approach.

On de Finetti's approach, Your preference ordering is still ideally total, as well as Archimedean, but later approaches \cite{walley1991,walley2000,seidenfeld1995} let go of these requirements, to allow for incomparability or indecision, and to incorporate preferences that aren't necessarily representable by the ordering on the real numbers.
Doing so has led to a rich literature on so-called \emph{sets of desirable} (or favourable) \emph{gambles}; see for instance Refs.~\cite{couso2011:desirable,cooman2010,moral2005b,quaeghebeur2012:itip,quaeghebeur2015:statement,debock2015:thesis,zaffalon2017:incomplete:preferences}.
It's in this generalisation of de Finetti's approach that we intend to develop our argumentation: on the one hand, because it allows for dealing with partial preferences, and on the other, because it's technically the least complicated.
We'll leave the more general but technically more involved approaches for future research.

It follows from the discussion near the beginning of this section that each measurement~\(\measurement{A}\in\measurements\) can be considered as an `act'
\begin{equation*}
\act{A}
\coloneqq\text{ ``perform the measurement~\(\measurement{A}\) on the quantum system''},
\end{equation*}
and these acts are in a one-to-one correspondence with the elements of the (real) linear space~\(\measurements\) of all measurements.
If we consider the states of the world~\(\omega\) to be the elements~\(\fket\) of the state space~\(\statespace\), all that still needs to be specified to complete the basic description of the decision problem, is the consequences, or in other words, the utilities (expressed in utiles) that correspond to any~\(\act{A}\) in each state of the world~\(\fket\).
This brings us to the first background assumption.

\begin{DMbackass}[DTB]\label{backass:dm:utility:function}
With every measurement~\(\measurement{A}\in\measurements\) on the Hilbert space~\(\hilbertspace\), You associate a \emph{reward function}, which is a map \(\utility{A}\colon\statespace\to\reals\), such that the real number~\(\utility{A}(\fket)\) is the reward, expressed in utiles, for~\(\act{A}\) when the quantum system under consideration system is in state~\(\fket\).
\end{DMbackass}
\noindent In other words, the uncertain real number~\(\utility{A}(\uket)\) is the uncertain reward associated with~\(\act{A}\); this reward is typically uncertain because the system state~\(\uket\) is unknown.
The reward functions~\(\utility{A}\) constitute a subset
\begin{equation*}
\utilities\coloneqq\cset{\utility{A}}{\measurement{A}\in\measurements}
\end{equation*}
of the linear space~\(\realmaps\) of all real-valued maps on~\(\statespace\).
Since we can associate an uncertain reward~\(\wval(\uket)\) with every reward function~\(\wval\in\utilities\), we'll agree to call these reward functions~\(\wval\) \emph{uncertain rewards} as well.
We stress that, in a general decision-making context, the reward functions are Yours to choose or determine, but we'll argue further on that in the specific context of quantum mechanics, a number of simple postulates make sure that this is no longer the case and that You're left with no choice about what these reward functions look like.

This first background assumption fits nicely within de Finetti's approach, but its consequences are stronger than they might seem at first encounter.
For a start, if You know that the quantum system is in a state~\(\fket\) --- so if You know that \(\uket=\fket\) --- then the reward~\(\utility{A}(\uket)\) that You get from performing any measurement~\(\measurement{A}\) is no longer uncertain, because You know that it must be equal to the real number~\(\utility{A}(\fket)\).
And since for any other measurement~\(\measurement{B}\) the reward~\(\utility{B}(\uket)\) is also known to You to be equal to the real number~\(\utility{B}(\fket)\), comparing the measurements~\(\measurement{A}\) and~\(\measurement{B}\) on the basis of their (now certain) rewards will become equivalent to comparing real numbers: the higher its certain real reward is, the better You'll prefer a measurement.
In other words, if You know that \(\uket=\fket\), then Your preferences between acts/measurements will necessariy be \emph{total} in the sense that You'll strictly prefer measurement~\(\measurement{A}\) to measurement~\(\measurement{B}\) if and only if \(\utility{A}(\fket)>\utility{B}(\fket)\), and You'll be indifferent between \(\measurement{A}\) and~\(\measurement{B}\) if and only if \(\utility{A}(\fket)=\utility{B}(\fket)\).
In this sense, the \emph{starting point} for our approach is similar to Deutsch's \cite{deutsch1999:quantum:decisions} and Wallace's \cite{wallace2003:defending:deutsch,wallace2007:improving:deutsch}, as we'll see in \cref{sec:deutsch,sec:wallace}, respectively: \emph{acts/measurements are linearly ordered when You know what state the system is in}.

Moreover, if You know that \(\uket=\fket\), then \(\utility{A}(\fket)\) is a one-number summary of all the possible outcomes~\(\eigval\in\spectrum{A}\) that the measurement~\(\measurement{A}\) may yield for a system in that state~\(\fket\).
\cref{backass:dm:utility:function} essentially assumes that such a one-number summary is possible, and implies that it's the only thing that matters for Your decisions when You know the state of the system.
What it doesn't do, however, is to specify, or impose restrictions on, what these reward functions may look like.
For that, we'll use specific postulates about the reward functions in the next section.

In general, when You don't know with certainty what value \(\uket\) assumes, we'll still assume that You can order the acts \(\act{A}\), but we'll not require, as Savage would have us do, that this ordering should be total.

\begin{DMbackass}[DTB]\label{backass:dm:preference:relation}
Your beliefs about the value of~\(\uket\) in~\(\statespace\) lead You to strictly prefer some uncertain rewards in~\(\utilities\) to others, which leads to a strict preference relation~\(\betterthan\) on the set of uncertain rewards~\(\utilities\) and therefore also on the set of acts, or measurements,~\(\measurements\).
\end{DMbackass}
\noindent By a \emph{strict preference} relation, we mean an irreflexive and transitive binary relation (a so-called \emph{strict partial order}).\footnote{When \(\utilities\) is a vector space, as we'll argue that it is further on, we'll assume that it's also a vector ordering.}
We see this requirement as fairly weak, and in some sense minimal, as there are only two ways it can be violated: first, if there's some \(\utility{A}\in\utilities\) that's strictly preferred to itself, which seems to run counter to the preference being called `strict'; and second, if there are \(\utility{A},\utility{B},\utility{C}\in\utilities\) such that \(\utility{A}\betterthan\utility{B}\) and \(\utility{B}\betterthan\utility{C}\) but not \(\utility{A}\betterthan{\utility{C}}\), which is a violation of transitivity.\footnote{That we're only looking at \emph{strict} preferences here, implies that we'll remain silent about what it means for two uncertain rewards to be \emph{equivalent} to each other, or for You to be \emph{indifferent} between them. As we'll see, restricting ourselves to strict preferences is already enough to get to a notion of quantum probabilities. Also dealing with indifference would lead to a richer framework, but for reasons of parsimony, we'll leave that for future research. See for instance Ref.~\cite{quaeghebeur2015:statement} for ways of going beyond strict preferences.}
We repeat that there's no totality assumption: it's \emph{not assumed} that You think incomparable measurements --- measurements~\(\measurement{A}\) and \(\measurement{B}\) for which neither \(\utility{A}(\uket)>\utility{B}(\uket)\) nor \(\utility{A}(\uket)<\utility{B}(\uket)\) holds --- are necessarily equivalent, so that You're indifferent between them.
We'll dive into the many and rich details of such belief representations in \cref{sec:desirability,sec:lower:upper:previsions,sec:coherent:previsions} further on.

In summary, the background assumptions \cref{backass:dm:utility:function,backass:dm:preference:relation} allow You to represent Your beliefs about where the state~\(\uket\) is by means of a strict partial (vector) ordering on the reward functions, but the ordering is required to be total --- or linear --- and in particular determined by linear ordering of the real numbers~\(\utilitypure{\bolleke}(\fket)\) only when You know with certainty that \(\uket=\fket\).

\subsection{The reward function postulates}\label{sec:decision:theoretic:postulates}
The assumptions \cref{backass:dm:utility:function,backass:dm:preference:relation} fix the context for our argument.
Against this backdrop, we'll now formulate four decision-theoretic postulates; more precisely, these are postulates that deal with the reward function aspect of our decision-making framework.
As we're about to argue, they're inspired by, and in a sense based on, the non-probabilistic postulates of quantum mechanics.
We'll prove in \cref{sec:proofs} that they determine what the uncertain rewards~\(\utility{A}\) look like: You don't get to play any role in fixing their shape, the postulates~\cref{post:dm:eigenket,post:dm:different:eigenspaces,post:dm:additivity,post:dm:continuity} we're about to introduce, will do that for You.

Every one of these four postulates is in itself intended to capture a simple idea, and we'll devote some attention to trying to point out what the four relevant ideas are.
We intend them to capture what is `essential' about the decision problem in quantum mechanics to allow us to recover Born's rule as a special case.
We emphatically don't want to claim that these central ideas can't be clarified further, or stated more succinctly or elegantly.
Nor do we necessarily believe that there are no simpler, weaker or more parsimonious sets of postulates that may lead to the same conclusions.
\emph{Our aim here is, simply stated, to propose a collection of postulates that we feel are convincing and natural enough, and which allow us to fix the reward functions.}

Let's begin with the simplest postulate, which fixes the utility gauge of the uncertain rewards.
It deals directly with the case governed by \cref{post:qm:deterministic}, where the system resides in an eigenstate~\(\eigket\) of the measurement operator~\(\measurement{A}\in\measurements\).
Since the outcome of the measurement is then necessarily the corresponding eigenvalue~\(\eigval\), we want this to also be the utility You obtain from performing the measurement in that state: the reward received should equal the outcome of the measurement in case of certainty.
This can be seen as a convention --- as are essentially all ways of fixing a gauge --- and it seems the simplest one that allows us to connect the reward --- the utility --- received from performing a measurement to its outcome.

\begin{RFpost}[RF]\label{post:dm:eigenket}
Let \(\eigspace[\eigval]\) be the eigenspace corresponding to an eigenvalue~\(\eigval\) of a Hermitian operator~\(\measurement{A}\) on a Hilbert space~\(\hilbertspace\).
Then necessarily \(\utility{A}(\eigket)=\eigval\) for all~\(\eigket\in\eigspace[\eigval]\) with \(\eigbraket=1\).
\end{RFpost}

The second postulate is a fairly strong one.
We'll first give a general formulation and then spend some effort in justifying it, through a number of examples and through an alternative formulation that perhaps captures its essence more clearly.
It's important to realise here that this postulate --- as does the next one --- identifies the values of reward functions across different state spaces, and therefore allows us to express powerful invariance properties.

\begin{RFpost}[RF]\label{post:dm:different:eigenspaces}
Consider any Hermitian operator~\(\measurement{A}\coloneqq\sum_{k=1}^r\eigval_k\projection[{\eigspace[k]}]\) on a Hilbert space~\(\hilbertspace_1\), with (distinct) real eigenvalues~\(\eigval_1,\dots,\eigval_r\) corresponding to respective mutually orthogonal eigenspaces~\(\eigspace[1],\dots,\eigspace[r]\) that span~\(\hilbertspace_1\).
Similarly, consider a Hermitian operator~\(\measurement{B}\coloneqq\sum_{k=1}^r\eigval_k\projection[{\alteigspace[k]}]\) on a Hilbert space~\(\hilbertspace_2\), with the same eigenvalues~\(\eigval_1,\dots,\eigval_r\), corresponding to respective mutually orthogonal eigenspaces~\(\alteigspace[1],\dots,\alteigspace[r]\) that span~\(\hilbertspace_2\).
Choose any normalised~\(\eigket[k]\in\eigspace[k]\) and~\(\alteigket[k]\in\alteigspace[k]\), and any~\(\con_k\in\complexes\) such that \(\sum_{k=1}^r\abs{\con_k}^2=1\), and consider the states~\(\fket[\measurement{A}]\coloneqq\sum_{k=1}^r\con_k\eigket[k]\in\hilbertspace_1\) and \(\fket[{\measurement{B}}]\coloneqq\sum_{k=1}^r\con_{k}\alteigket[k]\in\hilbertspace_2\).
Then \(\utility{A}(\fket[\measurement{A}])=\utility{B}(\fket[{\measurement{B}}])\).
\end{RFpost}
\noindent In essence, this requires that if a state~\(\fket\) is a superposition of eigenstates of a measurement~\(\measurement{A}\) corresponding to distinct eigenvalues, then \(\utility{A}\group{\fket}\) should depend \emph{only} on the superposition weights and on these eigenvalues, but not on the eigenstates themselves, nor on the Hilbert space they're embedded in.

The following examples are meant to illustrate two different applications of this second postulate and will provide intuition about what it entails and why it could be considered reasonable as an invariance requirement.

\begin{example}
Let's go back to our example involving a single qubit and assume that it represents the spin of an electron.
Assume that the respective states~\(\zeroket\) and~\(\oneket\), which constitute an orthonormal basis for the Hilbert space~\(\hilbertspace\), represent spin up and spin down in some physical direction, labelled as the \(z\)-direction.
We consider a measurement in this physical \(z\)-direction that is represented by the Hermitian operator~\(\measurement{A}\coloneqq\oneket\onebra-\zeroket\zerobra\in\measurements\); it yields the outcome~\(+1\) for spin up in this direction and \(-1\) for spin down.
The corresponding orthogonal eigenspaces are, respectively, \(\eigspace[+]=\linspanof{\set{\oneket}}\) and \(\eigspace[-]=\linspanof{\set{\zeroket}}\).

However, we can also use the transformed basis with states~\(\ket{0'}\coloneqq\oneket\) and~\(\ket{1'}\coloneqq\zeroket\) to model the same system.
Consider the operator~\(\measurement{B}\coloneqq\ket{1'}\bra{1'}-\ket{0'}\bra{0'}\), which is now a measurement in the negative physical \(z\)-direction.
\(\measurement{B}=-\measurement{A}\) has the same eigenvalues \(+1\) and \(-1\) as \(\measurement{A}\), but with exchanged eigenspaces~\(\alteigspace[+]=\linspanof{\set{\ket{1'}}}=\linspanof{\set{\ket{0}}}=\eigspace[-]\) and \(\alteigspace[-]=\linspanof{\set{\ket{0'}}}=\linspanof{\set{\ket{1}}}=\eigspace[+]\).

Now consider any~\(\con_1,\con_2\in\complexes\) such that \(\abs{\con_1}^2+\abs{\con_2}^2=1\), and the states~\(\fket[\measurement{A}]\coloneqq\alpha_1\ket{1}+\alpha_2\ket{0}\) and \(\fket[\measurement{B}]\coloneqq\alpha_1\ket{1'}+\alpha_2\ket{0'}=\alpha_1\ket{0}+\alpha_2\ket{1}\).

Even though the measurements~\(\measurement{A}\) and~\(\measurement{B}\), as well as the corresponding states~\(\fket[\measurement{A}]\) and~\(\fket[\measurement{B}]\) are clearly different, the second situation is a mere \emph{relabelling} of the first, and Postulate~\cref{post:dm:different:eigenspaces} therefore requires the corresponding rewards~\(\utility{B}(\fket[\measurement{B}])\) and~\(\utility{A}(\fket[\measurement{A}])\) in both situations to be the same.
\end{example}

\begin{example}
Consider an experiment where we conduct a measurement on two independent qubits.
The first qubit system is the one we're interested in, while the second qubit system isn't really of any interest to us; for example, this could be one of the spins of some free electron in extragalactic space.
We can then consider both systems simultaneously and regard them as one larger system.
In quantum mechanics, such a composition of two systems with respective Hilbert spaces~\(\hilbertspace\) and~\(\althilbertspace\) is described by the tensor product~\(\hilbertspace\tensortimes\althilbertspace\); see for instance Ref.~\cite[Sec.~II.4]{reed2003:methods:1}.

We define a measurement on this composite system as follows.
On the first system in a state~\(\gket\), we perform a measurement~\(\measurement{C}\), while leaving the second system untouched in some state~\(\fket\).
This leads to the tensor product~\(\measurement{C}\tensortimes\identity\) of the measurement~\(\measurement{C}\) on the first system with the identity measurement~\(\identity\) on the second one: \((\measurement{C}\tensortimes\identity)(\gket\tensortimes\fket)=(\measurement{C}\gket)\tensortimes\fket\).
Postulate~\cref{post:dm:different:eigenspaces} then implies that \(\utilitypure{\measurement{C}\tensortimes\identity}(\gket\tensortimes\fket)=\utility{C}(\gket)\).\footnote{Apply the postulate with \(\hilbertspace_1\instantiateas\hilbertspace\), \(\hilbertspace_2\instantiateas\hilbertspace\tensortimes\althilbertspace\), \(\measurement{A}\instantiateas\measurement{C}\), \(\measurement{B}\instantiateas\measurement{C}\tensortimes\identity\), \(\alteigspace[k]\instantiateas\eigspace[k]\tensortimes\althilbertspace\), \(\alteigket[k]\instantiateas\ket{a_k}\tensortimes\fket\), \(\fket[\measurement{A}]\instantiateas\gket\) and \(\fket[\measurement{B}]\instantiateas\gket\tensortimes\fket\).}
In more words, the reward~\(\utilitypure{\measurement{C}\tensortimes\identity}(\gket\tensortimes\fket)\) for this extended measurement on the larger system in the extended (so-called decoherent) state~\(\gket\tensortimes\fket\) shouldn't differ from the reward~\(\utility{C}(\gket)\) for the measurement on the first system in the state~\(\gket\); this is justifiable because the second system is independent of the first and the extended measurement leaves that second system untouched.
\end{example}
\noindent The first of the two examples above illustrates that \cref{post:dm:different:eigenspaces} requires reward functions to be \emph{invariant under relabelling}, and the second example hints at \emph{invariance under} a specific type of \emph{compression} of the Hilbert space.
We now give an alternative but equivalent formulation of the postulate that elucidates these two types of invariance requirements more clearly.

\medskip\noindent
{\bfseries Reward function postulate 2 (\cref{post:dm:different:eigenspaces}*, alternative formulation).}
\par\noindent{\itshape Consider any Hermitian operator~\(\measurement{A}\coloneqq\sum_{k=1}^r\eigval_k\projection[{\eigspace[k]}]\) on a Hilbert space~\(\hilbertspace_1\), with (distinct) real eigenvalues~\(\eigval_1,\dots,\eigval_r\) corresponding to respective mutually orthogonal eigenspaces~\(\eigspace[1],\dots,\eigspace[r]\) that span~\(\hilbertspace_1\).
Also consider any~\(r\)-dimensional Hilbert space~\(\hilbertspace_2\) with any orthonormal basis~\(\set{\alteigket[1]\dots,\alteigket[r]}\) and the Hermitian operator~\(\measurement{B}\coloneqq\sum_{k=1}^r\eigval_k\alteigket[k]\alteigbra[k]\) on~\(\hilbertspace_2\) with the same eigenvalues~\(\eigval_1,\dots,\eigval_r\) as~\(\measurement{A}\).
Choose any normalised~\(\eigket[k]\in\eigspace[k]\) and any~\(\con_k\in\complexes\) such that \(\sum_{k=1}^r\abs{\con_k}^2=1\), and consider the states~\(\fket[\measurement{A}]\coloneqq\sum_{k=1}^r\con_k\eigket[k]\in\hilbertspace_1\) and \(\fket[{\measurement{B}}]\coloneqq\sum_{k=1}^r\con_{k}\alteigket[k]\in\hilbertspace_2\).
Then \(\utility{A}(\fket[\measurement{A}])=\utility{B}(\fket[{\measurement{B}}])\).}

\begin{proof}[Brief argument that \cref{post:dm:different:eigenspaces}* and \cref{post:dm:different:eigenspaces} are equivalent]
\cref{post:dm:different:eigenspaces}* clearly follows from \cref{post:dm:different:eigenspaces} as a special case. That it also implies \cref{post:dm:different:eigenspaces} can be seen by applying it twice in opposite directions.
\end{proof}
\noindent This definitely exhibits the compression part of the invariance requirement: each eigenspace can be \emph{compressed} into a one-dimensional space without affecting the reward function.
Since, moreover, the choice of orthonormal basis in \(\hilbertspace_2\) has no effect, this shows that \emph{unitary transformations} (and in particular relabelling) have no impact on the reward function either.

The third postulate exploits the linearity of the utility and considers two measurements~\(\measurement{A}\) and~\(\measurement{B}\) with the same orthogonal eigenspaces~\(\eigspace[k]\) but possibly different corresponding eigenvalues~\(\lambda_k\) and~\(\mu_k\).
The sum~\(\measurement{A}+\measurement{B}\) then also has the same eigenspaces~\(\eigspace[k]\) with corresponding eigenvalues~\(\lambda_k+\mu_k\).
If the state~\(\fket\) is in one of these eigenspaces~\(\eigspace[k]\), then \cref{post:dm:eigenket} guarantees that \(\wval_{\measurement{A}+\measurement{B}}(\fket)=\lambda_k+\mu_k=\utility{A}(\fket)+\utility{B}(\fket)\), so the uncertain reward is additive on each of the eigenspaces~\(\eigspace[k]\).
We now require that this additivity should be extended from the eigenspaces to the entire Hilbert space~\(\hilbertspace\).

\begin{RFpost}[RF]\label{post:dm:additivity}
Consider any two Hermitian operators of the form~\(\measurement{A}\coloneqq\sum_{k=1}^r\eigval_k\projection[{\eigspace[k]}]\) and~\(\measurement{B}\coloneqq\sum_{k=1}^r\mu_k\projection[{\eigspace[k]}]\) on a Hilbert space~\(\hilbertspace\), where the \(r\) eigenspaces~\(\eigspace[k]\) are mutually orthogonal and span \(\hilbertspace\), where \(\eigval_k,\mu_k\in\reals\) and then all the~\(\lambda_k\) are distinct and all the~\(\mu_k\) are distinct.
Then \(\wval_{\measurement{A}+\measurement{B}}(\fket)=\utility{A}(\fket)+\utility{B}(\fket)\) for all~\(\fket\in\statespace\).
\end{RFpost}

The fourth and final postulate deals with the continuity of the uncertain rewards, the underlying idea being that if You're no longer able to distinguish between states, You shouldn't be able to distinguish between the corresponding rewards either: the reward functions~\(\utility{A}\) should be continuous in their state argument.

\begin{RFpost}[RF]\label{post:dm:continuity}
Let \(\measurement{A}\) be a Hermitian operator on a Hilbert space~\(\hilbertspace\) and let \(\fket[n]\) be any sequence of states, then \(\fket=\lim_{n\to+\infty}\fket[n]\) implies that \(\utility{A}(\fket)=\lim_{n\to+\infty}\utility{A}(\fket[n])\).
\end{RFpost}
\noindent Observe, by the way, that the continuity of the norm~\(\norm{\bolleke}\) in the topology it induces on the Hilbert space~\(\hilbertspace\), implies that the limit \(\fket=\lim_{n\to+\infty}\fket[n]\) of the sequence of states~\(\fket[n]\) is a state as well, so it makes sense to consider the value~\(\utility{A}(\fket)\) of the reward function~\(\utility{A}\) in that limit state~\(\fket\).

We now come to our main result.
To formulate it, we introduce the \emph{specific} so-called \emph{reward assignation}~\(\theutilitymap\colon\measurements\to\theutilities\colon\measurement{A}\mapsto\theutility{A}\), with
\begin{equation*}
\theutility{A}(\fket)\coloneqq\fbra\measurement{A}\fket
\text{ for all~\(\fket\in\statespace\)}
\end{equation*}
and with corresponding \emph{set of uncertain rewards}
\begin{equation*}
\theutilities
\coloneqq\cset{\theutility{A}}{\measurement{A}\in\measurements}
=\cset{\dotbra\measurement{A}\dotket}{\measurement{A}\in\measurements}.
\end{equation*}
We can and will prove that the postulates~\cref{post:dm:eigenket,post:dm:different:eigenspaces,post:dm:additivity,post:dm:continuity} determine the reward functions~\(\utility{A}\) unequivocally, in the sense that they imply that
\begin{equation}\label{eq:born:a:la:us}
\utility{A}(\fket)
=\theutility{A}(\fket)
=\fbra\measurement{A}\fket
\text{ for all }\measurement{A}\in\measurements
\text{ and all }\fket\in\statespace.
\end{equation}
We postpone the detailed and quite formal mathematical argumentation for this interesting result until \cref{sec:proofs}, so we can now, in the intervening sections, concentrate on a discussion of its implications.

\section{The basic decision-theoretic models}\label{sec:desirability}
Let's look at a specific decision problem involving the unknown state~\(\uket\) of a quantum system with Hilbert space~\(\hilbertspace\), with a corresponding set of measurement operators~\(\measurements\).
In one of its more general forms, decision theory will now use a \emph{strict (partial) vector ordering} to express Your preferences between acts, and therefore, indirectly, Your beliefs about~\(\uket\).

\subsection{Mathematical preliminaries}\label{sec:isomorphism}
Before delving into the details of this ordering, let's take a few moments to contemplate what it is that's being ordered, namely the acts~\(\act{A}\), which can be identified with the measurements~\(\measurement{A}\in\measurements\).
Our next result is based on the discussion in the previous section, and in particular on~\cref{eq:born:a:la:us}, and shows that the acts~\(\act{A}\), and the measurements~\(\measurement{A}\), are also in a one-to-one correspondence with the uncertain rewards~\(\theutility{A}\).

\begin{proposition}\label{prop:linear:isomorphism:measurements:and:rewards}
The reward assignation~\(\theutilitymap\) is a linear isomorphism between the real linear spaces~\(\measurements\) and~\(\theutilities\).
\end{proposition}
\noindent The proof is straightforward, but we include it for the sake of completeness.

\begin{proof}
Since, clearly, \(\measurements\) is a real linear space, it suffices that prove that (i) the reward assignation~\(\theutilitymap\) is linear; and (ii) that it's a bijection.

For~(i), consider any~\(\measurement{A},\measurement{B}\in\measurements\), any~\(\alpha,\beta\in\reals\) and any~\(\gket\in\statespace\), then \((\alpha\measurement{A}+\beta\measurement{B})\gket=\alpha\measurement{A}\gket+\beta\measurement{B}\gket\), so the bi-linearity of the inner product then guarantees that
\begin{align*}
\uval_{\alpha\measurement{A}+\beta\measurement{B}}(\gket)
&=\gbra(\alpha\measurement{A}+\beta\measurement{B})\gket
=\gbra(\alpha\measurement{A}\gket+\beta\measurement{B}\gket)
=\alpha\gbra\measurement{A}\gket+\beta\gbra\measurement{B}\gket\\
&=\alpha\theutility{A}(\gket)+\beta\theutility{B}(\gket),
\end{align*}
and therefore, indeed, \(\uval_{\alpha\measurement{A}+\beta\measurement{B}}=\alpha\theutility{A}+\beta\theutility{B}\), so the reward assignation~\(\theutilitymap\) is linear.

For~(ii), it suffices to prove that the reward assignation~\(\theutilitymap\) is one-to-one, as it's clearly onto by the definition of~\(\theutilities\).
So, consider any~\(\measurement{A},\measurement{B}\in\measurements\) and assume that \(\theutility{A}=\theutility{B}\), then we must show that \(\measurement{A}=\measurement{B}\).
If we let \(\measurement{C}\coloneqq\measurement{A}-\measurement{B}\), then we infer from the linearity of~\(\theutilitymap\) that \(\theutility{C}=\theutility{A}-\theutility{B}=0\), and we must prove that \(\measurement{C}=\zero\).
Since \(\measurement{C}\) is Hermitian, it has an orthogonal collection of eigenstates~\(\set{\eigket[1],\dots,\eigket[n]}\) that constitutes a basis for~\(\hilbertspace\).
We'll denote by \(\eigval_k\) the eigenvalue of~\(\measurement{C}\) that corresponds to the eigenstate~\(\eigket[k]\), for~\(k\in\set{1,\dots,n}\).
Consider any complex numbers~\(\con_1,\dots,\con_n\) with \(\sum_{k=1}^n\modulus{\con_k}^2=1\), and let \(\gket\coloneqq\sum_{k=1}^n\con_k\eigket[k]\).
Then it follows from the assumption and \cref{lem:quadratic:form:in:basis} that \(0=\theutility{C}(\gket)=\sum_{k=1}^n\modulus{\con_k}^2\eigval_k\), and since this must hold for all possible choices of the~\(\con_1,\dots,\con_n\), we infer that, necessarily, \(\eigval_1=\dots=\eigval_n=0\).
But then, indeed, by \cref{prop:basis:and:eigenvalues}, \(\measurement{C}=\sum_{k=1}^n\eigval_k\eigket[k]\eigbra[k]=\zero\).
\end{proof}

\noindent For this reason, we'll identify the measurements~\(\measurement{A}\) and the corresponding uncertain rewards~\(\theutility{A}=\dotbra\measurement{A}\dotket\), and the real linear spaces~\(\measurements\) and~\(\theutilities\).
There are two particular aspects of this identification that deserve extra attention.

First, for any real number~\(\mu\), the Hermitian operator~\(\mu\identity\) satisfies \((\mu\identity)\fket=\mu\fket\) for all~\(\fket\in\hilbertspace\) and therefore has a single eigenvalue~\(\mu\) with corresponding eigenspace~\(\hilbertspace\).
Postulate~\cref{post:qm:deterministic} --- or in this case equivalently \cref{post:qm:eigenvalues} --- then guarantees that the corresponding measurement always produces the outcome~\(\mu\) with certainty.
The reward assignation~\(\theutilitymap\) takes these constant measurements~\(\mu\identity\) to the (constant) maps~\(\uval_{\mu\identity}=\dotbra\mu\identity\dotket=\mu\dotbra\identity\dotket=\mu\).
We'll identify in our notations the real number~\(\mu\) and the constant map that assumes the value~\(\mu\).

Second, there are a few (vector) orderings of the real linear space~\(\theutilities\) that have a natural interpretation and will play an important role in what follows.
We begin with the so-called \emph{weak Pareto} or \emph{weak dominance} ordering~\(\geq\), which is the partial vector ordering on~\(\theutilities\) defined by
\[
\uval\geq\vval
\ifandonlyif\group{\forall\gket\in\statespace}\uval(\gket)\geq\vval(\gket),
\text{ for all~\(\uval,\vval\in\theutilities\),}
\]
and the corresponding strict vector ordering~\(\alwaysbetterthan\), also called \emph{weak strict dominance}, given by
\begin{align*}
\uval\alwaysbetterthan\vval
&\ifandonlyif\uval\geq\vval\text{ and }\uval\neq\vval\\
&\ifandonlyif\group{\forall\gket\in\statespace}\uval(\gket)\geq\vval(\gket)
\text{ and }\group{\exists\gket\in\statespace}\uval(\gket)>\vval(\gket),
\text{ for all~\(\uval,\vval\in\theutilities\).}
\end{align*}
There's also the \emph{strong dominance} ordering~\(\alwaysstrictlybetterthan\), which is the strict vector ordering on~\(\theutilities\) defined by
\[
\uval\alwaysstrictlybetterthan\vval
\ifandonlyif\group{\forall\gket\in\statespace}\uval(\gket)>\vval(\gket),
\text{ for all~\(\uval,\vval\in\theutilities\).}
\]
The (inverse of the) linear isomorphism~\(\theutilitymap\) induces corresponding vector orderings on the Hermitian operators, completely characterised as follows:
\begin{equation*}
\left.
\begin{aligned}
\measurement{A}\geq\zero
&\ifandonlyif\theutility{A}\geq\theutility{0}
\ifandonlyif\group{\forall\gket\in\statespace}\theutility{A}(\gket)\geq0
\ifandonlyif\min\spectrum{A}\geq0\\
\measurement{A}\alwaysbetterthan\zero
&\ifandonlyif\measurement{A}\geq\zero\text{ and }\measurement{A}\neq\zero\\
\measurement{A}\alwaysstrictlybetterthan\zero
&\ifandonlyif\theutility{A}\alwaysstrictlybetterthan\theutility{0}
\ifandonlyif\group{\forall\gket\in\statespace}\theutility{A}(\gket)>0
\ifandonlyif\min\spectrum{A}>0
\end{aligned}
\right\}
\text{ for all~\(\measurement{A}\in\measurements\).}
\end{equation*}
These orderings refer to well-known notions for Hermitian operators, some of which we had occasion to mention in our discussion of density operators in \cref{sec:quantum:mechanics:probabilistic}.
For instance, \(\measurement{A}\geq\zero\) means that \(\measurement{A}\) is \emph{positive semidefinite}, \(\measurement{A}\gneq\zero\) means that \(\measurement{A}\) is \emph{positive semidefinite and non-zero} and \(\measurement{A}\alwaysstrictlybetterthan\zero\) means that \(\measurement{A}\) is \emph{positive definite}.

On our way of looking at things, \(\measurement{A}\geq\measurement{B}\), or equivalently \(\measurement{A}-\measurement{B}\geq\zero\), means that \(\act{A}\) always produces an uncertain reward~\(\theutility{A}(\uket)\) that's at least as high as the uncertain reward~\(\theutility{B}(\uket)\) produced by~\(\act{B}\), regardless of what value~\(\uket\) assumes in~\(\statespace\).
That \(\measurement{A}\alwaysbetterthan\measurement{B}\), or equivalently \(\measurement{A}-\measurement{B}\alwaysbetterthan\zero\), means that, in addition, the corresponding rewards aren't equal for all values that~\(\uket\) may assume in~\(\statespace\).
We'll denote by~\(\posmeasurements\) the set of all non-null measurements with a non-negative uncertain reward, or in other words, all non-null positive semidefinite Hermitian operators.
Also, \(\negmeasurements\coloneqq-\posmeasurements\) is the set of all non-null negative semidefinite Hermitian operators.

Similarly, \(\measurement{A}\alwaysstrictlybetterthan\measurement{B}\), or equivalently  \(\measurement{A}-\measurement{B}\alwaysstrictlybetterthan\zero\), means that \(\act{A}\) always produces an uncertain reward~\(\theutility{A}(\uket)\) that's strictly higher than the uncertain reward~\(\theutility{B}(\uket)\) produced by~\(\act{B}\), regardless of what value~\(\uket\) assumes in~\(\statespace\).
We'll denote by \(\strictlyposmeasurements\) the set of all measurements with a positive uncertain reward, or in other words, all positive definite Hermitian operators.
Also, \(\strictlynegmeasurements\coloneqq-\strictlyposmeasurements\) is the set of all negative definite Hermitian operators.

\begin{example}\label{example:utility:maps}
To illustrate these ideas, we go back to the example of a qubit, and we refer to Ref.~\cite[Sec.~2.1.3]{nielsen2010:quantum} for the mathematical background details.

All Hermitian operators on its Hilbert space~\(\hilbertspace\) can be written as linear combinations of the identity operator~\(\identity\) and the Pauli operators~\(\paulix\), \(\pauliy\) and~\(\pauliz\).
In other words, for any~\(\measurement{A}\in\measurements\), there are scalars~\(w,x,y,z\in\reals\) such that \(\measurement{A}=w\identity+x\paulix+y\pauliy+z\pauliz\).

Moreover, for any~\(\alpha,\beta\in\complexes\), it holds that
\begin{multline*}
\paulix\group[\big]{\alpha\zeroket+\beta\oneket}=\beta\zeroket+\alpha\oneket
\text{ and }
\pauliy\group[\big]{\alpha\zeroket+\beta\oneket}=-i\beta\zeroket+i\alpha\oneket\\
\text{ and }
\pauliz\group[\big]{\alpha\zeroket+\beta\oneket}=\alpha\zeroket-\beta\oneket,
\end{multline*}
and therefore
\begin{multline*}
\uval_{\paulix}\group[\big]{\alpha\zeroket+\beta\oneket}=\alpha\beta^*+\alpha^*\beta
\text{ and }
\uval_{\pauliy}\group[\big]{\alpha\zeroket+\beta\oneket}=\im(\alpha\beta^*-\alpha^*\beta)\\
\text{ and }
\uval_{\pauliz}\group[\big]{\alpha\zeroket+\beta\oneket}=\alpha\alpha^*-\beta\beta^*,
\end{multline*}
so
\[
\theutility{A}\group[\big]{(\alpha\zeroket+\beta\oneket}
=\alpha\alpha^*(w+z)+\beta\beta^*(w-z)+\alpha\beta^*(x+\im y)+\alpha^*\beta(x-\im y).
\]
We also find after some algebraic manipulations that the eigenvalues~\(\lambda\) of~\(\measurement{A}\) are given by
\begin{equation}\label{eq:qubit:eigenvalues}
\lambda=w\pm\sqrt{x^2+y^2+z^2},
\end{equation}
and therefore
\begin{align}
\possemidefmeasurements
&=\cset{w\identity+x\paulix+y\pauliy+z\pauliz}{w,x,y,z\in\reals\text{ and }\sqrt{x^2+y^2+z^2}\leq w}\label{eq:qubit:possemidefcone}\\
\posmeasurements
&=\cset{w\identity+x\paulix+y\pauliy+z\pauliz}{w,x,y,z\in\reals\text{ and }\sqrt{x^2+y^2+z^2}\leq w\neq0}\label{eq:qubit:poscone}\\
\strictlyposmeasurements
&=\cset{w\identity+x\paulix+y\pauliy+z\pauliz}{w,x,y,z\in\reals\text{ and }\sqrt{x^2+y^2+z^2}<w}.\qedhere
\end{align}
\end{example}

We close this mathematical digression with a very brief foray into topology.
The so-called \emph{supremum} norm of a map~\(g\colon\statespace\to\reals\) is defined by~\(\supnorm{g}\coloneqq\sup_{\fket\in\statespace}\abs{g(\fket)}\).
It turns \(\theutilities\), and therefore indirectly also \(\measurements\) via the (inverse of the) linear isomorphism~\(\theutilitymap\), into a normed linear space; simply observe that all the elements~\(\theutility{A}\) of~\(\theutilities\) are \emph{bounded} --- have bounded supremum norm~\(\supnorm{\theutility{A}}\).
Indeed, for any~\(\measurement{A}\in\measurements\), we know from~\Cref{prop:basis:and:eigenvalues} that there's a basis of eigenvectors~\(\eigket[1],\dots,\eigket[n]\) of~\(\measurement{A}\) for the Hilbert space~\(\hilbertspace\), with corresponding eigenvalues~\(\eigval_1,\dots,\eigval_n\), so taking into account~\Cref{lem:normalisation,lem:quadratic:form:in:basis},
\begin{align}
\supnorm{\measurement{A}}
\coloneqq&\supnorm{\theutility{A}}
=\sup_{\fket\in\statespace}\abs{\theutility{A}(\fket)}
=\sup\cset[\bigg]{\abs[\bigg]{\sum_{k=1}^n\modulus{\con_k}^2\eigval_k}}{\sum_{k=1}^n\modulus{\con_k}^2=1}
\notag\\
=&\max\cset{\abs{\eigval}}{\eigval\in\spectrum{A}}.
\label{eq:supnorm}
\end{align}
Observe, by the way, that for the standard definition of the so-called \emph{operator norm} \cite[Ch.~23]{schechter1997} on the normed linear space~\(\hilbertspace\), there's a related result: taking into account that \(\adjoint{\measurement{A}}\measurement{A}=\measurement{A}^2\) is also Hermitian with the same eigenstates as~\(\measurement{A}\), and with eigenvalues that are the squares of the corresponding eigenvalues of~\(\measurement{A}\), we find that
\begin{align}
\opnorm{\measurement{A}}
\coloneqq&\sup_{\fket\in\hilbertspace\setminus\set{0}}\frac{\norm{\measurement{A}\fket}}{\norm{\fket}}
=\sup_{\fket\in\hilbertspace\setminus\set{0}}\frac{\sqrt{\fbra\adjoint{\measurement{A}}\measurement{A}\fket}}{\sqrt{\fbraket}}
=\sup_{\fket\in\statespace}\sqrt{\fbra\adjoint{\measurement{A}}\measurement{A}\fket}
=\supnorm{\adjoint{\measurement{A}}\measurement{A}}\notag\\
=&\max\cset{\abs\eigval}{\eigval\in\spec\group{\adjoint{\measurement{A}}\measurement{A}}}
=\max\cset{\eigval^2}{\eigval\in\spectrum{A}}.
\label{eq:opnorm:hermitian}
\end{align}

Since the real linear space~\(\measurements\) is finite-dimensional, all norms are equivalent --- lead to the same topology.
In \Cref{sec:hermitian:operators:basics} we came across the Frobenius norm~\(\frobnorm{\bolleke}\) on~\(\measurements\), which therefore gives rise to the same topology of open sets, the same topological interior operator~\(\interior(\bolleke)\) and the same topological closure operator~\(\closure(\bolleke)\), as the supremum norm~\(\supnorm{\bolleke}\) and the operator norm \(\opnorm{\bolleke}\) do.

Observe, by the way, that a similar argumentation, again based on~\Cref{lem:normalisation,lem:quadratic:form:in:basis}, allows us to infer that
\begin{align}
\inf\theutility{A}
\coloneqq&\inf_{\fket\in\statespace}\theutility{A}(\fket)
=\inf\cset[\bigg]{\sum_{k=1}^n\modulus{\con_k}^2\eigval_k}{\sum_{k=1}^n\modulus{\con_k}^2=1}
\notag\\
=&\min\cset{\eigval}{\eigval\in\spectrum{A}}
=\min\spectrum{A}.
\label{eq:infimum:utility}
\end{align}
and similarly that \(\sup\theutility{A}=\max\spectrum{A}\).

\subsection{Sets of desirable measurements}\label{sec:sets:of:desirable:measurements}
We'll take \(\theutility{A}\betterthan\theutility{B}\) to mean that, based on Your beliefs about the value of~\(\uket\) in~\(\statespace\), You strictly prefer\footnote{Such \emph{strict preference} can be given the following operationalisable meaning: You strictly prefer~\(\theutility{A}(\uket)\) to \(\theutility{B}(\uket)\) if You accept the uncertain reward~\(\theutility{A}(\uket)-\theutility{B}(\uket)\) but don't want to give it away; see for instance Ref.~\cite{quaeghebeur2015:statement} for a thorough discussion of such preferences, also leading to a justification for the axioms~\labelcref{ax:preference:irreflexitivy,ax:preference:transitivity,ax:preference:background,ax:preference:additivity,ax:preference:scaling} and \labelcref{ax:desirability:strict,ax:desirability:additivity,ax:desirability:scaling,ax:desirability:background}.} the uncertain reward~\(\theutility{A}(\uket)\) corresponding to measurement~\(\measurement{A}\) to the uncertain reward~\(\theutility{B}(\uket)\) corresponding to measurement~\(\measurement{B}\).
Thus, Your beliefs lead You to a (strict) preference relation~\(\betterthan\) on the set~\(\theutilities\): it collects those couples~\((\theutility{A},\theutility{B})\) for which You strictly prefer~\(\theutility{A}(\uket)\) to~\(\theutility{B}(\uket)\).
\emph{We'll say that Your beliefs are represented by the preference ordering~\(\betterthan\).}

As we'll see further on in \cref{sec:inference}, we'll want to allow for the possibility that the ordering~\(\betterthan\) is only a partial representation (also called an \emph{assessment}) of Your beliefs: You may not have all the time and resources needed to give an account of all Your preferences between the (infinitely many) uncertain rewards in~\(\theutilities\).
This is one practical reason why we don't require the ordering~\(\betterthan\) to be total.
Another, more fundamental, reason is that You may not have at Your disposal all the information that would lead You to impose a total ordering on the uncertain rewards in~\(\theutilities\).

But we do require the ordering~\(\betterthan\) to take into account the rational implications of those preferences that You do express, and to also satisfy certain non-inferential rationality requirements, such as respecting those \emph{background preferences} that should always be present, regardless of any information or beliefs You might have about the value of~\(\uket\) in~\(\statespace\).
These requirements, both inferential and non-inferential, are captured by the notion of \emph{coherence}, as expressed by the axioms~\labelcref{ax:preference:irreflexitivy,ax:preference:transitivity,ax:preference:background,ax:preference:additivity,ax:preference:scaling}, or equivalently \labelcref{ax:desirability:strict,ax:desirability:additivity,ax:desirability:scaling,ax:desirability:background}, further on.

Because we've argued above that the linear isomorphism~\(\theutilitymap\) allows us to move freely from measurements~\(\measurement{A}\) to uncertain rewards~\(\theutility{A}\) and backwards, we see that we can readily interpret a preference \(\theutility{A}\betterthan\theutility{B}\) as a preference between measurements (or between the corresponding acts): \(\measurement{A}\betterthan\measurement{B}\ifandonlyif\theutility{A}\betterthan\theutility{B}\).\footnote{In Appendix~\labelcref{app:POVMs}, we give an indication of how this preference ordering can be extended from so-called \emph{projective measurements} \(\measurement{A}\in\measurements\) to the more general case of POVM measurements.}

Once the step of focusing on such preference relations is taken, there's not much else we can do but apply the existing theory for choosing between uncertain rewards --- see, for instance Refs.~\cite{finetti19745,walley1991,debock2018,augustin2013:itip,seidenfeld1995,troffaes2013:lp,quaeghebeur2015:statement} and the discussion in \cref{sec:decision:theoretic:background} --- and translating everything back to preferences between measurements.
In doing so, we'll mainly follow the lead taken by Benavoli, Facchini and Zaffalon in their earlier work \cite{benavoli2019:computational,benavoli2016:quantum}, but we'll also be adding a few interesting details as we go along.
This is our programme for the remainder of this section and \cref{sec:lower:upper:previsions,sec:coherent:previsions}.

The notion of coherence captures the \emph{minimal} rationality requirements that we'll want Your preferences to satisfy.
We call a binary preference ordering~\(\betterthan\) on the space of measurements~\(\measurements\) \emph{coherent} if it satisfies the following conditions:
\begin{enumerate}[label={\upshape PO\arabic*.},ref={\upshape PO\arabic*},leftmargin=*,widest=5]
\item\label{ax:preference:irreflexitivy} \(\measurement{A}\notbetterthan\measurement{A}\) for all~\(\measurement{A}\in\measurements\);\hfill[irreflexivity]
\item\label{ax:preference:transitivity} \(\measurement{A}\betterthan\measurement{B}\) and \(\measurement{B}\betterthan\measurement{C}\then\measurement{A}\betterthan\measurement{C}\) for all~\(\measurement{A},\measurement{B},\measurement{C}\in\measurements\);\hfill[transitivity]
\item\label{ax:preference:scaling} \(\measurement{A}\betterthan\measurement{B}\then\lambda\measurement{A}\betterthan\lambda\measurement{B}\) for all~\(\measurement{A},\measurement{B}\in\measurements\) and all~\(\lambda\in\posreals\);\hfill[positive scaling]
\item\label{ax:preference:background} \(\measurement{A}\alwaysbetterthan\measurement{B}\then\measurement{A}\betterthan\measurement{B}\) for all~\(\measurement{A},\measurement{B}\in\measurements\).\hfill[monotonicity]
\item\label{ax:preference:additivity} \(\measurement{A}\betterthan\measurement{B}\then\group{\measurement{A}+\measurement{C}}\betterthan\group{\measurement{B}+\measurement{C}}\) for all~\(\measurement{A},\measurement{B},\measurement{C}\in\measurements\);\hfill[additivity]
\end{enumerate}
What lies behind these coherence requirements?

Axioms~\labelcref{ax:preference:irreflexitivy,ax:preference:transitivity} reflect the strict partial order aspect of the preference.
We don't require the ordering to be \emph{total}, by the way: \(\measurement{A}\notbetterthan\measurement{B}\) \emph{and} \(\measurement{B}\notbetterthan\measurement{A}\) needn't imply that \(\measurement{A}=\measurement{B}\); it may be that You're \emph{indifferent} between the different uncertain rewards~\(\theutility{A}(\uket)\) and~\(\theutility{B}(\uket)\) in the sense that You're willing to exchange any one of the two for the other, but alternatively \emph{also} that to You, \(\measurement{A}\) and \(\measurement{B}\) are \emph{incomparable} in the sense that You don't feel able or compelled to express a (weak or strict) preference between the uncertain rewards~\(\theutility{A}(\uket)\) and~\(\theutility{B}(\uket)\).

Axioms~\labelcref{ax:preference:additivity,ax:preference:scaling} turn the preference ordering into a \emph{vector} ordering that's compatible with, or preserved under, the addition and scalar multiplication of measurements.
This is a reflection of the linearity of the utility scale that rewards are expressed in, and which we assumed from the outset in \cref{sec:our:approach}.

And, finally, Axiom~\labelcref{ax:preference:background} expresses that Your preferences should take into account the natural \emph{background ordering}~\(\alwaysbetterthan\) between measurements.
If  \(\measurement{A}\alwaysbetterthan\measurement{B}\), or equivalently, \(\theutility{A}\alwaysbetterthan\theutility{B}\) then, regardless of Your beliefs about the unknown state~\(\uket\), You ought to always strictly prefer \(\measurement{A}\) over~\(\measurement{B}\), as the uncertain reward~\(\theutility{A}\) can never be lower than the uncertain reward~\(\theutility{B}\) in any given state, while for some states it will be strictly higher.\footnote{Most, if not all, of what we'll discuss in \cref{sec:sets:of:desirable:measurements,sec:inference,sec:lower:upper:previsions,sec:coherent:previsions} remains, mutatis mutandis, valid if we consider \(\alwaysstrictlybetterthan\) rather than \(\alwaysbetterthan\) as the background ordering and add the monotonicity requirement that \(\measurement{A}\betterthan\measurement{B}\) and \(\measurement{B}\geq\measurement{C}\) should imply that also \(\measurement{A}\betterthan\measurement{C}\), or equivalently, that \(\measurement{A}\betterthan\zero\) and \(\measurement{B}\geq\measurement{A}\) should imply that also \(\measurement{B}\betterthan\zero\); this connects directly to the remark in footnote~\ref{fn:alternative:desirability}.}
We can now summarise the result of our argumentation so far in the following way.

\begin{DMsumm}
Your beliefs about the value of~\(\uket\) in~\(\statespace\) lead You to strictly prefer some uncertain rewards in~\(\theutilities\) to others, which leads to a strict preference relation~\(\betterthan\) on the set of uncertain rewards~\(\theutilities\), and therefore also on the set of measurements~\(\measurements\), satisfying the coherence requirements~\labelcref{ax:preference:irreflexitivy,ax:preference:transitivity,ax:preference:background,ax:preference:additivity,ax:preference:scaling}.
\end{DMsumm}

Because a coherent~\(\betterthan\) is a strict vector ordering, we can represent it in a mathematically equivalent manner by the following convex cone of measurements:
\begin{equation*}
\desirset\coloneqq\cset{\measurement{A}\in\measurements}{\measurement{A}\betterthan\zero}.
\end{equation*}
Its elements are the measurements that You (strictly) prefer to the zero measurement~\(\zero\), which gives You zero utility and is therefore often called the \emph{status quo}.
We'll therefore call such measurements \emph{desirable} (to You), and the corresponding convex cone~\(\desirset\) a \emph{{\SDM}} (for You).
Observe that, by \labelcref{ax:preference:additivity},
\begin{equation}\label{eq:desirability:ordering:correspondence}
\measurement{A}\betterthan\measurement{B}
\ifandonlyif\measurement{A}-\measurement{B}\betterthan\zero
\ifandonlyif\measurement{A}-\measurement{B}\in\desirset,
\text{ for any two measurements~\(\measurement{A},\measurement{B}\in\measurements\)},
\end{equation}
which confirms that the ordering~\(\betterthan\) is indeed completely determined by the convex cone~\(\desirset\), and vice versa.

Such {\SDMs}~\(\desirset\) therefore constitute an \emph{equivalent representation} to preference orderings~\(\betterthan\).
The basic rationality requirements that are typically imposed on them are:
\begin{enumerate}[label={\upshape D\arabic*.},ref={\upshape D\arabic*},leftmargin=*,widest=4,series=desirability]
\item\label{ax:desirability:strict} \(\zero\notin\desirset\);
\item\label{ax:desirability:additivity} if \(\measurement{A},\measurement{B}\in\desirset\) then also \(\measurement{A}+\measurement{B}\in\desirset\);
\item\label{ax:desirability:scaling} if \(\measurement{A}\in\desirset\) and \(\lambda\in\posreals\) then also \(\lambda\measurement{A}\in\desirset\);
\item\label{ax:desirability:background} if \(\measurement{A}\alwaysbetterthan\zero\) then \(\measurement{A}\in\desirset\).
\end{enumerate}
They're the one-to-one counterparts, in that order and via \cref{eq:desirability:ordering:correspondence} and therefore \labelcref{ax:preference:additivity}, for the axioms~\labelcref{ax:preference:irreflexitivy,ax:preference:transitivity,ax:preference:background,ax:preference:scaling}.\footnote{In other words, and more explicitly, under \labelcref{ax:preference:additivity}, each of the requirements \labelcref{ax:preference:irreflexitivy,ax:preference:transitivity,ax:preference:background,ax:preference:scaling} for the preference ordering~\(\betterthan\) is equivalent to one of the requirements \labelcref{ax:desirability:strict,ax:desirability:additivity,ax:desirability:scaling,ax:desirability:background} for the {\SDM}~\(\desirset\), namely the one with the same number.}
Any subset~\(\desirset\) of~\(\measurements\) with these four properties is called a \emph{coherent} {\SDM} (in~\(\measurements\)).
Axiom~\labelcref{ax:desirability:strict} essentially states that Your ordering~\(\betterthan\) must be strict, and together with \labelcref{ax:desirability:additivity} guarantees, via \cref{eq:desirability:ordering:correspondence}, that You have a strict preference ordering between acts.
Axioms~\labelcref{ax:desirability:additivity,ax:desirability:scaling} state that \(\desirset\) must be a convex cone, so they --- in combination with \cref{eq:desirability:ordering:correspondence} and therefore \labelcref{ax:preference:additivity} --- also reflect the linearity of the utility scale.
The final axiom~\labelcref{ax:desirability:background} simply requires that non-null measurements that never yield a negative outcome --- that are positive semidefinite --- must be desirable to You.\footnote{\label{fn:alternative:desirability}For some applications, it may be preferable to replace \labelcref{ax:desirability:background} with the weaker condition that \(\desirset\) should include all positive definite measurements, and then also to explicitly add the requirement that \(\measurement{A}\in\desirset\) and \(\measurement{B}\geq\measurement{A}\) imply that \(\measurement{B}\in\desirset\); see \labelcref{ax:desirability:monotone} further on. Doing this wouldn't affect the main conclusions in this paper. Again, most of what we'll discuss in \cref{sec:sets:of:desirable:measurements,sec:inference,sec:lower:upper:previsions,sec:coherent:previsions} remains, mutatis mutandis, valid if we consider \(\alwaysstrictlybetterthan\) rather than \(\alwaysbetterthan\) as the background ordering.}

It follows readily from the coherence axioms~\labelcref{ax:desirability:strict,ax:desirability:additivity,ax:desirability:background} that a coherent {\SDM}~\(\desirset\) also has the following properties, where we denote by \(\negsemidefmeasurements\coloneqq-\possemidefmeasurements\) the set of all negative semidefinite measurements:
\begin{enumerate}[resume*=desirability,widest=6]
\item\label{ax:desirability:apl} \(\desirset\cap\negsemidefmeasurements=\emptyset\);
\item\label{ax:desirability:monotone} if \(\measurement{A}\in\desirset\) and \(\measurement{B}\geq\measurement{A}\) then also \(\measurement{B}\in\desirset\).
\end{enumerate}

We've been talking about desirable \emph{measurements}, but we want to remind  readers at this point that a measurement~\(\measurement{A}\) is considered to be desirable because You strictly prefer the associated uncertain reward~\(\theutility{A}(\uket)\) to the status quo~\(0\).
The desirability of measurements always goes back to the desirability of the associated uncertain rewards and thus to the desirability of \emph{gambles} --- bounded real-valued maps.
The coherence axioms~\labelcref{ax:desirability:strict,ax:desirability:additivity,ax:desirability:scaling,ax:desirability:background} are direct translations of coherence axioms for the desirability of gambles that appear in the relevant imprecise probabilities literature \cite{walley2000,quaeghebeur2015:statement,quaeghebeur2012:itip,decooman2015:coherent:predictive:inference,cooman2010,benavoli2016:quantum,cooman2021:archimedean:choice}, and in that sense the desirable measurements framework we're about to explore below, can be seen as a direct application of imprecise probabilities ideas to quantum theory.

It's important to stress here that there's \emph{nothing essentially probabilistic} about this representation of Your beliefs.
To give a simple example, the so-called \emph{vacuous} {\SDM}~\(\desirset[\mathrm{vac}]\coloneqq\cset{\measurement{A}\in\measurements}{\measurement{A}\alwaysbetterthan\zero}=\posmeasurements\) is coherent and therefore allowable as a candidate model for Your preferences.
It represents Your complete ignorance about~\(\uket\) in the sense that it reflects Your (quite conservative) lack of any inclination to engage in any~\(\act{A}\) whose reward function~\(\theutility{A}\) has a negative value in some states.
Recall that, for the corresponding \emph{vacuous strict preference ordering}~\(\alwaysbetterthan\),
\begin{equation*}
\measurement{A}\alwaysbetterthan\measurement{B}
\ifandonlyif\measurement{A}-\measurement{B}\alwaysbetterthan\zero
\ifandonlyif\measurement{A}\neq\measurement{B}\text{ and }\min\spec(\measurement{A}-\measurement{B})\geq0,
\text{ for all~\(\measurement{A},\measurement{B}\in\measurements\)}.
\end{equation*}
How this condition is expressed in terms of the eigenvalues of \emph{different}~\(\measurement{A}\) and \(\measurement{B}\) depends on how `commensurate' \(\measurement{A}\) and \(\measurement{B}\) are: the simplest condition arises when \(\measurement{A}\) and \(\measurement{B}\) commute, as then all eigenvalues of~\(\measurement{A}\) must be at least as large as the \emph{corresponding} eigenvalues of~\(\measurement{B}\) (with at least one eigenvalue being strictly larger).

Of course, many other, and arguably more interesting, types of preference orders are also allowed by this {\SDM} approach, and we'll have occasion to study quite a few of them in the following sections.
But before doing so, we want to briefly draw attention to an interesting aspect of our allowing Your preference ordering to be partial and not insisting on totality, as is often done \cite{savage1972,finetti19745,deutsch1999:quantum:decisions,wallace2003:defending:deutsch,wallace2007:improving:deutsch,wallace2009:born:arxiv}.
This is the type of approach advocated in the so-called imprecise probabilities literature \cite{walley1991,troffaes2013:lp,augustin2013:itip,debock2018,zaffalon2017:incomplete:preferences}, which allows the \emph{inferential aspects} of uncertain reasoning to come to the fore, as already hinted at in the Introduction.

\subsection{Conservative inference for desirability}\label{sec:inference}
When we want to model Your beliefs about the system under consideration, it may not be feasible to ask You to come up with all the measurements that are desirable to You --- or all Your strict preferences for that matter.
It may be much more reasonable to expect only a \emph{partial assessment}~\(\assessment\subseteq\measurements\) of the measurements --- uncertain rewards --- You deem desirable.
And because the coherence axioms~\labelcref{ax:desirability:strict,ax:desirability:additivity,ax:desirability:scaling,ax:desirability:background} are closed under taking arbitrary intersections, we can use the idea of closure, or conservative inference, to infer from this partial assessment the smallest --- most conservative --- coherent {\SDMs} that it includes, similarly to what is done in propositional logic by taking the deductive closure of a collection of propositions.
In this analogy, the measurements in~\(\posmeasurements\) play the role of the logical tautologies, and the measurements in~\(\negsemidefmeasurements\) that of the logical contradictions.

Let's investigate exactly how this conservative inference with sets of desirable measurements works: we want to find out \emph{if} and \emph{how} we can associate with a partial assessment~\(\assessment\) the most conservative coherent {\SDM} that includes it, where by `more conservative' we mean `containing fewer desirable measurements', so being smaller in terms of set inclusion.
The \emph{if} question is related to the notion of \emph{consistency}: we'll call an assessment~\(\assessment\) \emph{consistent} if it can be extended to --- is included in --- some coherent {\SDM}.
We're about to discover that the answers to the \emph{if} and the \emph{how} questions are intimately related.

We start by defining the \emph{positive hull} operator, which associates with any set of measurements~\(\assessment\subseteq\measurements\) the smallest convex cone that includes~\(\assessment\):
\[
\posi(\assessment)
\coloneqq\cset[\bigg]{\sum_{k=1}^n\lambda_k\measurement{A}_k}{n\in\naturals,\lambda_k>0,\measurement{A}_k\in\assessment}.
\]
It's now clear that \(\ext\group{\assessment} \coloneqq\posi(\assessment\cup\posmeasurements)\) is the smallest {\SDM} that includes~\(\assessment\) and satisfies axioms~\labelcref{ax:desirability:additivity,ax:desirability:scaling,ax:desirability:background}.
Indeed, any coherent {\SDM}~\(\desirset\) that includes~\(\assessment\) must also include~\(\assessment\cup\posmeasurements\) by \labelcref{ax:desirability:background}.
But \(\assessment\cup\posmeasurements\subseteq\desirset\) implies that also \(\ext\group{\assessment} =\posi\group{\assessment\cup\posmeasurements}\subseteq\posi\group{\desirset}=\desirset\), where the last equality follows from \labelcref{ax:desirability:additivity,ax:desirability:scaling}.
Moreover, the convex cone~\(\ext\group{\assessment} \) itself satisfies~\labelcref{ax:desirability:additivity,ax:desirability:scaling,ax:desirability:background}.
It's now easy to see that \emph{the assessment~\(\assessment\) is consistent if and only if \(\zero\notin\ext\group{\assessment} \), and that in that case \(\ext\group{\assessment} \) is the smallest coherent {\SDM} that includes~\(\assessment\).}

Indeed, if \(\zero\notin\ext\group{\assessment} \), then \(\ext\group{\assessment} \) satisfies all coherence axioms~\labelcref{ax:desirability:additivity,ax:desirability:scaling,ax:desirability:background,ax:desirability:strict} and is therefore the smallest coherent {\SDM} that includes~\(\assessment\), and therefore the most conservative coherent model that's compatible with Your assessment.
There are then typically infinitely many coherent extensions of~\(\assessment\), each one of which will include \(\ext\group{\assessment} \) and will therefore be more committal, or less conservative, in the sense that more desirable measurements will be included.
Your assessment is then indeed \emph{consistent}, because it can be extended to the coherent {\SDM}~\(\ext\group{\assessment} \).
If, on the other hand, \(\zero\in\ext\group{\assessment} \) then there's no compatible coherent model, meaning that Your assessment~\(\assessment\) is \emph{inconsistent}.

Since, moreover,
\begin{align}
\ext\group{\assessment}
&=\posi\group[\big]{\assessment\cup\posmeasurements}
=\posi\group[\big]{\posmeasurements}\cup\posi\group{\assessment}\cup\group[\big]{\posi\group{\assessment}+\posi\group{\posmeasurements}}\notag\\
&=\posmeasurements\cup\posi\group{\assessment}\cup\group[\big]{\posi\group{\assessment}+\posmeasurements}
=\posmeasurements\cup\group[\big]{\posi\group{\assessment}+\possemidefmeasurements},
\label{eq:natural:extension}
\end{align}
we can simplify the consistency condition as follows:
\begin{equation}\label{eq:consistency}
\zero\notin\ext\group{\assessment}
\ifandonlyif\zero\notin\posi\group{\assessment}+\possemidefmeasurements\\
\ifandonlyif\posi\group{\assessment}\cap\negsemidefmeasurements=\emptyset.
\end{equation}
It should be clear from this discussion that the operator~\(\ext\) acts like a deductive closure, or conservative inference, operator that associates with any assessment~\(\assessment\) the set \(\ext\group{\assessment} \) of all measurements whose desirability can be deduced from~\(\assessment\) taking into account the `production axioms' \labelcref{ax:desirability:additivity,ax:desirability:scaling,ax:desirability:background}.

\begin{example}\label{example:desirability:and:conservative:inference}
To illustrate these ideas, we go back to the example of a qubit.
If You have no knowledge at all about the system's state, then the set~\(\posmeasurements\) of all non-null positive semidefinite Hermitian operators is a reasonable model for Your desirable measurements.

Suppose You make the assessment that \(\indmeasurement{A}{o}=w_o\identity+x_o\paulix+y_o\pauliy+z_o\pauliz\) is desirable as well, for some~\(w_o,x_o,y_o,z_o\in\reals\).
If we recall \cref{eq:natural:extension}, then we see that the most conservative {\SDM} that corresponds to the assessment~\(\assessment\coloneqq\set{\indmeasurement{A}{o}}\) is
\begin{align}
\ext\group{\assessment}
&=\posmeasurements\cup\group[\big]{\posi\group{\assessment}+\possemidefmeasurements}
=\posmeasurements\cup\group[\big]{\cset{a\indmeasurement{A}{o}}{a>0}+\possemidefmeasurements}
\label{eq:qubit:natural:extension}\\
&=\Big\{(aw_o+w)\identity+(ax_o+x)\paulix+(ay_o+y)\pauliy+(az_o+z)\pauliz\colon\notag\\
&\hspace{18em}(a,w)>0\text{ and }w\geq\sqrt{x^2+y^2+z^2}\Big\}\notag\\
&=\Big\{w\identity+x\paulix+y\pauliy+z\pauliz\colon\notag\\
&\hspace{3em}(a,w-aw_o)>0\text{ and }w-aw_o\geq\sqrt{(x-ax_o)^2+(y-ay_o)^2+(z-az_o)^2}\Big\},
\label{eq:qubit:natural:extension:detailed}
\end{align}
where we let `\((a,b)>0\)' be a shorthand for `\(a,b\geq0\text{ and }a+b>0\)'; the third equality follows from~\cref{eq:qubit:possemidefcone,eq:qubit:poscone} and the final equality from the replacement \(ax_o+x\instantiateas x\), \(ay_o+y\instantiateas y\), \(az_o+z\instantiateas z\) and \(aw_o+w\instantiateas w\).
\cref{eq:consistency} guarantees that the consistency condition~\(\zero\notin\ext\group{\assessment} \) will be satisfied if and only if \(\indmeasurement{A}{o}\notin\negsemidefmeasurements\), which is in its turn equivalent to~\(\sqrt{x_o^2+y_o^2+z_o^2}>-w_o\), by \cref{eq:qubit:possemidefcone}.
In that case, Your assessment that \(\indmeasurement{A}{o}\) is desirable, is consistent, and \(\ext(\set{\indmeasurement{A}{o}})\) is the set of all measurements that You ought to then also deem desirable, as a result of Your making the desirability assessment~\(\assessment=\set{\indmeasurement{A}{o}}\).
\end{example}

To summarise: it's a consequence of our assumptions and argumentation that Your beliefs about the value of the uncertain quantum state~\(\uket\) in~\(\statespace\) can be modelled by a strict vector ordering on the linear space of uncertain rewards~\(\theutilities\), or equivalently, on the linear space of measurements~\(\measurements\), satisfying the axioms~\labelcref{ax:preference:irreflexitivy,ax:preference:transitivity,ax:preference:background,ax:preference:additivity,ax:preference:scaling}.
Such a strict vector ordering was then found to be mathematically equivalent to a coherent set of desirable measurements~\(\desirset\), which satisfies the axioms~\labelcref{ax:desirability:strict,ax:desirability:additivity,ax:desirability:scaling,ax:desirability:background}.
Any such coherent~\(\desirset\) is then a \emph{candidate model} for Your uncertainty about the system state~\(\uket\), very much like in classical probability theory, any probability measure is candidate model for your beliefs about the outcome of a classical experiment.
The specific~\(\desirset\) that You choose to use is then the one that You find most appropriate for Your beliefs about the system state~\(\uket\), the one that according to You expresses Your epistemic uncertainty about the system state~\(\uket\).
We'll see further on in \cref{prop:mix} that some of these coherent sets -- the weakly maximal ones -- correspond to density operators, so our theory encompasses the standard quantum mechanical formalism for describing Your epistemic uncertainty about the system state.

There are many ways You could come up with a coherent set of desirable measurements~\(\desirset\) that reflects Your beliefs.
One such way, based on partial assessments and conservative inference, has been discussed above, and we've tried to clarify how it could be implemented in practice in the various instalments of our qubit system example thus far.
Other possibilities are open to You, however, which are rather more closely related to the established culture of using density operators to represent Your epistemic uncertainty about a quantum system's state, and which we'll discuss in the next two sections.

\section{Archimedean models: coherent lower and upper previsions}\label{sec:lower:upper:previsions}
Coherent {\SDMs}~\(\desirset\), our most general models for describing Your uncertainty about the system state~\(\uket\), are indeed very general and can be quite complex.
Their complexity can be reduced somewhat by limiting ourselves to the special case of so-called \emph{Archimedean} models, where the associated preference ordering is essentially reduced to comparing real numbers.
This simplification is achieved by looking at \emph{coherent lower} and \emph{upper previsions}, which are studied in the imprecise probabilities literature; see for instance, Refs.~\cite{walley1991,augustin2013:itip,troffaes2013:lp,walley2000,cooman2021:archimedean:choice}.
Doing so will take us closer to the language of probabilities, but we'll warn in Section~\labelcref{sec:a:warning} against too naive an enthusiasm in making such interpretational jumps.

We're now going to discuss how such coherent lower (and upper) previsions can be derived from coherent sets of desirable measurements, and what their properties are.
In \cref{sec:coherent:previsions}, we'll then turn to the special case of so-called \emph{coherent previsions}, which are coherent lower and upper previsions that coincide, and show how they provide a connection with density operators, the existing models for uncertainty about a quantum state.
An overview of the relationship between the various types of models introduced and described in this paper is given in \cref{fig:models} near the end of that section.

In the imprecise probabilities literature, coherent lower and upper previsions are typically defined for \emph{gambles}, which are bounded real-valued maps interpreted as uncertain rewards.
To apply the theory to the present context, where we'll define them for \emph{measurements}, rather than gambles, two roads lie open to us, and we could venture upon both.
The more circuitous one consists in recalling from \cref{sec:isomorphism} that measurements~\(\measurement{A}\) are in a one-to-one relationship with their reward functions~\(\theutility{A}\), which are gambles.
Everything that can be done for gambles \cite{walley1991,augustin2013:itip,troffaes2013:lp,walley2000} can therefore also indirectly be done for measurements, through the identification between the two.
The more direct road, which we'll follow below for the most part, exploits that coherent lower and upper previsions have also been defined and studied more generally and directly for abstract vectors that live in a normed linear space \cite{cooman2021:archimedean:choice}, rather than for just for gambles.
The discussion below for measurements can therefore also be seen as a direct instantiation of that study.
Even though this means that, in much of the discussion below, we could content ourselves by merely stating results and referring to the above-mentioned literature for their proofs, we'll nevertheless provide short proofs for all of them, to make the discussion sufficiently self-contained.

\subsection{Buying and selling price functionals for measurements}
Let's start with a coherent {\SDM}~\(\desirset\) as a model for Your beliefs about~\(\uket\).
Its elements~\(\measurement{A}\in\desirset\) are the measurements (acts) You strictly prefer to the status quo~\(\zero\).
We now define Your corresponding \emph{supremum buying price}~\(\lowprev[\desirset]\group{\measurement{A}}\) for measurement~\(\measurement{A}\in\measurements\) as
\begin{equation}\label{eq:define:lower:prevision}
\lowprev[\desirset]\group{\measurement{A}}
\coloneqq\sup\cset{\alpha\in\reals}{\measurement{A}-\alpha\identity\in\desirset}
=\sup\cset{\alpha\in\reals}{\measurement{A}\betterthan\alpha\identity},
\end{equation}
the supremum (certain) price --- amount of utility --- You're willing to pay to acquire the uncertain reward~\(\theutility{A}(\uket)\) associated with performing measurement~\(\measurement{A}\).
Let's explain better.
The statement `\(\measurement{A}-\alpha\identity\in\desirset\)' means that the uncertain reward~\(\theutilitypure{\measurement{A}-\alpha\identity}(\uket)=\theutility{A}(\uket)-\alpha\theutility{I}(\uket)=\theutility{A}(\uket)-\alpha\) is desirable to You, or in other words, that You find it desirable to get the uncertain reward~\(\theutility{A}(\uket)\) and give away the fixed amount of utility~\(\alpha\) --- to buy the uncertain reward~\(\theutility{A}(\uket)\) for a price~\(\alpha\).
Hence, our use of the name `supremum buying price' for~\(\lowprev[\desirset]\group{\measurement{A}}\).

Similarly, we define Your \emph{infimum selling price}~\(\uppprev[\desirset]\group{\measurement{A}}\) for measurement~\(\measurement{A}\in\measurements\) as
\begin{equation}\label{eq:define:upper:prevision}
\uppprev[\desirset]\group{\measurement{A}}
\coloneqq\inf\cset{\beta\in\reals}{\beta\identity-\measurement{A}\in\desirset}
=\inf\cset{\beta\in\reals}{\beta\identity\betterthan\measurement{A}}
=-\lowprev[\desirset]\group{-\measurement{A}},
\end{equation}
the infimum (certain) price --- amount of utility --- You're willing to receive to give away the uncertain reward~\(\theutility{A}(\uket)\) associated with performing measurement~\(\measurement{A}\).
We'll also call \(\lowprev[\desirset]\) and \(\uppprev[\desirset]\) \emph{price functionals} for the {\SDM}~\(\desirset\).

What we've done so far, is to start from a coherent set of desirable measurements~\(\desirset\) and to derive from it Your corresponding buying price functional~\(\lowprev[\desirset]\) and selling price functional~\(\uppprev[\desirset]\).
But we could also go the other way around, and start from buying and/or selling price \emph{assessments} to come up with a desirability assessment~\(\assessment\), which in turn leads to a coherent set~\(\ext\group{\assessment} \) through conservative inference.
Let's see how this could work in practice.

\begin{example}\label{example:lower:and:upper:previsions}
We return to our running example about the qubit system, where we now consider the measurement~\(\paulix\).

Were You, for instance, to be certain that the outcome of the measurement~\(\paulix\) will be~\(+1\), then You'd find the transactions \(\paulix-\alpha\identity\) desirable for all~\(\alpha<1\), corresponding to a supremum desirable buying price of~\(1\) for~\(\paulix\).
Conversely, if You were certain that the outcome of the measurement~\(\paulix\) will be~\(-1\), then You'd find the transactions \(\beta\identity-\paulix\) desirable for all~\(\beta>1\), corresponding to an infimum desirable selling price of~\(1\) for~\(\paulix\).

But suppose You were, instead, to consider it (strictly) more likely that measuring the spin in the \(x\)-direction will yield the outcome~\(+1\) than that it will yield the other possible outcome~\(-1\).
How could You express that particular belief?
Well, in that case You'd strictly prefer getting the reward for the measurement~\(\paulix\) to giving it away and paying some non-negative fixed amount of utility~\(\beta\geq0\) for it, so You'd find the transaction~\(\paulix-(-\paulix-\beta\identity)=2\paulix+\beta\identity\) desirable for all real~\(\beta\geq0\), or by~\labelcref{ax:desirability:scaling}, You'd find \(\paulix-\alpha\identity\) desirable for all real~\(\alpha\leq0\).
Hence, Your belief corresponds to a desirability assessment \(\assessment\coloneqq\cset{\paulix-\alpha\identity}{\alpha\leq0}\) and leads to a supremum desirable buying price of~\(0\) for~\(\paulix\).

To see what the consequences of this assessment~\(\assessment\) are, we can use conservative inference to find the smallest coherent set of desirable measurements~\(\ext\group{\assessment} \) that includes~\(\assessment\).
With the notation~\(\indmeasurement{A}{o}=w_o\identity+x_o\paulix+y_o\pauliy+z_o\pauliz\) from Instalment~\labelcref{example:desirability:and:conservative:inference}, we see that
\begin{equation*}
\posi\group{\assessment}+\posmeasurements
=\posi\group{\set{\indmeasurement{A}{o}}}+\posmeasurements,
\end{equation*}
if we let \(\indmeasurement{A}{o}=\paulix\), so \((w_o,x_o,y_o,z_o)=(0,1,0,0)\).
This implies that \(\ext\group{\assessment} =\ext(\set{\indmeasurement{A}{o}})=\ext(\set{\paulix})\), so we can use the arguments from Instalment~\labelcref{example:desirability:and:conservative:inference} to calculate \(\ext\group{\assessment} \).
It's then easy to see that \(0=-w_o<\sqrt{x_o^2+y_o^2+z_o^2}=1\), so this assessment is consistent, and therefore \(\ext\group{\assessment} \) as characterised by~\cref{eq:qubit:natural:extension:detailed} is the smallest coherent {\SDM} that contains~\(\indmeasurement{A}{o}=\paulix\) and therefore represents Your assessment.

Your assessment \(\assessment\) has consequences for the supremum buying and infimum selling prices for measurements other than~\(\paulix\), which we can calculate via the coherent set \(\ext\group{\assessment} \) of desirable measurements that it leads to.
Indeed, with \(\measurement{A}\coloneqq w\identity+x\paulix+y\pauliy+z\pauliz\) and \(\alpha\in\reals\), we find that
\begin{align*}
\measurement{A}-\alpha\identity\in\posi(\assessment)+\possemidefmeasurements
&\ifandonlyif\group{\exists a>0}\alpha\leq w-\sqrt{(x-a)^2+y^2+z^2}\\
\shortintertext{and}
\measurement{A}-\alpha\identity\in\posmeasurements
&\ifandonlyif\group[\big]{w\neq\alpha\text{ and }\alpha\leq w-\sqrt{x^2+y^2+z^2}},
\end{align*}
where we started from \cref{eq:qubit:natural:extension} and used \cref{eq:qubit:possemidefcone,eq:qubit:poscone}.
These conditions, together with~\cref{eq:define:lower:prevision}, allow us to find the values for~\(\lowprev[\ext\group{\assessment} ]\group{\measurement{A}}\) and~\(\uppprev[\ext\group{\assessment} ]\group{\measurement{A}}=-\lowprev[\ext\group{\assessment} ]\group{-\measurement{A}}\) in any measurement~\(\measurement{A}\in\measurements\).
After some calculations, we find that
\begin{multline}\label{eq:qubit:lower:upper:pauli}
\lowprev[\ext\group{\assessment} ]\group{w\identity+x\paulix+y\pauliy+z\pauliz}
=\begin{cases}
w-\sqrt{y^2+z^2}&\text{if \(x\geq0\)}\\
w-\sqrt{x^2+y^2+z^2}&\text{if \(x\leq0\)}
\end{cases}
\text{ and }\\
\uppprev[\ext\group{\assessment} ]\group{w\identity+x\paulix+y\pauliy+z\pauliz}
=\begin{cases}
w+\sqrt{x^2+y^2+z^2}&\text{if \(x\geq0\)}\\
w+\sqrt{y^2+z^2}&\text{if \(x\leq0\)},
\end{cases}
\end{multline}
so, in particular,
\begin{equation}\label{eq:qubit:lower:upper:pauli:special:cases}
\left\{
\begin{aligned}
\lowprev[\ext\group{\assessment} ]\group{\paulix}&=0\\
\uppprev[\ext\group{\assessment} ]\group{\paulix}&=1
\end{aligned}
\right.
\text{ and }
\left\{
\begin{aligned}
\lowprev[\ext\group{\assessment} ]\group{\pauliy}&=-1\\
\uppprev[\ext\group{\assessment} ]\group{\pauliy}&=1
\end{aligned}
\right.
\text{ and }
\left\{
\begin{aligned}
\lowprev[\ext\group{\assessment} ]\group{\paulix-\pauliy}&=-1\\
\uppprev[\ext\group{\assessment} ]\group{\paulix-\pauliy}&=\sqrt2.
\end{aligned}
\right.
\end{equation}
The discussion in \cref{sec:coherent:previsions}, and in particular in \cref{sec:a:warning}, will allow us to conclude that, on a more traditional view of probability in quantum mechanics, \(\lowprev[\ext\group{\assessment} ]\group{\measurement{A}}\) and \(\uppprev[\ext\group{\assessment} ]\group{\measurement{A}}\) \emph{can} be interpreted as tight lower and upper bounds on the expected values of the outcome of the measurement~\(\measurement{A}\), based on Your assessment \(\assessment\).
We now see in the expressions above that Your assessment that \(1\) is a likelier outcome than~\(-1\) for the measurement~\(\paulix\), allows You to tightly bound the expected value of its outcome to lie between~\(0\) and~\(1\).
For the expected outcome of the measurement~\(\pauliy\), Your assessment \(\assessment\) leads to the \emph{vacuous} bounds \(-1\) and \(1\), so it's really totally non-informative about it.

By the way, the difference \(\uppprev[\ext\group{\assessment} ]\group{\measurement{A}}-\lowprev[\ext\group{\assessment} ]\group{\measurement{A}}\) is called the \emph{imprecision}, or \emph{bid-ask spread}, for~\(\measurement{A}=w\identity+x\paulix+y\pauliy+z\pauliz\), and it's given by \(\sqrt{x^2+y^2+z^2}+\sqrt{y^2+z^2}\).

Although typically somewhat more involved, essentially the same type of arguments can be used to find \(\ext\group{\assessment} \) and \(\lowprev[\ext\group{\assessment} ]\) and \(\uppprev[\ext\group{\assessment} ]\) for more involved assessments~\(\assessment\) and more general Hilbert spaces~\(\hilbertspace\), and they will typically involve solving some kind of semidefinite programming problem \cite[Sec.~4.6.2]{boyd:2004:convex:optimization}.
\end{example}

\subsection{Coherent lower and upper previsions}
We've seen that price functionals can be derived from desirability assessments and that, conversely, specifying buying and selling prices corresponds to specific desirability assessments.
Let's explore the connection a bit further.

First, all the price functionals~\(\lowprev[\desirset]\) that are obtained by letting~\(\desirset\) range over all the coherent {\SDMs} are exactly all the real functionals~\(\lowprev\colon\measurements\to\reals\) that satisfy the following properties:
\begin{enumerate}[label={\upshape LP\arabic*.},ref={\upshape LP\arabic*},leftmargin=*,series=LP]
\item\label{ax:lower:prevision:bounds} \(\lowprev\group{\measurement{A}}\geq\min\spectrum{A}\) for all~\(\measurement{A}\in\measurements\);\hfill[bounds]
\item\label{ax:lower:prevision:superadditivity} \(\lowprev\group{\measurement{A}+\measurement{B}}\geq\lowprev\group{\measurement{A}}+\lowprev\group{\measurement{B}}\) for all~\(\measurement{A},\measurement{B}\in\measurements\);\hfill[super-additivity]
\item\label{ax:lower:prevision:scaling} \(\lowprev\group{\lambda\measurement{A}}=\lambda\lowprev\group{\measurement{A}}\) for all~\(\measurement{A}\in\measurements\) and all real~\(\lambda\geq0\).\hfill[non-negative homogeneity]
\end{enumerate}

\begin{proof}
First, let's assume that \(\desirset\) is a coherent {\SDM}, and show that the corresponding price functional~\(\lowprev[\desirset]\) is real-valued and satisfies \labelcref{ax:lower:prevision:bounds,ax:lower:prevision:superadditivity,ax:lower:prevision:scaling}.
First, consider any~\(\measurement{A}\in\measurements\), then it follows from \labelcref{ax:desirability:monotone} that the set \(L_{\measurement{A}}\coloneqq\cset{\alpha\in\reals}{\measurement{A}-\alpha\identity\in\desirset}\) is decreasing and therefore bounded above by any element of its set-theoretic complement \(\reals\setminus L_{\measurement{A}}=\cset{\alpha\in\reals}{\measurement{A}-\alpha\identity\notin\desirset}\).
Since it readily follows from \labelcref{ax:desirability:apl} that \(\max\spectrum{A}\notin L_{\measurement{A}}\) and from \labelcref{ax:desirability:background} that \((-\infty,\min\spectrum{A})\subseteq L_{\measurement{A}}\), we find that \(\min\spectrum{A}\leq\sup L_{\measurement{A}}\leq\max\spectrum{A}\) and therefore that \(\lowprev[\desirset]\group{\measurement{A}}=\sup L_{\measurement{A}}\) is real.
This also shows that \(\lowprev[\desirset]\) satisfies~\labelcref{ax:lower:prevision:bounds}.
Also, consider any~\(\alpha,\beta\in\reals\) such that \(\measurement{A}-\alpha\identity\in\desirset\) and \(\measurement{B}-\beta\identity\in\desirset\), then \(\measurement{A}+\measurement{B}-(\alpha+\beta)\identity\in\desirset\) by \labelcref{ax:desirability:additivity}.
Invoking \cref{eq:define:lower:prevision} now leads to the conclusion that \(\lowprev[\desirset]\) satisfies~\labelcref{ax:lower:prevision:superadditivity}.
For~\labelcref{ax:lower:prevision:scaling}, first infer from~\labelcref{ax:lower:prevision:bounds,ax:lower:prevision:superadditivity} for~\(\measurement{A}=\measurement{B}=\zero\) that \(\lowprev[\desirset](\zero)=0\), so we need only check that \(\lowprev[\desirset]\) satisfies~\labelcref{ax:lower:prevision:scaling} for~\(\lambda>0\).
Now observe that, by \labelcref{ax:desirability:scaling}, \(\lambda\measurement{A}-\alpha\identity\in\desirset\ifandonlyif\measurement{A}-\nicefrac{\alpha}{\lambda}\identity\in\desirset\), and invoke \cref{eq:define:lower:prevision}.

Conversely, assume that the real functional~\(\lowprev\colon\measurements\to\reals\) satisfies \labelcref{ax:lower:prevision:bounds,ax:lower:prevision:superadditivity,ax:lower:prevision:scaling}, then we must show that there's some coherent {\SDM}~\(\desirset\) such that \(\lowprev=\lowprev[\desirset]\).
If we let
\begin{equation}\label{eq:desirs:from:lowprev:basis}
\desirset[\lowprev]
\coloneqq\cset{\measurement{A}\in\measurements}{\lowprev\group{\measurement{A}}>0\text{ or }\measurement{A}\alwaysbetterthan\zero},
\end{equation}
then
\begin{equation}\label{eq:desirs:from:lowprev:first}
\desirset[\lowprev]'
\coloneqq\cset{\measurement{A}\in\measurements}{\lowprev\group{\measurement{A}}>0}
\subseteq\desirset[\lowprev]
\subseteq\cset{\measurement{A}\in\measurements}{\lowprev\group{\measurement{A}}\geq0}
\eqqcolon\desirset[\lowprev]'',
\end{equation}
where the second inclusion follows from~\labelcref{ax:lower:prevision:bounds}.
Also observe that by the constant additivity property~\labelcref{ax:lower:prevision:constant:additivity}, which as we'll show further on follows from \labelcref{ax:lower:prevision:bounds,ax:lower:prevision:superadditivity,ax:lower:prevision:scaling}, it holds that
\begin{equation}\label{eq:desirs:from:lowprev:second}
\measurement{A}-\alpha\identity\in\desirset[\lowprev]'
\ifandonlyif\lowprev\group{\measurement{A}}>\alpha
\text{ and }
\measurement{A}-\alpha\identity\in\desirset[\lowprev]''
\ifandonlyif\lowprev\group{\measurement{A}}\geq\alpha
\text{ for all~\(\alpha\in\reals\)}.
\end{equation}
If we now combine \cref{eq:desirs:from:lowprev:first,eq:desirs:from:lowprev:second} with \cref{eq:define:lower:prevision}, we find that \(\lowprev=\lowprev[{\desirset[\lowprev]}]\), so it only remains to show that \(\desirset[\lowprev]\) is (a) coherent ({\SDM}), or in other words satisfies~\labelcref{ax:desirability:strict,ax:desirability:additivity,ax:desirability:scaling,ax:desirability:background}.
For~\labelcref{ax:desirability:strict}, simply observe that we've already proved above that \(\lowprev(\zero)=0\), so indeed \(\zero\notin\desirset[\lowprev]\).
For~\labelcref{ax:desirability:additivity}, consider any~\(\measurement{A},\measurement{B}\in\desirset[\lowprev]\), then there are a number of possibilities.
If \(\measurement{A}\alwaysbetterthan\zero\) and \(\measurement{B}\alwaysbetterthan\zero\), then clearly also \(\measurement{A}+\measurement{B}\alwaysbetterthan\zero\), and therefore \(\measurement{A}+\measurement{B}\in\desirset[\lowprev]\).
If \(\lowprev\group{\measurement{A}}>0\) and \(\lowprev\group{\measurement{B}}>0\), then by~\labelcref{ax:lower:prevision:superadditivity} also \(\lowprev\group{\measurement{A}+\measurement{B}}\geq\lowprev\group{\measurement{A}}+\lowprev\group{\measurement{B}}>0\), and therefore \(\measurement{A}+\measurement{B}\in\desirset[\lowprev]\).
For the remaining possibilities, we may assume without loss of generality that \(\lowprev\group{\measurement{A}}>0\) and \(\measurement{B}\alwaysbetterthan\zero\).
But then \(\lowprev\group{\measurement{B}}\geq0\) by \labelcref{ax:lower:prevision:bounds}, and therefore by~\labelcref{ax:lower:prevision:superadditivity} also \(\lowprev\group{\measurement{A}+\measurement{B}}\geq\lowprev\group{\measurement{A}}+\lowprev\group{\measurement{B}}>0\), whence, here too, \(\measurement{A}+\measurement{B}\in\desirset[\lowprev]\).
That \(\desirset[\lowprev]\) satisfies~\labelcref{ax:desirability:scaling} follows readily from~\labelcref{ax:lower:prevision:scaling}, and, finally, \(\desirset[\lowprev]\) satisfies~\labelcref{ax:desirability:background} by construction.
\end{proof}

We see that, in our present context, there's something very special about the real functionals~\(\lowprev\colon\measurements\to\reals\) that satisfy the so-called \emph{coherence conditions}~\labelcref{ax:lower:prevision:bounds,ax:lower:prevision:superadditivity,ax:lower:prevision:scaling}: each of them is a buying price functional for some coherent set of desirable measurements.
Adopting Walley's \cite{walley1991} nomenclature,\footnote{Since, as we've shown in \cref{eq:infimum:utility}, \(\min\spectrum{A}=\inf\theutility{A}\) for all \(\measurement{A}\in\measurements\), the correspondence between our coherence requirements for lower previsions on the linear space of measurements and Walley's coherence result for lower previsions on a subspace of the space of all gambles, is immediate; see \cite[Thm.~2.8.4]{walley1991}.} we'll call them \emph{coherent lower previsions} on~\(\measurements\).
They're worthy of study by themselves, as alternative, and as we'll see further on, almost equivalent but slightly less expressive, models for Your uncertainty about the state of a quantum system.

The corresponding \emph{coherent upper previsions}~\(\uppprev\) are related to the coherent lower previsions~\(\lowprev\) through the \emph{conjugacy relationship}: \(\uppprev\group{\measurement{A}}=-\lowprev\group{-\measurement{A}}\) for all~\(\measurement{A}\in\measurements\).

Coherent lower and upper previsions also always satisfy the following useful properties, which are immediate consequences of the coherence conditions~\labelcref{ax:lower:prevision:bounds,ax:lower:prevision:superadditivity,ax:lower:prevision:scaling}.
The convergence in~\labelcref{ax:lower:prevision:convergence} is convergence in supremum norm~\(\supnorm{\bolleke}\), or equivalently, since all norms are equivalent --- lead to the same topology --- on the finite dimensional space~\(\measurements\), convergence in Frobenius norm~\(\frobnorm{\bolleke}\) or in operator norm~\(\opnorm{\bolleke}\).
\begin{enumerate}[resume*=LP,widest=6]
\item\label{ax:lower:prevision:more:bounds} \(\min\spectrum{A}\leq\lowprev\group{\measurement{A}}\leq\uppprev\group{\measurement{A}}\leq\max\spectrum{A}\);\upshape\hfill[more bounds]
\item\label{ax:lower:prevision:mixed:additivity} \(\lowprev\group{\measurement{A}}+\lowprev\group{\measurement{B}}\leq\lowprev\group{\measurement{A}+\measurement{B}}\leq\lowprev\group{\measurement{A}}+\uppprev\group{\measurement{B}}\leq\uppprev\group{\measurement{A}+\measurement{B}}\leq\uppprev\group{\measurement{A}}+\uppprev\group{\measurement{B}}\);
\item\label{ax:lower:prevision:constant:additivity} \(\uppprev\group{\measurement{A}+\mu\identity}=\uppprev\group{\measurement{A}}+\mu\identity\) and~\(\lowprev\group{\measurement{A}+\mu}=\lowprev\group{\measurement{A}}+\mu\);\upshape\hfill[constant additivity]
\item\label{ax:lower:prevision:monotonicity} if \(\measurement{A}\geq\measurement{B}\) then \(\uppprev\group{\measurement{A}}\geq\uppprev\group{\measurement{B}}\) and~\(\lowprev\group{\measurement{A}}\geq\lowprev\group{\measurement{B}}\);\upshape\hfill[monotonicity]
\item\label{ax:lower:prevision:convergence} if the sequence \(\indmeasurement{A}{n}\in\measurements\) converges to the measurement~\(\measurement{A}\), then \(\uppprev\group{\indmeasurement{A}{n}}\to\uppprev\group{\measurement{A}}\) and~\(\lowprev\group{\indmeasurement{A}{n}}\to\lowprev\group{\measurement{A}}\).\upshape\hfill[continuity]
\end{enumerate}
\begin{proof}
\labelcref{ax:lower:prevision:more:bounds}.
The first and third inequalities follow directly from~\labelcref{ax:lower:prevision:bounds} and conjugacy, so we concentrate on the second inequality.
Simply observe that
\[
0
=\lowprev\group{\zero}
=\lowprev\group{\measurement{A}-\measurement{A}}
\geq\lowprev\group{\measurement{A}}+\lowprev\group{-\measurement{A}}
=\lowprev\group{\measurement{A}}-\uppprev\group{\measurement{A}},
\]
where the first equality follows from~\labelcref{ax:lower:prevision:scaling} and the first inequality from~\labelcref{ax:lower:prevision:superadditivity}.

\labelcref{ax:lower:prevision:mixed:additivity}.
It suffices to prove the second equality, as all the other inequalities will follow then by also invoking~\labelcref{ax:lower:prevision:superadditivity} and conjugacy.
Now simply observe that
\[
\lowprev\group{\measurement{A}}
=\lowprev\group{\measurement{A}+\measurement{B}-\measurement{B}}
\geq\lowprev\group{\measurement{A}+\measurement{B}}+\lowprev\group{-\measurement{B}}
=\lowprev\group{\measurement{A}+\measurement{B}}-\uppprev\group{\measurement{B}},
\]
where the first inequality follows from~\labelcref{ax:lower:prevision:superadditivity}.

\labelcref{ax:lower:prevision:constant:additivity}.
It suffices to prove the second equality, as the first then follows by conjugacy.
First, infer from~\labelcref{ax:lower:prevision:more:bounds} that \(\lowprev\group{\mu\identity}=\uppprev\group{\mu\identity}=\mu\), and then invoke~\labelcref{ax:lower:prevision:mixed:additivity} to find that
\[
\lowprev\group{\measurement{A}}+\mu
=\lowprev\group{\measurement{A}}+\lowprev\group{\mu\identity}
\leq\lowprev\group{\measurement{A}+\mu\identity}
\leq\lowprev\group{\measurement{A}}+\uppprev\group{\mu\identity}
=\lowprev\group{\measurement{A}}+\mu.
\]

\labelcref{ax:lower:prevision:monotonicity}.
We concentrate on proving the second implication, as the first will then follow from conjugacy.
Assume that \(\measurement{A}\geq\measurement{B}\), or equivalently, \(\measurement{A}-\measurement{B}\geq\zero\), then
\[
0
\leq\min\spec\group{\measurement{A}-\measurement{B}}
\leq\lowprev\group{\measurement{A}-\measurement{B}}
\leq\lowprev\group{\measurement{A}}+\uppprev\group{-\measurement{B}}
=\lowprev\group{\measurement{A}}-\lowprev\group{\measurement{B}},
\]
where the first inequality follows from the definition of the vector ordering~\(\geq\), the second inequality follows from \labelcref{ax:lower:prevision:bounds} and the third inequality from~\labelcref{ax:lower:prevision:mixed:additivity}.

\labelcref{ax:lower:prevision:convergence}.
We concentrate on proving the second implication, as the first will then follow from conjugacy.
Simply observe that
\[
\lowprev\group{\measurement{A}}-\lowprev\group{\measurement{B}}
=\lowprev\group{\measurement{A}}+\uppprev\group{-\measurement{B}}
\leq\uppprev\group{\measurement{A}-\measurement{B}}
\leq\max\spec\group{\measurement{A}-\measurement{B}}
\leq\supnorm{\measurement{A}-\measurement{B}},
\]
where the first inequality follows from~\labelcref{ax:lower:prevision:mixed:additivity}, the second one from~\labelcref{ax:lower:prevision:more:bounds} and the third one from \cref{eq:supnorm}.
Since, similarly, \(\lowprev\group{\measurement{B}}-\lowprev\group{\measurement{A}}\leq\supnorm{\measurement{A}-\measurement{B}}\), we find that
\begin{equation*}
\abs{\lowprev\group{\measurement{A}}-\lowprev\group{\measurement{B}}}
\leq\supnorm{\measurement{A}-\measurement{B}}
\text{ for all~\(\measurement{A},\measurement{B}\in\measurements\)},
\end{equation*}
which even proves the Lipschitz continuity of~\(\lowprev\).
\end{proof}

To better understand the relationship between coherent sets of desirable measurements and coherent lower previsions, we observe that
\begin{equation}\label{eq:desirs:lowprev}
\measurement{A}\in\interior\group{\desirset}
\ifandonlyif\lowprev[\desirset]\group{\measurement{A}}>0
\text{ and }
\measurement{A}\in\closure\group{\desirset}
\ifandonlyif\lowprev[\desirset]\group{\measurement{A}}\geq0,
\text{ for all~\(\measurement{A}\in\measurements\)}.
\end{equation}

\begin{proof}
To prove~\cref{eq:desirs:lowprev}, consider any~\(\measurement{A}\in\measurements\).

Assume that \(\measurement{A}\in\interior\group{\desirset}\), implying that there's some real~\(\epsilon>0\) such that the open ball~\(B_\epsilon(\measurement{A})\coloneqq\cset{\measurement{B}\in\measurements}{\supnorm{\measurement{B}-\measurement{A}}<\epsilon}\subseteq\desirset\).
Since \(\supnorm{\group{\measurement{A}-\nicefrac\epsilon2\identity}-\measurement{A}}=\nicefrac\epsilon2\), we find that \(\measurement{A}-\nicefrac\epsilon2\identity\in\desirset\), and therefore, indeed, \(\lowprev[\desirset]\group{\measurement{A}}\geq\nicefrac\epsilon2>0\).

Assume that \(\measurement{A}\notin\interior\group{\desirset}\), and therefore \(\measurement{A}\in\closure\group{\measurements\setminus\desirset}\), implying that there's some sequence \(\measurement{A}_n\in\measurements\setminus\desirset\) such that \(\supnorm{\measurement{A}-\measurement{A}_n}\to0\).
But \(\measurement{A}_n\notin\desirset\) implies that \(\lowprev[\desirset]\group{\measurement{A}_n}\leq0\) [use \cref{lem:desirs:from:lowprev}], and therefore the continuity of~\(\lowprev[\desirset]\) [use~\labelcref{ax:lower:prevision:convergence}] implies that also \(\lowprev[\desirset]\group{\measurement{A}}\leq0\).

Assume that \(\measurement{A}\in\closure\group{\desirset}\), which implies that there's some sequence \(\measurement{A}_n\in\desirset\) such that \(\supnorm{\measurement{A}-\measurement{A}_n}\to0\).
But \(\measurement{A}_n\in\desirset\) implies that \(\lowprev[\desirset]\group{\measurement{A}_n}\geq0\), and therefore the continuity of~\(\lowprev[\desirset]\) [use~\labelcref{ax:lower:prevision:convergence}] implies that also \(\lowprev[\desirset]\group{\measurement{A}}\geq0\).

Assume that \(\measurement{A}\notin\closure\group{\desirset}\) and therefore \(\measurement{A}\in\interior\group{\measurements\setminus\desirset}\), implying that there's some real~\(\epsilon>0\) such that the open ball~\(B_\epsilon(\measurement{A})\coloneqq\cset{\measurement{B}\in\measurements}{\supnorm{\measurement{B}-\measurement{A}}<\epsilon}\subseteq\measurements\setminus\desirset\).
Since \(\supnorm{\group{\measurement{A}+\nicefrac\epsilon2\identity}-\measurement{A}}=\nicefrac\epsilon2\), we find that \(\measurement{A}+\nicefrac\epsilon2\identity\notin\desirset\), which implies that \(\lowprev[\desirset]\group{\measurement{A}+\nicefrac\epsilon2\identity}\leq0\) [use \cref{lem:desirs:from:lowprev}] and therefore \(\lowprev[\desirset]\group{\measurement{A}}\leq-\nicefrac\epsilon2<0\) [use~\labelcref{ax:lower:prevision:constant:additivity}].
\end{proof}

\begin{lemma}\label{lem:desirs:from:lowprev}
Consider any coherent {\SDM}~\(\desirset\) and let~\(\lowprev[\desirset]\) be the corresponding price functional.
Then, for any~\(\measurement{A}\in\measurements\), we have that \(\measurement{A}\notin\desirset\then\lowprev[\desirset]\group{\measurement{A}}\leq0\).
\end{lemma}

\begin{proof}
That \(\measurement{A}\notin\desirset\), or in other words, that \(\measurement{A}-0\identity\notin\desirset\), implies that \(0\) dominates all elements of the decreasing set~\(L_{\measurement{A}}=\cset{\alpha\in\reals}{\measurement{A}-\alpha\identity\in\desirset}\) and therefore dominates its supremum~\(\lowprev[\desirset]\group{\measurement{A}}\).
\end{proof}

\noindent \cref{eq:desirs:lowprev} tells us that all measurements~\(\measurement{A}\) with a positive supremum buying price~\(\lowprev[\desirset]\group{\measurement{A}}\) lie definitely inside~\(\desirset\), and those with a negative supremum buying price definitely outside~\(\desirset\).
It follows at once from \cref{eq:desirs:lowprev} that
\[
\measurement{A}\in\closure(\desirset)\setminus\interior(\desirset)
\ifandonlyif\lowprev[\desirset]\group{\measurement{A}}=0,
\text{ for all~\(\measurement{A}\in\measurements\)},
\]
so for the so-called \emph{marginally desirable} measurements in the topological boundary of the convex cone~\(\desirset\), which are the measurements with supremum buying price zero, their supremum buying price typically won't determine whether they belong to~\(\desirset\), unless~\(\desirset\) is known to satisfy additional properties, besides coherence.\footnote{\dots{} such as, for instance, being open.}

The upshot of all this is that the price functional~\(\lowprev[\desirset]\) characterises the {\SDM}~\(\desirset\) it's derived from \emph{only up to the measurements in its topological boundary}.
Coherent lower previsions therefore have at most as much, and typically (slightly) less, representational power than coherent sets of desirable measurements: the correspondence between coherent sets of desirable measurements and coherent lower previsions is typically many-to-one.
As is apparent from \cref{eq:desirs:from:lowprev:basis} in the proof above that investigates the correspondence between the two models, one way to make a coherent lower prevision~\(\lowprev\) correspond to exactly one strict preference~\(\betterthan\), is to pick the one given by
\begin{equation}\label{eq:preference:and:lower:prevision}
\measurement{A}\betterthan\measurement{B}
\ifandonlyif\measurement{A}-\measurement{B}\in\desirset[\lowprev]
\ifandonlyif\group[\big]{\lowprev\group{\measurement{A}-\measurement{B}}>0\text{ or }\measurement{A}\alwaysbetterthan\measurement{B}}
\text{ for all~\(\measurement{A},\measurement{B}\in\measurements\)}.
\end{equation}

To summarise, using a coherent lower prevision~\(\lowprev\) to describe Your beliefs about the system state amounts to specifying a supremum buying price~\(\lowprev\group{\measurement{A}}\) for all possible measurements~\(\measurement{A}\in\measurements\).
That it satisfies the coherence requirements~\labelcref{ax:lower:prevision:bounds,ax:lower:prevision:superadditivity,ax:lower:prevision:scaling} guarantees that there's always at least one, but typically more than one, coherent set of desirable measurements~\(\desirset\) that corresponds to it in the sense that \(\lowprev[\desirset]=\lowprev\).
All these sets differ only in their topological boundary, which is the set of marginally desirable measurements~\(\cset{\measurement{A}\in\measurements}{\lowprev\group{\measurement{A}}=0}\).

To conclude this discussion of coherent lower and upper previsions, let's consider a few special cases.
At one extreme end, we have the price functionals for the vacuous {\SDM}~\(\desirset[\mathrm{vac}]=\posmeasurements\), which are the so-called \emph{vacuous} lower and upper previsions, defined by
\begin{equation*}
\lowprev[\mathrm{vac}]\group{\measurement{A}}=\min\spectrum{A}
\text{ and }
\uppprev[\mathrm{vac}]\group{\measurement{A}}=\max\spectrum{A}
\text{ for all~\(\measurement{A}\in\measurements\)}.
\end{equation*}
Here, the \emph{bid-ask spread}~\(\uppprev[\mathrm{vac}]\group{\measurement{A}}-\lowprev[\mathrm{vac}]\group{\measurement{A}}=\max\spectrum{A}-\min\spectrum{A}\) between Your selling and buying prices for measurements~\(\measurement{A}\) is as large as allowed by the coherence requirements [in particular~\labelcref{ax:lower:prevision:bounds}].

At the other extreme end, we recover the cases that are most often considered in more classical approaches to decision theory, where the bid-ask spread~\(\uppprev\group{\measurement{A}}-\lowprev\group{\measurement{A}}\) between Your infimum selling and supremum buying prices for measurements~\(\measurement{A}\) is as small as it's allowed to be by the coherence requirements [in particular~\labelcref{ax:lower:prevision:more:bounds}], namely \(0\) everywhere.
This very special subclass of the coherent lower previsions is the type of uncertainty model we'll now turn to in the coming section.

\section{Coherent previsions and density operators}\label{sec:coherent:previsions}

\subsection{Bruno de Finetti's coherent previsions}\label{sec:finetti}
Coherent lower previsions \(\lowprev\) with the smallest allowable bid-ask spread are \emph{self-conjugate} in the sense that \(\lowprev=\uppprev\), so we can and will use the simpler notation~\(\linprev\) for~\(\lowprev=\uppprev\) in this special case.
They're at the same time coherent lower previsions and coherent upper previsions.

According to the discussion and definitions in the previous section, \(\linprev\group{\measurement{A}}\) is then at the same time Your supremum price for buying the uncertain reward~\(\theutility{A}(\uket)\) and Your infimum price for selling it, and can therefore also be seen as \emph{Your fair price} for the uncertain reward~\(\theutility{A}(\uket)\).

In conjunction with self-conjugacy, the coherence conditions~\labelcref{ax:lower:prevision:bounds,ax:lower:prevision:superadditivity,ax:lower:prevision:scaling} turn into
\begin{enumerate}[label={\upshape CP\arabic*.},ref={\upshape CP\arabic*},leftmargin=*]
\item\label{ax:linprev:bounds} \(\linprev\group{\measurement{A}}\geq\min\spectrum{A}\) for all~\(\measurement{A}\in\measurements\);\hfill[bounds]
\item\label{ax:linprev:additivity} \(\linprev\group{\measurement{A}+\measurement{B}}=\linprev\group{\measurement{A}}+\linprev\group{\measurement{B}}\) for all~\(\measurement{A},\measurement{B}\in\measurements\).\hfill[additivity]
\end{enumerate}
Homogeneity then follows from these two coherence conditions:
\begin{enumerate}[label={\upshape CP\arabic*}.,ref={\upshape CP\arabic*},labelwidth=*,leftmargin=*,resume]
\item\label{ax:linprev:homogeneity} \(\linprev\group{\lambda\measurement{A}}=\lambda\linprev\group{\measurement{A}}\) for all~\(\measurement{A}\in\measurements\) and~\(\lambda\in\reals\).\hfill[homogeneity]
\end{enumerate}

\begin{proof}
First, \labelcref{ax:linprev:additivity} implies that \labelcref{ax:linprev:homogeneity} holds for all rational~\(\lambda\).
It's clear that we can extend the equality to the reals, provided we can show that \labelcref{ax:linprev:bounds,ax:linprev:additivity} imply the continuity of~\(\linprev\).
To this end, observe that \labelcref{ax:linprev:additivity} implies in particular that \(\linprev\group{\zero}=0\) and \(\linprev\group{-\measurement{A}}=-\linprev\group{\measurement{A}}\) for all~\(\measurement{A}\in\measurements\).
But then also
\begin{align*}
\linprev\group{\measurement{A}}-\linprev\group{\measurement{B}}
&=\linprev\group{\measurement{A}}+\linprev\group{-\measurement{B}}\\
&=\linprev\group{\measurement{A}-\measurement{B}}=-\linprev\group{-\group{\measurement{A}-\measurement{B}}}\\
&\leq\max\spec\group{\measurement{A}-\measurement{B}}
\leq\supnorm{\measurement{A}-\measurement{B}},
\end{align*}
where the second equality follows from~\labelcref{ax:linprev:additivity}, the first inequality follows from~\labelcref{ax:linprev:bounds}, and the second inequality from \cref{eq:supnorm}.
Similarly, \(\linprev\group{\measurement{B}}-\linprev\group{\measurement{A}}\leq\supnorm{\measurement{A}-\measurement{B}}\), so
\begin{equation*}
\abs{\linprev\group{\measurement{A}}-\linprev\group{\measurement{B}}}
\leq\supnorm{\measurement{A}-\measurement{B}}
\text{ for all~\(\measurement{A},\measurement{B}\in\measurements\)},
\end{equation*}
which even proves the Lipschitz continuity of~\(\linprev\).

It's now clearly enough to prove that we can derive \labelcref{ax:linprev:bounds,ax:linprev:additivity} from the conditions~\labelcref{ax:lower:prevision:bounds,ax:lower:prevision:superadditivity,ax:lower:prevision:scaling} and self-conjugacy, which tells us that \(\linprev=\lowprev=\uppprev\).
\labelcref{ax:linprev:bounds} is clearly equivalent to \labelcref{ax:lower:prevision:bounds}.
For \labelcref{ax:linprev:additivity}, consider any~\(\measurement{A},\measurement{B}\in\measurements\).
Self-conjugacy implies that \(\lowprev\group{\measurement{A}}+\lowprev\group{\measurement{B}}=\uppprev\group{\measurement{A}}+\uppprev\group{\measurement{B}}\).
On the other hand, it follows from \labelcref{ax:lower:prevision:mixed:additivity} [which follows from \labelcref{ax:lower:prevision:superadditivity}] that \(\lowprev\group{\measurement{A}}+\lowprev\group{\measurement{B}}\leq\lowprev\group{\measurement{A}+\measurement{B}}\leq\uppprev\group{\measurement{A}+\measurement{B}}\leq\uppprev\group{\measurement{A}}+\uppprev\group{\measurement{B}}\) and therefore, indeed, that \(\lowprev\group{\measurement{A}}+\lowprev\group{\measurement{B}}=\lowprev\group{\measurement{A}+\measurement{B}}\).
\end{proof}

We'll call any self-conjugate coherent lower prevision a \emph{coherent prevision}, and denote the set of all such coherent previsions on~\(\measurements\) by~\(\linprevs\).
The coherent previsions are therefore the real functionals on~\(\measurements\) that satisfy the coherence conditions~\labelcref{ax:linprev:additivity,ax:linprev:bounds}.
It's easy to see that they can also be characterised exactly as the real functionals~\(\linprev\colon\measurements\to\reals\) that are (i) \emph{linear} in the sense that
\begin{enumerate}[label={\upshape L\arabic*}.,ref={\upshape L\arabic*},leftmargin=*]
\item\label{ax:prevision:linear} \(\linprev\group{\lambda\measurement{A}+\mu\measurement{B}}=\lambda\linprev\group{\measurement{A}}+\mu\linprev\group{\measurement{B}}\) for all~\(\measurement{A},\measurement{B}\in\measurements\) and all~\(\lambda,\mu\in\reals\);
\end{enumerate}
(ii) \emph{positive} in the sense that
\begin{enumerate}[label={\upshape L\arabic*}.,ref={\upshape L\arabic*},leftmargin=*,resume]
\item\label{ax:prevision:positive} \(\linprev\group{\measurement{A}}\geq0\) for all~\(\measurement{A}\geq\zero\);
\end{enumerate}
and (iii) \emph{normalised} in the sense that
\begin{enumerate}[label={\upshape L\arabic*}.,ref={\upshape L\arabic*},leftmargin=*,resume]
\item\label{ax:prevision:norm} \(\linprev\group{\identity}=1\).
\end{enumerate}
\begin{proof}
Clearly, \labelcref{ax:prevision:linear} is equivalent to \labelcref{ax:linprev:additivity,ax:linprev:homogeneity}. Condition \labelcref{ax:prevision:positive} is an immediate consequence of \labelcref{ax:linprev:bounds}.
Applying \labelcref{ax:linprev:bounds} to~\(\identity\) and \(-\identity\) leads to \labelcref{ax:prevision:norm}.
It now only remains to prove that \labelcref{ax:linprev:bounds} follows from \labelcref{ax:prevision:linear,ax:prevision:positive,ax:prevision:norm}.
Consider, to this end, any measurement~\(\measurement{A}\) and let \(c\coloneqq\min\spectrum{A}\), then \(\measurement{A}-c\identity\possemidef\zero\) and therefore \(\linprev\group{\measurement{A}-c\identity}\geq0\), due to \labelcref{ax:prevision:positive}.
\labelcref{ax:prevision:linear,ax:prevision:norm} now ensure that \(\linprev\group{\measurement{A}-c\identity}=\linprev\group{\measurement{A}}-c\linprev\group{\identity}=\linprev\group{\measurement{A}}-c\).
\end{proof}

\noindent If we follow the standard functional-analytic approach in defining the following operator norm
\[
\opnorm{\Gamma}
\coloneqq\sup_{\measurement{A}\neq\zero}\frac{\abs[\big]{\Gamma(\measurement{A})}}{\supnorm{\measurement{A}}}
\]
on (linear) real functionals~\(\Gamma\colon\measurements\to\reals\), then it follows from \cref{eq:supnorm}, \labelcref{ax:lower:prevision:more:bounds} and \labelcref{ax:prevision:norm} that
\begin{equation}\label{eq:boundedness}
\opnorm{\linprev}=1\text{ for all self-conjugate coherent lower previsions~\(\linprev\)};
\end{equation}
they're bounded and therefore continuous, as also shown directly in the proof of \labelcref{ax:linprev:homogeneity} above.

For all these reasons, our coherent previsions can also be seen as coherent previsions or fair prices in Bruno de Finetti's sense; see for instance Refs.~\cite{finetti1937} and \cite[Ch.~3 and App.]{finetti19745}.
Coherent previsions on real Hilbert spaces are indeed the probability models envisaged by de Finetti, but we must caution against interpreting them without further ado as expectation operators associated with a (finitely additive) probability on some sample space, as this is an interpretational step that asks for further justification, which isn't provided by our argumentation above; see also the discussion further on in \cref{sec:a:warning}.

Interestingly, we can use the coherent previsions, or rather, sets of them, to characterise the coherent \emph{lower} previsions, of which they're special cases.
Taking the lower envelope of \emph{any} non-empty set of coherent previsions~\(\solp\subseteq\linprevs\) leads to a coherent lower prevision~\(\lowprev[\solp]\group{\bolleke}\coloneqq\inf\cset{\linprev\group{\bolleke}}{\linprev\in\solp}\).
But, even stronger than this, it's a consequence of the Hahn--Banach theorem (or, in this finite-dimensional context, of Minkowski's hyperplane separation theorem) that any coherent lower prevision~\(\lowprev\) is the lower envelope of a (non-empty) \emph{convex and closed}\footnote{\dots\ in the (so-called weak\(^\star\)) topology of pointwise convergence on the set of all linear functionals on the linear space~\(\measurements\).} set of coherent previsions~\(\linprevs(\lowprev)\), uniquely defined by
\begin{equation}\label{eq:solp:from:lowprev}
\linprevs(\lowprev)
\coloneqq\cset{\linprev\in\linprevs}{\group{\forall\measurement{A}\in\measurements}\lowprev\group{\measurement{A}}\leq\linprev\group{\measurement{A}}},
\end{equation}
in the sense that
\begin{equation}\label{eq:lowprev:from:solp}
\lowprev\group{\measurement{A}}
=\min\cset[\big]{\linprev\group{\measurement{A}}}{\linprev\in\linprevs(\lowprev)}
\text{ for all~\(\measurement{A}\in\measurements\)}.
\end{equation}
This result, also known as the \emph{Lower Envelope Theorem}, guarantees that there's a one-to-one correspondence between coherent lower previsions and non-empty convex closed sets of coherent previsions.
Following Levi's \cite{levi1980a} terminology, we'll call the convex closed set~\(\linprevs(\lowprev)\) the \emph{credal set} that corresponds to, and characterises, the coherent lower prevision~\(\lowprev\).
We refer to Ref.~\cite[Cor.~13]{cooman2021:archimedean:choice} for a more detailed argumentation, or alternatively to Refs.~\cite[Ch.~3]{walley1991} and~\cite[Ch.~8]{troffaes2013:lp} for proofs in the context of gambles that can be readily transported to the present context.
See also Ref.~\cite[Prop.~IV.3]{benavoli2016:quantum}, where \citeauthor{benavoli2016:quantum} are the first to introduce credal sets in the quantum mechanical context, with a somewhat different interpretation.

\begin{example}\label{example:coherent:previsions}
Let's return to our qubit system and consider an arbitrary coherent prevision~\(\linprev\) on~\(\measurements\).
Since any Hermitian operator~\(\measurement{C}\in\measurements\) can be written as \(\measurement{C}=w\identity+x\paulix+y\pauliy+z\pauliz\) for some unique~\((w,x,y,z)\in\reals^4\), we find that, due to the linear character of~\(\linprev\) [coherence property~\labelcref{ax:prevision:linear}],
\begin{equation}\label{eq:qubit:linprev:form}
\linprev\group{\measurement{C}}
=w+x\linprev\group{\paulix}+y\linprev\group{\pauliy}+z\linprev\group{\pauliz},
\end{equation}
where we also took into account that \(\linprev\group{\identity}=1\), by coherence property~\labelcref{ax:prevision:norm}.
We see that the coherent prevision~\(\linprev\) is therefore completely determined by its values~\(\linprev[x]\coloneqq\linprev\group{\paulix}\), \(\linprev[y]\coloneqq\linprev\group{\pauliy}\) and \(\linprev[z]\coloneqq\linprev\group{\pauliz}\) in the Pauli operators.
We can find the constraints imposed by coherence on these real numbers by looking at the coherence condition~\labelcref{ax:lower:prevision:more:bounds}, which, applied to any~\(\measurement{C}\), tells us that \(\min\spectrum{C}\leq\linprev\group{\measurement{C}}\leq\max\spectrum{C}\).
If we now recall~\cref{eq:qubit:eigenvalues} for the eigenvalues of~\(\measurement{C}\), we see that this can be rewritten as
\[
w-\sqrt{x^2+y^2+z^2}
\leq w+x\linprev[x]+y\linprev[y]+z\linprev[z]
\leq w+\sqrt{x^2+y^2+z^2},
\]
or equivalently, as
\[
\group{x\linprev[x]+y\linprev[y]+z\linprev[z]}^2\leq x^2+y^2+z^2,
\text{ for all~\(x,y,z\in\reals\)}.
\]
Combining this with the Cauchy--Schwarz inequality, which can be rewritten as
\[
\linprev[x]^2+\linprev[y]^2+\linprev[z]^2
=\max_{(x,y,z)\in\reals^3\setminus\set{(0,0,0)}}
\frac{\group{x\linprev[x]+y\linprev[y]+z\linprev[z]}^2}{x^2+y^2+z^2},
\]
leads to
\begin{equation}\label{eq:qubit:linprev:bounds}
\linprev[x]^2+\linprev[y]^2+\linprev[z]^2\leq1,
\end{equation}
which, taken together with~\(\linprev\group{\identity}=1\), is all that coherence imposes on the values of a coherent prevision~\(\linprev\) in the Pauli operators, and therefore on its values in any~\(\measurement{C}\in\measurements\) through~\cref{eq:qubit:linprev:form}.

As we've seen in Instalment~\labelcref{example:lower:and:upper:previsions}, the assessment \(\assessment\coloneqq\set{\paulix}\) corresponds to a lower prevision~\(\lowprev[\ext\group{\assessment} ]\).
Taking into account \cref{eq:qubit:lower:upper:pauli}, we find that a coherent prevision~\(\linprev\) belongs to the credal set~\(\linprevs(\lowprev[\ext\group{\assessment} ])\) that corresponds to this lower prevision when
\begin{equation*}
\begin{cases}
w-\sqrt{y^2+z^2}&\text{if \(x\geq0\)}\\
w-\sqrt{x^2+y^2+z^2}&\text{if \(x\leq0\)}
\end{cases}
\leq w+x\linprev[x]+y\linprev[y]+z\linprev[z]
\text{ for all~\((x,y,y,z)\in\reals^4\)},
\end{equation*}
and \(\linprev[x]^2+\linprev[y]^2+\linprev[z]^2\leq1\) [these are the general coherence bounds established in \cref{eq:qubit:linprev:bounds}].
After some algebra, we find that the credal set~\(\linprevs(\lowprev[\ext\group{\assessment} ])\) is completely determined by the conditions
\begin{equation}\label{eq:qubit:linprevs:constraints}
\linprev[x]\geq0
\text{ and }
\linprev[x]^2+\linprev[y]^2+\linprev[z]^2\leq1.
\end{equation}
We then infer from~\cref{eq:lowprev:from:solp} and conjugacy that, to find \(\lowprev[\ext\group{\assessment} ]\group{\measurement{C}}\) and \(\uppprev[\ext\group{\assessment} ]\group{\measurement{C}}\), we have to minimise respectively maximise \(w+x\linprev[x]+y\linprev[y]+z\linprev[z]\) subject to the constraints in \cref{eq:qubit:linprevs:constraints}.
Some small effort shows that this leads to the same values as those given in \cref{eq:qubit:lower:upper:pauli:special:cases} for the particular choices \(\paulix\), \(\pauliy\) and \(\paulix-\pauliy\) for~\(\measurement{C}\).
These calculations correspond to solving \emph{dual} semidefinite programming problems for the ones mentioned near the end of Instalment~\labelcref{example:lower:and:upper:previsions} of our running example.
\end{example}

\subsection{Density operators}\label{sec:density:operators}
In the more standard treatment of uncertainty in quantum mechanics, as discussed in \cref{sec:quantum:mechanics:probabilistic}, uncertainty about the state is often expressed through density operators.
For a system in a mixed state, described by such a density operator~\(\density\), the \emph{expected outcome} of the measurement~\(\measurement{A}\) on the system was found to be \(\expec[\density]{\measurement{A}}=\trace{\density\measurement{A}}\); see \cref{eq:qm:expectation}.
This formalism doesn't come, however, with a clear view on how to interpret the probabilities and expected values that play a role in it.

As it's one of our aims in this paper to provide a decision-theoretic interpretation for such density operators and expected outcomes, by relating them to our coherent lower previsions and {\SDMs}, we'll now show that the special subclass of the coherent previsions on~\(\measurements\) provide us with an appropriate means of connecting density operators with our decision-theoretic {\SDM} approach.
To do so, we'll again follow the path that was cleared by \citeauthor{benavoli2016:quantum} in Ref.~\cite[Sec.~IV]{benavoli2016:quantum}, by focusing on the mathematical links between coherent previsions and density operators.
We postpone a discussion of how to interpret these links to \cref{sec:a:warning}.

We begin this brief mathematical exploration with the observation that expected outcomes make for perfectly acceptable coherent previsions.

\begin{proposition}\label{prop:from:density:to:linprev}
For every density operator~\(\density\in\densities\), the real functional~\(\linprev[\density]\group{\bolleke}\coloneqq\expec[\density]{\bolleke}\coloneqq\trace{\density\,\bolleke}\) is a coherent prevision on~\(\measurements\).
\end{proposition}

\begin{proof}
The additivity condition~\labelcref{ax:linprev:additivity} follows at once from the linearity of the trace, so it only remains to prove the bounds condition~\labelcref{ax:linprev:bounds}.
Since~\(\density\) is a density operator, \cref{prop:qm:density} tells us that there are kets~\(\gket[1]\), \dots, \(\gket[q]\) in~\(\statespace\) and real numbers~\(\con_1,\con_2,\dots,\con_q\in[0,1]\) such that \(\sum_{k=1}^q\con_k=1\) and \(\density=\sum_{k=1}^q\con_k\gket[k]\gbra[k]\).
For every~\(\measurement{A}\in\measurements\), we then find that, indeed,
\[
\trace{\density\measurement{A}}
=\sum_{k=1}^q\con_k\trace{\gket[k]\gbra[k]\measurement{A}}
=\sum_{k=1}^q\con_k\gbra[k]\measurement{A}\gket[k]
\geq\min\spectrum{A},
\]
where the first equality follows from the linearity of the trace, the second equality from its cyclic character, and the inequality from~\Cref{lem:quadratic:form:in:basis} and the fact that \(\sum_{k=1}^q\con_k=1\).
\end{proof}

To prove a converse result, namely, that for every coherent prevision~\(\linprev\) there's a (unique) density operator~\(\density\) such that \(\linprev\group{\bolleke}=\trace{\density\,\bolleke}\), we need a bit of mathematical background on Hilbert spaces.
We've already had occasion to mention that we can provide the real linear space~\(\measurements\) with the Frobenius inner product, defined by \(\inprod{\measurement{A}}{\measurement{B}}\coloneqq\trace{\measurement{A}\measurement{B}}\) for all~\(\measurement{A},\measurement{B}\in\measurements\), which turns \(\measurements\) into a real Hilbert space.
We can then apply to this Hilbert space~\(\measurements\) the Riesz representation theorem, which states that continuous --- or equivalently, bounded --- linear functionals on a Hilbert space are in a one-to-one relationship with its vectors.

\begin{lemma}[Riesz representation theorem {\protect\cite[Thm.~II.4]{reed2003:methods:1}}]\label{lem:riesz}
For any continuous linear functional~\(\Gamma\) on~\(\measurements\), there's a \emph{unique}~\(\indmeasurement{B}{\Gamma}\in\measurements\) such that \(\Gamma(\measurement{A})=\trace{\indmeasurement{B}{\Gamma}\measurement{A}}\) for all~\(\measurement{A}\in\measurements\).
\end{lemma}

We now have the tools to prove the desired one-to-one relationship; see also the discussion in Section~IV of \cite{benavoli2016:quantum}, where Proposition~IV.6 and Theorem~IV.4 turn out to be closely related to what we have here.

\begin{proposition}\label{prop:one-to-one}
There's a one-to-one correspondence~\(\todensity\) between coherent previsions~\(\linprev\) on~\(\measurements\) and density operators~\(\density\) on~\(\hilbertspace\), through~\(\linprev\group{\bolleke}=\trace{\density\,\bolleke}\) with \(\density=\todensity(\linprev)\).
This bijection preserves convex combinations, in the sense that \(\todensity(\alpha\linprev[1]+(1-\alpha)\linprev[2])=\alpha\todensity(\linprev[1])+(1-\alpha)\todensity(\linprev[2])\), for all real~\(\alpha\in[0,1]\) and all~\(\linprev[1],\linprev[2]\in\linprevs\).
\end{proposition}

\begin{proof}
We begin with the first statement.
By~\cref{prop:from:density:to:linprev}, it suffices to prove that with any coherent prevision~\(\linprev\in\linprevs\) there corresponds a unique density operator~\(\density\) such that \(\linprev\group{\bolleke}=\trace{\density\,\bolleke}\).
We can then denote this~\(\density\) by~\(\todensity(\linprev)\).
We know that \(\linprev\) is bounded and therefore continuous, from the discussion leading up to~\cref{eq:boundedness}, so we can infer from \cref{lem:riesz} that there's a unique~\(\indmeasurement{B}{\linprev}\in\measurements\) such that \(\linprev\group{\bolleke}=\trace{\indmeasurement{B}{\linprev}\bolleke}\).
It's enough, therefore, to prove that the Hermitian operator~\(\indmeasurement{B}{\linprev}\) is a density operator, or in other words, that \(\indmeasurement{B}{\linprev}\geq0\) and that \(\trace{\indmeasurement{B}{\linprev}}=1\).

First, observe that \(\trace[\big]{\indmeasurement{B}{\linprev}}=\trace[\big]{\indmeasurement{B}{\linprev}\identity}=\linprev\group{\identity}=1\), where the last equality follows from~\labelcref{ax:prevision:norm}.

Next, consider any~\(\gket\in\statespace\).
Then for the projection~\(\projection[\gket]\coloneqq\gket\gbra\in\measurements\), we find that \(\spectrumpure{\projection[\gket]}\subseteq\set{0,1}\) and therefore, by the bounds condition~\labelcref{ax:linprev:bounds}, that \(\linprev\group{\projection[\gket]}\geq0\).
As the trace is cyclic, this leads to~\(0\leq\linprev\group[\big]{\projection[\gket]}=\trace[\big]{\indmeasurement{B}{\linprev}\projection[\gket]}=\trace[\big]{\indmeasurement{B}{\linprev}\gket\gbra}=\gbra\indmeasurement{B}{\linprev}\gket\), which implies that, indeed, \(\indmeasurement{B}{\linprev}\geq0\).

To prove that the second statement, consider any real~\(\alpha\in[0,1]\) and any~\(\linprev[1],\linprev[2]\in\linprevs\), and let \(\density[1]\coloneqq\todensity(\linprev[1])\) and \(\density[2]\coloneqq\todensity(\linprev[2])\).
And, by the linearity of the trace,
\begin{align*}
\trace[\big]{\group[\big]{\alpha\density[1]+(1-\alpha)\density[2]}\bolleke}
&=\alpha\trace[\big]{\density[1]\bolleke}+(1-\alpha)\trace[\big]{\density[2]\bolleke}\\
&=\alpha\linprev[1](\bolleke)+(1-\alpha)\linprev[2](\bolleke)
=\group[\big]{\alpha\linprev[1]+(1-\alpha)\linprev[2]}(\bolleke).\qedhere
\end{align*}
\end{proof}

More generally, when Your beliefs about~\(\uket\) are modelled by a coherent lower prevision~\(\lowprev\) that isn't self-conjugate, then we've already established that this~\(\lowprev\) can be written as the lower envelope of the convex closed set of coherent previsions, or credal set, \(\linprevs(\lowprev)\).
Because it preserves convex combinations, the map~\(\todensity\) turns this set into the \emph{convex and closed\footnote{\dots\ in the topology on~\(\densities\) induced by the map~\(\todensity\), characterised by the fact that a sequence of densities~\(\density[n]\) converges to the density~\(\density\) if and only if corresponding sequence of linear functionals~\(\trace{\density[n]\bullet}\) converges (pointwise, or equivalently, in norm, due to the finite dimension of the underlying Hilbert space~\(
\hilbertspace\)) to the linear functional~\(\trace{\density\bullet}\).} set of density operators}
\begin{align}
\densities(\lowprev)
&\coloneqq\todensity\group{\linprevs(\lowprev)}
=\cset{\todensity(\linprev)}{\linprev\in\linprevs(\lowprev)}\notag\\
&=\cset[big]{\density\in\densities}
{\group{\forall\measurement{A}\in\measurements}\trace{\density\measurement{A}}\geq\lowprev\group{\measurement{A}}},
\label{eq:from:lowprev:to:set:of:densities}
\end{align}
and, clearly, this set of density operators completely determines the coherent lower prevision~\(\lowprev\) in the sense that
\begin{equation}\label{eq:from:set:of:densities:to:lowprev}
\lowprev\group{\measurement{A}}
=\min\cset{\trace{\density\measurement{A}}}{\density\in\densities(\lowprev)}
\text{ for all~\(\measurement{A}\in\measurements\)}.
\end{equation}
This is the generalisation of Born's Rule, in its most general version~\eqref{eq:qm:expectation}, that we derive from our decision-theoretic argumentation, where the single density operator~\(\density\) is replaced by a convex closed set~\(\densities(\lowprev)\) of them, and where the expected outcome~\(\expec[\density]{\measurement{A}}\) of a measurement~\(\measurement{A}\) is replaced its lower prevision~\(\lowprev\group{\measurement{A}}\).

One meaning that our decision-theoretic approach attributes to this lower prevision~\(\lowprev\) is, on the one hand, that it characterises (up to border behaviour, see the previous section) Your preferences between the acts~\(\act{A}\) associated with measurements~\(\measurement{A}\).
On the other hand, we also recall that \(\lowprev(\measurement{A})\) was defined as Your supremum price for buying the uncertain reward~\(\theutility{A}(\uket)\) associated with performing the measurement~\(\measurement{A}\) on the quantum system in its unknown state~\(\uket\).
And finally, if we were to allow ourselves to follow a more standard approach to probability in quantum mechanics, \(\lowprev\group{\measurement{A}}\) could be seen as a tight lower bound on the expected values of the measurement~\(\measurement{A}\), as is indicated by \cref{eq:from:set:of:densities:to:lowprev}.

We've summarised the relationship between the various types of models introduced and described in this paper in \cref{fig:models}.

\begin{figure}[h]
\centering\small
\def\colsep{8mm}
\def\rowsep{8mm}
\def\boxwidth{3.25cm}
\def\boxheight{1cm}
\def\boxinnersep{3pt}
\def\boxrounding{2pt}
\begin{tikzpicture}[scale=1, every node/.style={transform shape}]
\tikzset{block/.style={draw,rectangle,rounded corners=\boxrounding,align=center,text width=\boxwidth,minimum height=\boxheight,inner sep=\boxinnersep},
emptyblock/.style={draw=white,rectangle,rounded corners=\boxrounding,align=center,text width=\boxwidth,minimum height=\boxheight,inner sep=\boxinnersep},
conn/.style={-Latex},
bconn/.style={<->, >=Latex}
}
\node[block] (B1) {Coherent sets~\(\desirset\) of desirable measurements satisfying \labelcref{ax:desirability:strict,ax:desirability:additivity,ax:desirability:scaling,ax:desirability:background}};
\node[emptyblock, right=\colsep of B1] (C2R1) {}; 
\node[block, right=\colsep of C2R1] (B2) {Strict vector orderings~\(\betterthan\) satisfying \labelcref{ax:preference:irreflexitivy,ax:preference:transitivity,ax:preference:scaling,ax:preference:background,ax:preference:additivity}};
\node[draw=none, minimum height=0pt, minimum width=0pt, inner sep=0pt] at (C2R1) {}; 
\node[block, below=\rowsep of B1] (B3) {Coherent lower previsions~\(\lowprev\) satisfying \labelcref{ax:lower:prevision:bounds,ax:lower:prevision:superadditivity,ax:lower:prevision:scaling}};
\node[block, right=\colsep of B3] (B4) {Convex closed sets~\(\solp\) of coherent previsions};
\node[block, right=\colsep of B4] (B5) {Convex closed sets~\(\densities\) of density operators};
\node[block, below=\rowsep of B3] (B6) {Coherent previsions~\(\linprev\) satisfying \labelcref{ax:linprev:additivity,ax:linprev:bounds}};
\node[emptyblock, right=\colsep of B6] (C2R3) {}; 
\node[block, right=\colsep of C2R3] (B7) {Density operators~\(\density\)};
\node[draw=none, minimum height=0pt, minimum width=0pt, inner sep=0pt] at (C2R3) {}; 
\draw[bconn] (B1) -- (B2);
\draw[conn]  (B1) -- (B3);
\draw[bconn] (B3) -- (B4);
\draw[bconn] (B4) -- (B5);
\draw[conn,densely dashed]  (B3) -- (B6);
\draw[conn,densely dashed]  (B5) -- (B7);
\draw[bconn] (B6) -- (B7);
\end{tikzpicture}
\caption{The relationship between the various types of models for uncertainty about the quantum state in this paper; double arrows indicate a one-to-one correspondence; a solid single arrow indicates a many-to one correspondence, in the sense that each member of the origin set corresponds to a single member of the target set, but a member of the target set may correspond to multiple members of the origin set; a dashed single arrow points towards a subclass, in the sense that all the members in the target set correspond to members of the origin set, but not vice versa. The top row represents the most general models and the bottom row the least general ones, which are also the ones typically used in quantum mechanics to date.}
\label{fig:models}
\end{figure}

\begin{example}\label{example:density:operators}
Let's now find out how to apply these ideas in our running example.

First, let's consider a density operator~\(\density\).
Since \(\density\) is in particular Hermitian, we know that there are~\(a,b,c,d\in\reals\) such that \(\density=d\identity+a\paulix+b\pauliy+c\pauliz\); its eigenvalues are then \(d\pm\sqrt{a^2+b^2+c^2}\), due to~\cref{eq:qubit:eigenvalues}.
Expressing that \(\trace{\density}=1\) and \(\density\geq0\) therefore leads to the conditions:
\begin{equation}\label{eq:qubit:density:constraints}
d=\frac12\text{ and }a^2+b^2+c^2\leq\frac14.
\end{equation}
Now consider any measurement~\(\measurement{C}\in\measurements\) and let, as before, \(\measurement{C}=w\identity+x\paulix+y\pauliy+z\pauliz\), with \(x,y,z,w\in\reals\).
Another application of the argumentation behind~\cref{eq:qubit:eigenvalues} tells us that \(\trace{\density\measurement{C}}\) is twice the \(\identity\)-component of~\(\density\measurement{C}\).
Moreover, by the properties of Pauli operators \cite[p.~418]{cohen1977:quantum:1}, we find that
\begin{align*}
\density\measurement{C}
&=\group[\Big]{\frac12\identity+a\paulix+b\pauliy+c\pauliz}\group{w\identity+x\paulix+y\pauliy+z\pauliz}\\
&=\group[\Big]{\frac12w+ax+by+cz}\identity+\group{\dots}\paulix+\group{\dots}\pauliy+\group{\dots}\pauliz
\end{align*}
and therefore \(\trace{\density\measurement{C}}=w+2(ax+by+cz)\).
If we now consider the coherent prevision~\(\linprev\) that corresponds to~\(\density\) in the sense that \(\linprev\group{\bolleke}=\trace{\density\,\bolleke}\), and invoke \cref{eq:qubit:linprev:form}, we find that
\[
w+2(ax+by+cz)
=w+x\linprev[x]+y\linprev[y]+z\linprev[z]
\text{ for all~\(w,x,y,z\in\reals\)}.
\]
This tells us that \(\linprev[x]=2a\), \(\linprev[y]=2b\) and \(\linprev[z]=2c\), and therefore
\begin{equation}\label{eq:qubit:from:linprev:to:density}
\density=\frac12\identity+\frac{\linprev[x]\paulix+\linprev[y]\pauliy+\linprev[z]\pauliz}2,
\end{equation}
while the constraints~\eqref{eq:qubit:density:constraints} turn into~\(\linprev[x]^2+\linprev[y]^2+\linprev[z]^2\leq1\), which is in complete accordance with~\cref{eq:qubit:linprev:bounds}.

The elements~\(\density\) of the closed convex set~\(\densities(\lowprev[\ext\group{\assessment} ])\) are therefore characterised by
\[
\density=\frac12\identity+\frac{\linprev[x]\paulix+\linprev[y]\pauliy+\linprev[z]\pauliz}2
\text{ with }
\linprev[x]\geq0
\text{ and }
\linprev[x]^2+\linprev[y]^2+\linprev[z]^2\leq1.
\qedhere
\]
\end{example}

\subsection{A few words of caution \dots\ and a curious observation}\label{sec:a:warning}
It's important to stress again that the version of Born's rule~\eqref{eq:from:set:of:densities:to:lowprev} that we've derived from our assumptions, doesn't start from a probability argument, but rather from a preference argument, which then leads to coherent lower previsions, or sets of coherent previsions.
We feel this distinction to be quite important, and we'll use this and the next section to clarify our point.

The lower previsions~\(\lowprev\) that we've introduced here, are functionals that serve as an alternative characterisation for Your preferences, in the sense that, up to boundary behaviour, \(\measurement{A}\betterthan\measurement{B}\) if and only if \(\lowprev(\measurement{A}-\measurement{B})>0\) or \(\measurement{A}\alwaysbetterthan\measurement{B}\); see \cref{eq:preference:and:lower:prevision} and the discussion leading to it.

We'll see further on that this will definitely be the case if a coherent set of desirable measurement is \emph{maximal}, that is, not included in any other such set, for which an equivalent condition is that
\begin{enumerate}[label={\upshape M}.,ref={\upshape M},labelwidth=*,leftmargin=*]
\item\label{ax:desirs:maximality} \(\set{\measurement{A},-\measurement{A}}\cap\desirset\neq\emptyset\) for all non-null~\(\measurement{A}\in\measurements\).\hfill[maximality]
\end{enumerate}
\begin{proof}[Proof of the equivalence]
Assume that the coherent {\SDM}~\(\desirset\) is maximal and consider any non-null~\(\measurement{A}\in\measurements\) such that \(-\measurement{A}\notin\desirset\), then we need to prove that \(\measurement{A}\in\desirset\).
We're done if we can show that \(\desirset\cup\set{\measurement{A}}\) is consistent, because then \(\ext(\desirset\cup\set{\measurement{A}})\) will be coherent and therefore equal to~\(\desirset\) by its maximality; and clearly \(\measurement{A}\in\ext(\desirset\cup\set{\measurement{A}})\).
Suppose towards contradiction that \(\desirset\cup\set{\measurement{A}}\) isn't consistent, then we infer from \cref{eq:consistency} that \(\zero\in\posi\group{\desirset\cup\set{\measurement{A}}}+\possemidefmeasurements\).
Taking into account that
\begin{align*}
\posi\group{\desirset\cup\set{\measurement{A}}}+\possemidefmeasurements
&=\group[\Big]{\posi\group{\desirset}\cup\posi\group{\measurement{A}}\cup\group[\big]{\posi\group{\desirset}+\posi\group{\set{\measurement{A}}}}}+\possemidefmeasurements\\
&=\group[\big]{\desirset\cup\posi\group{\set{\measurement{A}}}\cup\group{\desirset+\posi\group{\set{\measurement{A}}}}}+\possemidefmeasurements\\
&=\group{\desirset+\possemidefmeasurements}\cup\group{\posi\group{\set{\measurement{A}}}+\possemidefmeasurements}\cup\group{\desirset+\posi\group{\set{\measurement{A}}}+\possemidefmeasurements}\\
&=\desirset\cup\group{\posi\group{\set{\measurement{A}}}+\possemidefmeasurements}\cup\group{\desirset+\posi\group{\set{\measurement{A}}}},
\end{align*}
where the second equality follows from \(\desirset=\posi\group{\desirset}\) [use the coherence of~\(\desirset\) and \labelcref{ax:desirability:additivity,ax:desirability:scaling}] and the last equality follows from \(\desirset+\possemidefmeasurements=\desirset\) [use the coherence of~\(\desirset\) and \labelcref{ax:desirability:monotone}], we then find that, necessarily, \(-\measurement{A}\in\posmeasurements\) [because \(\measurement{A}\neq\zero\)] or \(-\measurement{A}\in\desirset\), contradicting our assumption that \(-\measurement{A}\notin\desirset\) [use the coherence of~\(\desirset\) and \labelcref{ax:desirability:background}].

Conversely, assume that \(\desirset\) satisfies the condition~\labelcref{ax:desirs:maximality}, then we must show that \(\desirset\) is maximal.
Let \(\desirset'\) be any coherent {\SDM} that includes~\(\desirset\), assume towards contradiction that the inclusion is strict and consider any~\(\measurement{A}\in\desirset'\setminus\desirset\).
But then \(-\measurement{A}\in\desirset\), and therefore also \(-\measurement{A}\in\desirset'\).
Since, by assumption, also \(\measurement{A}\in\desirset'\), we infer from the coherence of~\(\desirset'\) [namely, \labelcref{ax:desirability:additivity}] that \(\zero=-\measurement{A}+\measurement{A}\in\desirset'\), contradicting the coherence of~\(\desirset'\) [namely, \labelcref{ax:desirability:strict}].
\end{proof}
\noindent Indeed, the self-conjugacy of~\(\lowprev[\desirset]\) is equivalent to~\(\desirset\) satisfying the so-called \emph{weak maximality} condition, which is formally similar to, but somewhat weaker than, the condition~\labelcref{ax:desirs:maximality} for maximality; see also Refs.~\cite{cooman2021:archimedean:choice,debock2019:interpretation} for related discussion.
\begin{enumerate}[label={\upshape WM}.,ref={\upshape WM},labelwidth=*,leftmargin=*]
\item\label{ax:desirs:mixingness} \(\set[\big]{\measurement{A}+\epsilon\identity,-\measurement{A}+\epsilon\identity}\cap\desirset\neq\emptyset\) for all~\(\measurement{A}\in\measurements\) and all real~\(\epsilon>0\).\hfill[weak maximality]
\end{enumerate}
\noindent Weak maximality means that, if we include Your marginal preferences to come to a notion of weak preference, You always have a(t least a weak) preference between any~\(\measurement{A}\) and its additive inverse~\(-\measurement{A}\).
This statement is clarified and made more precise by the following proposition, which established a direct connection between weakly maximal coherent sets of desirable measurements and coherent previsions.

\begin{proposition}\label{prop:mix}
A coherent {\SDM}~\(\desirset\) is weakly maximal if and only if the corresponding lower prevision~\(\lowprev[\desirset]\) is a coherent prevision, so \(\lowprev[\desirset]=\uppprev[\desirset]=\linprev[\desirset]\).
In that case, the corresponding credal set~\(\linprevs(\linprev[\desirset])\) is the singleton~\(\set{\linprev[\desirset]}\).
\end{proposition}

\begin{proof}
First, assume that \(\desirset\) is weakly maximal.
Consider any~\(\measurement{A}\in\measurements\) and any real~\(\epsilon>0\), then it follows from \cref{eq:define:upper:prevision} that \((\uppprev[\desirset]\group{\measurement{A}}-\epsilon)-\measurement{A}\notin\desirset\).
So it follows from \labelcref{ax:desirs:mixingness} that \(\measurement{A}-(\uppprev[\desirset]\group{\measurement{A}}-\epsilon)\in\desirset\), and therefore \cref{eq:define:lower:prevision} implies that \(\uppprev[\desirset]\group{\measurement{A}}-\epsilon\leq\lowprev[\desirset]\group{\measurement{A}}\), implying that \(\lowprev[\desirset]\group{\measurement{A}}\geq\uppprev[\desirset]\group{\measurement{A}}\) and therefore also that \(\lowprev[\desirset]\group{\measurement{A}}=\uppprev[\desirset]\group{\measurement{A}}\), taking into account \labelcref{ax:lower:prevision:more:bounds}.

Conversely, assume that \(\uppprev[\desirset]=\lowprev[\desirset]\) is a coherent prevision~\(\linprev[\desirset]\).
Assume towards contradiction that there are \(\measurement{A}\in\measurements\) and real~\(\epsilon>0\) such that both \(\measurement{A}+\epsilon\identity\notin\desirset\) and \(-\measurement{A}+\epsilon\identity\notin\desirset\).
Hence, both \(\linprev[\desirset](\measurement{A}+\epsilon\identity)=\lowprev[\desirset](\measurement{A}+\epsilon\identity)\leq0\) and \(\linprev[\desirset](-\measurement{A}+\epsilon\identity)=\lowprev[\desirset](-\measurement{A}+\epsilon\identity)\leq0\), by \cref{lem:desirs:from:lowprev}.
Using~\labelcref{ax:prevision:linear,ax:prevision:norm}, this implies that both \(\linprev[\desirset](\measurement{A})\leq-\epsilon\) and \(-\linprev[\desirset](\measurement{A})\leq-\epsilon\), a contradiction.

The rest of the proof is now immediate.
\end{proof}

In the special, so-called \emph{precise}, case that is characterised by \cref{prop:mix}, we find that the uncertainty is described similarly as on a more standard account of uncertainty in quantum mechanics, because we've seen in \cref{prop:one-to-one} that the coherent previsions~\(\linprev\) on~\(\measurements\) are in a one-to-one correspondence with the density operators~\(\density\) in~\(\densities\).
In fact, we then retrieve a version of Born's rule in the presence of \emph{epistemic uncertainty} about the state of the system, as \cref{prop:one-to-one} tells us that
\[
\linprev\group{\measurement{A}}
=\trace{\density\measurement{A}}
=\expec[\density]{\measurement{A}}
\text{ for all~\(\measurement{A}\in\measurements\)},
\]
where the last equality follows from \cref{eq:qm:expectation}.
We see that when we let \(\density\coloneqq\todensity(\linprev)\), Your fair price~\(\linprev(\measurement{A})=\trace{\density\measurement{A}}\) for the measurement~\(\measurement{A}\) can be made to correspond to the expected value~\(\expec[\density]{\measurement{A}}\) of the outcome of a measurement as considered in the standard approach to quantum mechanics, and discussed in \cref{sec:quantum:mechanics:probabilistic}.

But, we have to be cautious in making the connection between our approach and the more standard one.
In the way we've set up the argument in this paper, the decision problem comes first --- is \emph{primary} --- and the linear price functional~\(\linprev[\desirset]\) is a mathematical tool that, under certain rather restrictive conditions, allows us to characterise Your preferences as captured in~\(\desirset\), in the sense that
\[
\measurement{A}\in\desirset
\ifandonlyif\group[\big]{\linprev[\desirset](\measurement{A})>0\text{ or }\measurement{A}\alwaysbetterthan\zero}.
\]
It turns out that this \(\linprev[\desirset]\) can be characterised mathematically by a density operator~\(\density[\desirset]=\todensity(\linprev[\desirset])\), which, through \cref{prop:qm:density}, is often given a probabilistic interpretation.\footnote{This interpretation is not without its problems, because the `decomposition' in \cref{prop:qm:density} typically isn't unique.}
But this probabilistic interpretation, and the corresponding probabilistic interpretation of the trace~\(\trace{\density[\desirset]\measurement{A}}\) as the expected value of the outcome of the measurement~\(\measurement{A}\), is only \emph{secondary}, or \emph{derivative}.
Indeed, what we have shown above is that \(\linprev\group{\measurement{A}}=\trace{\density\measurement{A}}\) is Your fair price for the uncertain reward~\(\theutility{A}(\uket)\), which \emph{isn't necessarily the same thing as Your expected value for the outcome of the measurement~\(\measurement{A}\)}.

\subsection{An interesting special case}\label{sec:born}
To shed still more light on this issue, let's now consider the special case that You know with certainty that the state is some specific~\(\fket\) in~\(\statespace\) --- so You know that \(\uket=\fket\).\footnote{Recall that the state of the system is actually identified by a ray in Hilbert space~\(\hilbertspace\), of which~\(\fket\) is only one of the elements, which completely determines it. By `You know that \(\uket=\fket\)', we mean that You know that the system is in the state~\(\fket\), up to a phase factor.}
This could for instance be the case after You've just performed a measurement~\(\measurement{B}\) on the quantum system and observed an outcome --- eigenvalue --- with a one-dimensional eigenspace; see \cref{post:qm:after:measurement}.
Through what coherent {\SDM}~\(\desirset\) can this knowledge You have be represented?
Let's answer that question using the tools we acquired for doing so in \cref{sec:inference}.

As we mentioned before right after announcing \cref{backass:dm:utility:function} in \cref{sec:decision:theoretic:background}, there's now no longer any uncertainty about the reward You'll get from performing any measurement~\(\measurement{A}\): it's the real number~\(\theutility{A}(\fket)=\fbra\measurement{A}\fket\), and You'll clearly strictly prefer that to the status quo provided that \(\fbra\measurement{A}\fket>0\).

In other words, Your knowledge leads You to the desirability assessment
\begin{equation*}
\assessment_{\fket}
\coloneqq\cset{\measurement{A}\in\measurements}{\fbra\measurement{A}\fket>0}.
\end{equation*}
Check that \(\assessment_{\fket}=\posi\group{\assessment_{\fket}}=\assessment_{\fket}+\posmeasurements\), so we can infer from \cref{eq:natural:extension} that
\begin{equation*}
\ext\group{\assessment_{\fket}}
=\posmeasurements\cup\posi\group{\assessment+\posmeasurements}
=\posmeasurements\cup\assessment_{\fket}.
\end{equation*}
Since \(\zero\notin\ext\group{\assessment_{\fket}}\), we see that the desirability assessment~\(\assessment_{\fket}\) is consistent, so the smallest coherent set of desirable measurements that includes it, is given by
\begin{equation*}
\desirset[\fket]\coloneqq\ext\group{\assessment_{\fket}}=\posmeasurements\cup\assessment_{\fket}.
\end{equation*}
Check that \(\assessment_{\fket}\cap\set{\measurement{A}+\epsilon\identity,-\measurement{A}+\epsilon\identity}\neq\emptyset\) for all real~\(\epsilon>0\) and~\(\measurement{A}\in\measurements\), which guarantees that \(\desirset[\fket]\) is weakly maximal, i.e.~satisfies the criterion~\labelcref{ax:desirs:mixingness}.
\cref{prop:mix} then guarantees that the corresponding buying price functional~\(\lowprev[{\desirset[\fket]}]\) is a \emph{coherent prevision}, which we'll also denote by~\(\linprev[\fket]\) and which is then given by
\begin{align*}
\linprev[\fket]\group{\measurement{A}}
\coloneqq&\sup\cset{\alpha\in\reals}{\measurement{A}-\alpha\identity\in\desirset[\fket]}
=\sup\cset{\alpha\in\reals}{\fbraketwithop{\measurement{A}}>\alpha\text{ or }\measurement{A}\alwaysbetterthan\alpha\identity}\\
=&\sup\cset{\alpha\in\reals}{\fbraketwithop{\measurement{A}}>\alpha\text{ or }(\min\spectrum{A}\geq\alpha\text{ and }\measurement{A}\neq\alpha\identity)}\\
=&\fbraketwithop{\measurement{A}}
=\theutility{A}(\fket),
\text{ for all~\(\measurement{A}\in\measurements\)},
\end{align*}
where the penultimate equality follows from the fact that \(\fbraketwithop{\measurement{A}}\geq\min\spectrum{A}\) [see \cref{eq:infimum:utility}].
Again, we retrieve a form of our version of Born's rule, now for the case that the state of the system is known with certainty; see \cref{eq:qm:conditional:expectation}.
Clearly, our derivation for it in this special case is merely an application of conservative inference and deriving the buying price functional based on the desirability assessment~\(\assessment_{\fket}\).

In other words, and summarising, if You know that \(\uket=\fket\), then \(\theutility{A}(\fket)=\fbraketwithop{\measurement{A}}\) is necessarily \emph{Your fair price for performing the measurement~\(\measurement{A}\)}.
Again, this isn't necessarily the same thing as Your expected value for the outcome of the measurement~\(\measurement{A}\).

Stating that Your fair price \(\theutility{A}(\fket)=\fbraketwithop{\measurement{A}}\) for performing the measurement~\(\measurement{A}\) \emph{is} the same thing as Your expected value for the outcome of the measurement~\(\measurement{A}\), is tantamount to making an \emph{extra assumption}, namely that there are probabilities for each of the possible outcomes~\(\eigval\in\spectrum{A}\) of the measurement~\(\measurement{A}\) and that these probabilities are given by \cref{eq:qm:born}.
Making this extra assumption therefore amounts to accepting \cref{post:qm:born}, which is something we've been wanting to avoid all along in this paper.

We'll now explain why it \emph{might} make sense to make this extra assumption, all the while realising that it \emph{is} indeed an additional step to take.
In \cref{backass:dm:utility:function} we assumed the existence of a reward function~\(\utility{A}\) such that \(\utility{A}(\fket)\) is the reward for performing measurement~\(\measurement{A}\) when the system is in the state~\(\fket\) and there intended to be a one-number summary of the possible outcomes~\(\eigval\in\spectrum{A}\) that the measurement~\(\measurement{A}\) may yield in the state~\(\fket\).
The exact form of the reward function~\(\utility{A}\) was then fixed by the later postulates~\cref{post:dm:eigenket,post:dm:additivity,post:dm:different:eigenspaces,post:dm:continuity} to be~\(\utility{A}=\theutility{A}=\dotbra\measurement{A}\dotket\).
Interestingly, these postulates make sure that the one-number summary~\(\theutility{A}(\fket)\) behaves linearly in \(\measurement{A}\) and therefore \emph{acts as if it were an expectation}:
\begin{align*}
\theutility{A}(\fket)
=\fbraketwithop{\measurement{A}}
=\fbraketwithop[\bigg]{\smashoperator[r]{\sum_{\eigval\in\spectrum{A}}}\lambda\projection[\eigval]}
=\smashoperator[r]{\sum_{\eigval\in\spectrum{A}}}\lambda\fbraketwithop{\projection[\eigval]}
=\smashoperator[r]{\sum_{\eigval\in\spectrum{A}}}\lambda\theutilitypure{\projection[\eigval]}(\fket),
\end{align*}
where the \(\theutilitypure{\projection[\eigval]}(\fket)=\fbraketwithop{\projection[\eigval]}\) then act as if they were the expected outcomes of the measurements~\(\projection[\eigval]\).
Since the only possible outcomes of a projection~\(\projection[\eigval]\) are \(0\) and \(1\), the one-number summary~\(\theutilitypure{\projection[\eigval]}(\fket)\) then also acts as if it were the probability of its outcome being~\(1\), or in other words, the probability of the measurement~\(\measurement{A}\) producing the outcome~\(\lambda\), in accordance with Born's rule \cref{post:qm:born}.

This special case of our decision-theoretic argument brings us quite close to a result by Deutsch, which has a similar interpretation, but follows a different argumentation, as we'll see next.

\section{Deutsch's decision-theoretic approach}\label{sec:deutsch}
In an important and seminal paper \cite{deutsch1999:quantum:decisions}, David Deutsch presented a different, and earlier, approach to what we're attempting here.

Deutsch aimed at deriving the probabilistic postulate~\cref{post:qm:born} of quantum mechanics --- or Born's rule in its simplest version~\eqref{eq:qm:born} --- from the non-probabilistic ones and straightforward `non-probabilistic' decision-making principles.
While our argumentation in this paper is inspired by Deutsch's take on this issue, it's also rather different in quite a number of respects, in that we're not necessarily convinced that the `non-probabilistic' decision-making principles he suggested are all that straightforward.
This is why we made quite some effort in this paper to clearly state exactly what decision-making approach we follow, which assumptions we make and which we don't make.
We believe our derivation can be defended on any interpretation of quantum mechanics that accepts the postulates~\cref{post:qm:kets,post:qm:operator,post:qm:eigenvalues,post:qm:deterministic} (and \cref{post:qm:after:measurement}) and the reward function postulates~\cref{post:dm:additivity,post:dm:continuity,post:dm:eigenket,post:dm:different:eigenspaces}, which make sense against the decision-making background (\cref{backass:dm:preference:relation,backass:dm:utility:function}) sketched in \cref{sec:decision:theoretic:background}.

To allow readers to better see the similarities and differences between Deutsch's argument and ours, we now give a brief account of it, which should also be compared to and contrasted with the approach we've set up in \cref{sec:our:approach,sec:desirability}.

Deutsch considers a quantum system whose states belong to some complex Hilbert space~\(\hilbertspace\).
He focuses on a measurement~\(\measurement{A}\in\measurements\) on this quantum system, which can be in any of several possible states~\(\gket\in\statespace\) right before the measurement is carried out.

His decision-making set-up considers a subject, whom we'll also call \emph{You} and who can choose between playing several games of the type
\begin{equation*}
\game{\psi}
\coloneqq\text{ ``perform the measurement~\(\measurement{A}\) on the quantum system in state~\(\gket\)''},
\end{equation*}
one for each~\(\gket\in\statespace\) and~\(\measurement{A}\in\measurements\).
The outcome of~\(\game{\psi}\) is unknown to You before playing the game --- doing the measurement --- but, after the measurement, You'll receive as a pay-off the actual outcome of the measurement, expressed in units of some linear utility scale.
While the latter is an assumption we also make in our set-up in \cref{sec:setup}, our general approach is different in that Deutsch focuses on the case that You know what the system's state is, while we allow You to be uncertain about it from the outset.

The \emph{first assumption} that Deutsch makes,\footnote{The order in which we list these assumptions doesn't follow the order in which Deutsch presents them, but rather reflects our attempt at structuring his argumentation to allow for a comparison with our approach.} is that each game~\(\game{\psi}\) has some real \emph{value}~\(\valuation(\gket)\), intended as a one-number summary of the possible pay-offs that the game \(\game{\psi}\) may yield.
He seems to be inspired by de Finetti's \cite{finetti1937} definition of a prevision for a random quantity as a fair price for it when he defines this value as
\begin{quote}
the utility of a hypothetical pay-off such that [You are] indifferent to playing the game and receiving that pay-off unconditionally.
\end{quote}
Since it follows from the non-probabilistic postulates of quantum mechanics [in particular \cref{post:qm:deterministic}] that performing the measurement~\(\alpha\identity\) yields the real outcome~\(\alpha\) unconditionally --- meaning independently of~\(\gket\) ---, this implies in effect that \(\valuation[\alpha\identity]=\alpha\).
Hence, there is --- You are assumed to have --- some \emph{valuation}~\(\valuationsymbol\colon\measurements\times\statespace\to\reals\colon(\measurement{A},\fket)\mapsto\valuation(\fket)\) that assigns a \emph{value} \(\valuation(\gket)\) to each~\(\game{\psi}\) in such a way that
\begin{equation}\label{eq:valuation:definition}
\valuation(\gket)=\valuation[\alpha\identity](\gket)
\ifandonlyif\game{\psi}\equiv\game[\alpha]{\psi},
\end{equation}
where we define the statement~`\(\game{\psi}\equiv\game[\alpha]{\psi}\)' to mean that You're \emph{indifferent} between the game~\(\game{\psi}\) and the game~\(\game[\alpha]{\psi}\) with unconditional pay-off~\(\alpha\).
That there should be such a real number for all such~\(\game{\psi}\), is an assumption that is very close to our decision-theoretic background assumption \cref{backass:dm:utility:function}, where we assume that there's some reward function\footnote{Recall that, initially, we use the notation~\(\utility{A}\) for any such reward function, and we reserve the notation~\(\theutility{A}\) for the unique reward function that turns out to satisfy our reward function postulates.}~\(\utility{A}\colon\statespace\to\reals\) such that \(\utility{A}(\gket)\) is the reward, expressed in utiles, for~\(\act{A}\) when the quantum system under consideration system is in state~\(\gket\), intended as a one-number summary of the possible pay-offs that the measurement~\(\measurement{A}\) may yield on a system in state~\(\gket\).

Deutsch's further assumptions deal with the nature of Your `value of a game', and he argues that they're strong enough to fix its form.
The \emph{third assumption} is that You adhere to the \emph{principle of substitutability} in composite games, which are games that involve subgames.
It means that if any of the subgames is replaced by some game of equal value, the value of the composite game remains the same to You.

The \emph{fourth and final assumption} that Deutsch makes, is what he calls the \emph{zero-sum rule}, and it concerns Your attitude towards games where You can take one of two possible roles --- say, \(R\) and \(S\) --- with the following property: whenever You were to receive a pay-off~\(x\) in role~\(R\), You would receive~\(-x\) in role~\(S\).
The rule then states that if \(\valuation[R]\) and \(\valuation[S]\) are Your respective values for playing the game in roles~\(R\) and~\(S\), then it must be that \(\valuation[R]+\valuation[S]=0\).

Using these decision-theoretic assumptions, together with the non-probabilistic quantum mechanical postulates, Deutsch argues that Your valuation must then necessarily be given by
\begin{equation}\label{eq:born:a:la:deutsch}
\valuation(\fket)=\fbra\measurement{A}\fket
\text{ for all }\measurement{A}\in\measurements\text{ and }\fket\in\statespace.
\end{equation}
Recall that, according to Deutsch's argumentation, this \(\valuation(\gket)\) is the fair price You're willing to pay for playing the game~\(\game{\psi}\), or in other words, for performing the measurement~\(\measurement{A}\) when the quantum system's state is~\(\gket\).
It's in this specific sense that Deutsch's argument leads to a justification of, and interpretation for, Born's rule as put forth in Postulate~\cref{post:qm:born}, and in particular \cref{eq:qm:conditional:expectation}.
This is comparable to our justification for it using the postulates~\cref{post:dm:additivity,post:dm:continuity,post:dm:eigenket,post:dm:different:eigenspaces} and the background assumptions~\cref{backass:dm:preference:relation,backass:dm:utility:function}, and caveats about the interpretation of these results apply here that are similar to the ones we formulated for our derivation near the end of \cref{sec:born,sec:a:warning}.

On Deutsch's approach, the system state~\(\gket\) is assumed to be known, or given, and the value~\(\valuation(\gket)\) of~\(\game{\psi}\) is Your fair price for performing the measurement~\(\measurement{A}\) upon the system in that state~\(\gket\), in the sense that You deem performing the measurement and getting as a reward its outcome in units of a linear utility to be equivalent to getting the fixed amount~\(\valuation(\gket)\) of that linear utility.
His way of dealing with the problem is, like ours, very much de Finetti-like, in that he starts with a pre-determined linear utility and implicitly assumes that all games~\(\game{\psi}\) with~\(\measurement{A}\in\measurements\), for a fixed~\(\gket\in\statespace\) can be compared to one another through the intervention of the values that their valuation~\(\valuation(\gket)\) assumes on the real line;\footnote{Bruno de Finetti \cite{finetti1937,finetti19745} likewise assumes that all uncertain quantities~\(X\) can be compared to one another through the intervention of the values that their prevision, or fair price,~\(P(X)\) assume on the real line.} see also our discussion on decision-theoretic foundations in \cref{sec:decision:theoretic:background}.

It should be clear that Deutsch's decision problem, and the context he places it in, is in its conception a bit different from ours.
We start from the assumption that You are (or at least may be) uncertain about the state~\(\uket\) the system is in, and that You express Your beliefs about what that state is by expressing preferences between the uncertain rewards~\(\utility{A}(\uket)\) for different measurements~\(\measurement{A}\).
In addition, our postulates, as announced in \cref{sec:decision:theoretic:postulates}, are quite different from Deutsch's.
But, as we point out in \cref{sec:born}, in the special case that You know that the state~\(\uket\) of the system is actually \(\gket\), we're still led to the same conclusion as Deutsch: You must have a \emph{fair price} for the uncertain reward associated with performing the measurement~\(\measurement{A}\) in that --- now certain --- state, \emph{and} that fair price must be given by \(\theutility{A}(\gket)=\gbraketwithop{\measurement{A}}\).
In summary, to the extent that Deutsch's argumentation is correct,\footnote{See also the discussion in the next section about objections against Deutsch's approach, and the observation in footnote~\ref{fn:unitary:invariance} that our postulate \cref{post:dm:different:eigenspaces} meets those objections.} it leads to similar conclusions to ours \emph{when the quantum systems's state is known}: Born's rule is encapsulated in the reward function that can be derived from the decision-theoretic assumptions that are being made.
Both sets of assumptions have quite a different flavour, though, and Deutsch's approach is tailored to defending the Everettian interpretation of quantum mechanics, which isn't necessarily what we're after.

To be sure, while Deutsch's discussion in Ref.~\cite{deutsch1999:quantum:decisions} deals directly only with the case that You know the system to be in some pure state~\(\gket\), he does consider that You might be uncertain about the quantum state, by mentioning the possibility that the system might be in a mixed state.
Even if he leaves the discussion of mixed states implicit, he does hint, at the end of Section~3 there, at ways to deal with them:
\begin{quote}
Generalizing these results to cases where [the quantum system] is not in a pure state is trivial if [the system] is part of a larger system that is in a pure state, for then every measurement on [the system] is also a measurement on the larger system. Further generalisation to exotic situations in which the universe as a whole may be in a mixed state [\dots] is left as an exercise for the readers.
\end{quote}
It appears from the first sentence of this quote that Deutsch wants to treat epistemic uncertainty about the system's state by using the well-known observation that a mixed state can be treated as a pure state for a larger system, which would make his earlier arguments about how to derive Born's rule for pure states amenable to mixed states as well.
But, since he assumes a linear ordering of games/measurements for such pure states, it seems fair to infer that he'd be willing to order games/measurements linearly also when the system's state isn't perfectly known to You.\footnote{This of course assumes that You're able to identify the pure state that the larger system is in, which is quite a strong assumption to make. The strength of this assumption is a good indicator of the strength of the \emph{totality} assumption for Your preference ordering. In addition, it's not entirely clear to us why Deutsch would consider a situation where You have epistemic uncertainty about the `state of the universe as a whole' --- that the universe as a whole is in a mixed state --- to be `exotic'.}

Our willingness to work, in that case, with \emph{partial} preference orderings --- which aren't reducible to mixed states nor to pure states in higher-dimensional state spaces ---  when dealing with Your epistemic uncertainty about the system's state, can therefore only be seen as a fundamental departure from Deutsch's way of thinking.
It's interesting in this light that, as we've seen in \cref{sec:density:operators}, such partial preference orderings are mathematically (almost-)equivalent to convex closed sets of mixed states --- density operators --- and therefore to \emph{sets of pure states} for a higher-dimensional system.

\section{Wallace's decision-theoretic approach}\label{sec:wallace}
Soon after Deutsch published his argument, \citeauthor{barnum2000:quantum} \cite{barnum2000:quantum} pointed to a flaw in it: in order for Deutsch's proof to work, another critical \emph{assumption} is necessary, namely that Your value for the game~\(\game{\psi}\) should equal Your value for game~\(\game[B]{\phi}\), where \(\measurement{B}=\measurement{U}\measurement{A}\adjoint{\measurement{U}}\) and \(\fket=\measurement{U}\gket\), and where \(\measurement{U}\) is some unitary operator on~\(\hilbertspace\).
In words, a unitary transformation of the Hilbert space and the corresponding measurements mustn't have any influence on Your valuation of the game.\footnote{\label{fn:unitary:invariance}This invariance under unitary transformations is in particular implied by our postulate \cref{post:dm:different:eigenspaces}.}
Wallace \cite{wallace2003:defending:deutsch} has argued that the above condition can be derived in the context of the Everettian interpretation and has continued to improve on this work \cite{wallace2007:improving:deutsch,wallace2009:born:arxiv}.

However, in contradistinction with Deutsch, who we've already argued relies on de Finetti's approach in representing Your preferences between games involving measurements, Wallace uses Savage's approach; see \cref{sec:decision:theoretic:background} for a short summary of Savage's ideas.
Because Wallace's argument is much more detailed and intricate, we'll try to hint at its decision-theoretic essence by using what we believe to be relevant quotations from his more recent work \cite{wallace2009:born:arxiv} and providing them with our interpretations,\footnote{We're referring to the arXiv version, which is dated `October 22, 2018', of a chapter published in 2010 as \emph{How to prove the Born Rule} in Ref.~\cite{saunders2018:many:worlds}.} where we've preserved the wording, but adapted the mathematical notation to accord with ours.
We'll discuss the differences with our approach as we go along.
\begin{quote}
A quantum state is to be prepared in some superposition; the system is measured in some basis; a bet is made by the agent on the outcome of that measurement.
Our agent knows (we assume) that the Everett interpretation is correct; he is also assumed to know the universal quantum state, or at least the state of his branch.
[\dots]
His preferences can be represented by an ordering relation on these bets.
Since (in Everettian quantum mechanics, at any rate) preparations, measurements and payments made to agents are all physical processes, there is a certain simplification available: any preparation-followed-by-measurement-followed-by-payments can be represented by a single unitary transformation.
So our agent’s rational preference is actually representable by an ordering on unitary transformations.
\end{quote}
So, we see that You are to express preferences, in a quantum state~\(\gket\), between acts, which are the unitary transformations~\(\operator{U}\) that are available to You in that state.
\begin{quote}
We now need to represent the agent’s preferences between acts.
Since those preferences may well depend on the state, we write it as follows: if the agent prefers (at~\(\gket\)) act~\(\operator{U}\) to act~\(\operator{U}'\), we write
\[
\operator{U}\succ^{\gket}\operator{U}'.
\]
To be meaningful, of course, this requires that \(\operator{U}\) and \(\operator{U}'\) are both available at \(\gket\)'s macrostate.
So \(\succ^{\gket}\) is to be a two-place relation on the set of acts available at that macrostate.
\end{quote}
So, like Deutsch, and unlike us, Wallace assumes the quantum system to be in a known (macro)state.
Unlike Deutsch, and like us, he starts with a strict preference ordering that You have on a set of acts and where each act has an uncertain reward.
But as we'll see, unlike both Deutsch and us, he doesn't assume that the rewards can be expressed in some predetermined linear utility scale; like Savage, he `constructs' the utilities from the properties of the preference ordering and the richness of the act (and reward) space.

Indeed, Wallace then imposes a number of (so-called \emph{richness}) axioms, or assumptions, that guarantee that the set of acts (and rewards) is rich enough, and a number of (so-called \emph{rationality}) axioms on Your preferences on such acts, the most important of which for our present discussion is the following:
\begin{quote}
{\bfseries Ordering:} The relation~\(\succeq^{\gket}\) is a total ordering for each~\(\psi\) on the set of acts available at~\(\psi\), for each~\(\psi\) (that is: it's transitive, irreflexive and asymmetric, and if we define~\(\operator{U}\sim^{\gket}\operator{V}\) as holding whenever \(\operator{U}\succ^{\gket}\operator{V}\) and \(\operator{V}\succ^{\gket}\operator{U}\) fail to hold, then \(\sim^{\gket}\) is an equivalence relation).
\end{quote}
First off, we believe the symbol~`\(\succeq^{\gket}\)' to be a typo for `\(\succ^{\gket}\)', as the former symbol is nowhere defined in Wallace's text and the `weak' aspect that the extra horizontal bar might suggest is contradicted by the irreflexivity requirement.
So, this axiom, fully in line with Savage's approach, requires that Your preference ordering~\(\succ^{\gket}\) should be a strict partial ordering\footnote{A strict partial ordering is an irreflexive and transitive binary relation; the asymmetry is then implied.} (as is also assumed on our approach) that has the additional \emph{totality} property: the associated relation~\(\sim^{\gket}\), which generally speaking represents both indifference and incomparability, must actually be an equivalence relation and can therefore only represent indifference; incomparability in a given state is excluded on Wallace's approach (as it is on ours as well).

Based on his richness and rationality axioms and following the path blazed by Savage, Wallace is able to provide an ingenious argument to show that there's a utility function on the rewards, such that the strict weak ordering~\(\succ^{\gket}\) on the acts~\(\operator{U}\) can be represented by the strict total ordering of their expected utilities, where the underlying probabilities are provided by Born's rule.

\section{What are the advantages to our approach?}\label{sec:our:approach:is different}
It's clear that the decision problem that Wallace considers, is quite similar to (if more involved than) Deutsch's, in that acts are compared for a quantum system in a given, perfectly known, quantum state~\(\gket\).
Contrary to our approach, there's no uncertainty about the quantum system's state.

In our more general decision problem --- whose solution coincides with Deutsch's and Wallace's solutions in the special case that the system state is indeed known --- we allow for strict partial (preference) orderings \emph{that needn't be total}, so we allow for incomparability between acts: \(\measurement{A}\notbetterthan\measurement{B}\) and \(\measurement{B}\notbetterthan\measurement{A}\) needn't imply that You are indifferent between the measurements~\(\measurement{A}\) and~\(\measurement{B}\): You may also hold them to be incomparable, and then You won't be able to compare them on a linear scale.

Indeed, much of what we do in this paper amounts to showing that this total ordering requirement isn't necessary \emph{in our more general decision problem}: we can still derive interesting results, and in particular also recover Born's rule as a special case, without it, using our decision-theoretic assumptions.
At the same time, we've shown in the previous sections that in doing so, we uncover an interesting mathematical toolbox for representing, and making conservative inferences about, uncertainty about quantum systems.
We start out from the assumption that You are (or may be) uncertain about the state~\(\uket\) of the quantum system under consideration and that performing measurements~\(\measurement{A}\) on the system, where the rewards are expressed in units of some predetermined linear utility, will therefore yield uncertain rewards, which You can order via a strict ordering that is allowed to be merely \emph{partial when You don't know the system's state}.
This freedom has at least two advantages, besides its making the resulting theory more general: (i) it's more realistic because specifying only a partial order is a lot easier for You to do, and all the more so if You only have limited --- finite --- time and resources at Your disposal; and (ii) it uncovers the constructive conservative inference mechanism that underlies doing inferences about the state of a quantum system, that tends to get hidden by a focus on total orderings.\footnote{The analogy with propositional logic is telling: (coherent) partial orderings correspond to deductively closed sets and the total orderings to maximal deductively closed sets, or complete theories; focussing on complete theories alone hides the deductive character of propositional logic and it surfaces only when other, non-maximal, deductively closed sets are considered.}

Born's rule is therefore still recovered, even if we allow for partial orderings to represent Your preferences in our more general decision problem where the system state is unknown to You.

\section{Formulation and proof of the main theorem}\label{sec:proofs}
Let's now prove the main result in this paper, which is, essentially, that the postulates~\cref{post:dm:eigenket,post:dm:different:eigenspaces,post:dm:additivity,post:dm:continuity} guarantee that the shape of the reward functions~\(\utility{A}\) is completely fixed, as in \cref{eq:born:a:la:us}.

The exact formulation in \cref{thm:utility} below takes some care and preparation, however.
We'll need to be able to switch freely between various finite-dimensional Hilbert spaces, which is why we'll need to consider the set~\(\allhilbertspaces\) of all of them.

We'll use the notation~\(\hilbertspace,\hilbertspace',\hilbertspace''\) for elements of~\(\allhilbertspaces\).
Every such Hilbert space comes with its own inner product and its own real linear space of Hermitian operators, and we'll always use the same respective notations~\(\dotbraket\) and~\(\measurements\) for these, as it will always be clear from the context which Hilbert space they're associated with.

In each such finite-dimensional Hilbert space~\(\hilbertspace\in\allhilbertspaces\), there are of course infinitely many ways to associate a reward function~\(\utility{A}\in\realmaps\) with each measurement~\(\measurement{A}\in\measurements\), and each of them corresponds to a so-called \emph{reward assignation}~\(\utilitymap\colon\measurements\to\realmaps\colon\measurement{A}\mapsto\utility{A}\) \emph{for}~\(\hilbertspace\).
Recall that we denote by~\(\realmaps\) the set of all real-valued maps on the state space~\(\statespace\).
We'll denote by~\(\utilitymaps\) the set of all such possible reward assignations~\(\utilitymap\) for the Hilbert space~\(\hilbertspace\).

Each of the many ways of doing this for all possible finite-dimensional Hilbert spaces is captured in a specific map~\(\assignation\) on the set~\(\allhilbertspaces\) that selects a reward assignation~\(W(\hilbertspace)\in\utilitymaps\) for each Hilbert space~\(\hilbertspace\in\allhilbertspaces\).
We'll call any such map~\(\assignation\colon\allhilbertspaces\to\utilitymaps\) a \emph{reward assignation system}.
We'll use the notation~\(\utilitymap\) for the reward assignation~\(W(\hilbertspace)\), as it will always be clear from the context what Hilbert space~\(\hilbertspace\) we're working in.

The task before us, now, is to show that the postulates~\cref{post:dm:eigenket,post:dm:different:eigenspaces,post:dm:additivity,post:dm:continuity} allow for only one specific reward assignation system, namely the one specified by~\cref{eq:born:a:la:us}.
We'll say that a reward assignation system~\(\assignation\) \emph{obeys} a specific postulate if the statement in the postulate is true for the reward assignations~\(\utilitymap=\assignation(\hilbertspace)\) of all the possible instantiations~\(\hilbertspace\) of all the various Hilbert spaces mentioned in the postulate.

\begin{theorem}\label{thm:utility}
There's a unique reward assignation system~\(\theassignation\) that obeys Postulates~\cref{post:dm:eigenket,post:dm:different:eigenspaces,post:dm:additivity,post:dm:continuity}.
For all Hilbert spaces~\(\hilbertspace\in\allhilbertspaces\) and all corresponding~\(\measurement{A}\in\measurements\), its corresponding reward function~\(\theutility{A}\colon\statespace\to\reals\) is given by
\begin{equation}\label{eq:born:a:la:us:in:the:theorem}
\theutility{A}(\gket)=\gbra\measurement{A}\gket
\text{ for all }\gket\in\statespace.
\end{equation}
\end{theorem}
\noindent We want to stress here that \cref{eq:born:a:la:us:in:the:theorem} determines the reward assignation system~\(\theassignation\) fully: for any Hilbert space~\(\hilbertspace\), it specifies the values~\(\theutility{A}\) of the reward assignation~\(\theutilitymap=\theassignation(\hilbertspace)\) in all Hermitian operators~\(\measurement{A}\in\measurements\) on~\(\hilbertspace\).

In our proof for this result further on, we'll rely heavily on the additivity property of the uncertain rewards for commuting measurements.
Before we can formulate this property in \cref{prop:equi_additivity_post} and prove it as a special consequence of Postulate~\cref{post:dm:additivity}, it's useful to recall the following essential property of commuting Hermitian operators.

\begin{proposition}[\protect{\cite[Thm.~2.2]{nielsen2010:quantum}}]\label{prop:commuting:operators}
Two Hermitian operators~\(\measurement{A},\measurement{B}\) on a Hilbert space~\(\hilbertspace\) commute, meaning that \(\commutator{\measurement{A}}{\measurement{B}}\coloneqq\measurement{A}\measurement{B}-\measurement{B}\measurement{A}=\hat{0}\), if and only if there's some orthonormal basis~\(\set{\eigket[1],\dots,\eigket[n]}\) of~\(\hilbertspace\) such that each of its elements~\(\eigket[\ell]\) is an eigenstate for both operators.
\end{proposition}

\begin{proposition}\label{prop:equi_additivity_post}
Consider any reward assignation system~\(\assignation\) that obeys Postulate~\cref{post:dm:additivity}, and any two commuting Hermitian operators~\(\measurement{A},\measurement{B}\) on a Hilbert space~\(\hilbertspace\).
Then \(\wval_{\measurement{A}+\measurement{B}}(\fket)=\utility{A}(\fket)+\utility{B}(\fket)\) for all~\(\fket\in\statespace\).
\end{proposition}

\begin{proof}
\cref{prop:commuting:operators} guarantees that there's some orthogonal basis of states~\(\set{\eigket[1],\dots,\eigket[n]}\) that are eigenstates for both \(\measurement{A}\) and \(\measurement{B}\).
Let \(\lambda_1,\dots,\lambda_n\) and \(\mu_1,\dots,\mu_n\) be the (not necessarily distinct) corresponding eigenvalues of~\(\measurement{A}\) and \(\measurement{B}\) respectively, then \(\measurement{A}=\sum_{k=1}^n\lambda_k\eigket[k]\eigbra[k]\) and \(\measurement{B}=\sum_{k=1}^n\mu_k\eigket[k]\eigbra[k]\), by \cref{prop:basis:and:eigenvalues}.
Let \(\lambda_{\mathrm{max}}\coloneqq\max_{k=1}^n\abs{\lambda_k}=\supnorm{\measurement{A}}\), \(\mu_{\mathrm{max}}\coloneqq\max_{k=1}^n\abs{\mu_k}=\supnorm{\measurement{B}}\) and \(\con\coloneqq2(\lambda_{\mathrm{max}}+\mu_{\mathrm{max}})+\epsilon\) with \(\epsilon>0\).
Also let \(\lambda_k'\coloneqq\lambda_k+k\con\), \(\lambda_k''\coloneqq-k\con\), \(\mu_k'\coloneqq\mu_k+k\con\) and \(\mu_k''\coloneqq-k\con\) for all~\(k\in\set{1,\dots,n}\).
Then, on the one hand,
\begin{equation}\label{eq:equi_additivity_post:sums}
\lambda_k=\lambda_k'+\lambda_k''\text{ and }\mu_k=\mu_k'+\mu_k''
\text{ for all~\(k\in\set{1,\dots,n}\)},
\end{equation}
while, on the other hand,
\begin{equation}\label{eq:equi_additivity_post:distinct}
\left.
\begin{aligned}
&\lambda_k'+\mu_k'\neq\lambda_\ell'+\mu_\ell'\text{ and }\lambda_k''+\mu_k''\neq\lambda_\ell''+\mu_\ell''\\
&\lambda_k'\neq\lambda_\ell'\text{ and }\lambda_k''\neq\lambda_\ell''\\
&\mu_k'\neq\mu_\ell'\text{ and }\mu_k''\neq\mu_\ell''\\
\end{aligned}
\right\}
\text{ for all~\(k,\ell\in\set{1,\dots,n}\) such that \(k\neq\ell\).}
\end{equation}
Let \(\measurement{A}'\coloneqq\sum_{k=1}^n\lambda_k'\eigket[k]\eigbra[k]\) and \(\measurement{A}''\coloneqq\sum_{k=1}^n\lambda_k''\eigket[k]\eigbra[k]\), and similarly \(\measurement{B}'\coloneqq\sum_{k=1}^n\mu_k'\eigket[k]\eigbra[k]\) and \(\measurement{B}''\coloneqq\sum_{k=1}^n\mu_k''\eigket[k]\eigbra[k]\), then it follows from \cref{eq:equi_additivity_post:sums} that \(\measurement{A}=\measurement{A}'+\measurement{A}''\) and \(\measurement{B}=\measurement{B}'+\measurement{B}''\).
Moreover, \cref{eq:equi_additivity_post:distinct} guarantees that the measurements~\(\measurement{A}'\) and~\(\measurement{A}''\),  as well as the measurements~\(\measurement{B}'\) and~\(\measurement{B}''\), have distinct eigenvalues.
Postulate~\cref{post:dm:additivity} then ensures that
\begin{equation}\label{eq:equi_additivity_post:ab1}
\wval_{\measurement{A}}
=\wval_{\measurement{A}'}+\wval_{\measurement{A}''}
\text{ and }
\wval_{\measurement{B}}
=\wval_{\measurement{B}'}+\wval_{\measurement{B}''}.
\end{equation}
But \cref{eq:equi_additivity_post:distinct} also guarantees that the measurements~\(\measurement{A}'\) and~\(\measurement{B}'\), as well as the measurements~\(\measurement{A}''\) and~\(\measurement{B}''\), have distinct eigenvalues, so we can equally well apply~\cref{post:dm:additivity} to find that
\begin{equation}\label{eq:equi_additivity_post:ab2}
\wval_{\measurement{A}'+\measurement{B}'}
=\wval_{\measurement{A}'}+\wval_{\measurement{B}'}
\text{ and }
\wval_{\measurement{A}''+\measurement{B}''}
=\wval_{\measurement{A}''}+\wval_{\measurement{B}''}.
\end{equation}
Finally, \cref{eq:equi_additivity_post:distinct} guarantees that the measurements~\(\measurement{A}'+\measurement{B}'\) and \(\measurement{A}''+\measurement{B}''\) have distinct eigenvalues, and since \(\measurement{A}+\measurement{B}=\group{\measurement{A}'+\measurement{B}'}+\group{\measurement{A}''+\measurement{B}''}\), we can again use Postulate~\cref{post:dm:additivity} to get
\begin{equation}\label{eq:equi_additivity_post:ab3}
\wval_{\measurement{A}+\measurement{B}}
=\wval_{\measurement{A}'+\measurement{B}'}+\wval_{\measurement{A}''+\measurement{B}''}.
\end{equation}
But then, summarising,
\begin{align*}
\wval_{\measurement{A}+\measurement{B}}
\overset{\textrm{\eqref{eq:equi_additivity_post:ab3}}}{=}\wval_{\measurement{A}'+\measurement{B}'}+\wval_{\measurement{A}''+\measurement{B}''}
&\overset{\textrm{\eqref{eq:equi_additivity_post:ab2}}}{=}(\wval_{\measurement{A}'}+\wval_{\measurement{B}'})+(\wval_{\measurement{A}''}+\wval_{\measurement{B}''})\\
&=(\wval_{\measurement{A}'}+\wval_{\measurement{A}''})+(\wval_{\measurement{B}}+\wval_{\measurement{B}''})
\overset{\textrm{\eqref{eq:equi_additivity_post:ab1}}}{=}\wval_{\measurement{A}}+\wval_{\measurement{B}}.
\qedhere
\end{align*}
\end{proof}

We'll also rely on the following consequence of \cref{post:dm:different:eigenspaces}, in a form that doesn't require the eigenvalues considered to be different.

\begin{proposition}\label{prop:different:eigenvectors}
Consider any reward assignation system~\(\assignation\) that obeys Postulate~\cref{post:dm:different:eigenspaces}, any Hermitian operator~\(\measurement{A}\) on a Hilbert space~\(\hilbertspace_1\), and any orthogonal basis of its eigenstates~\(\set{\eigket[1],\dots,\eigket[n]}\) with corresponding (not necessarily distinct) real eigenvalues~\(\eigval_1\), \dots, \(\eigval_n\), implying that \(\measurement{A}=\sum_{k=1}^n\eigval_k\eigket[k]\eigbra[k]\).
Also consider a Hilbert space~\(\hilbertspace_2\) and an Hermitian operator~\(\measurement{B}\) on~\(\hilbertspace_2\) with the same eigenvalues \(\eigval_1\), \dots, \(\eigval_n\) and with orthogonal eigenstates~\(\alteigket[1]\), \dots, \(\alteigket[n]\) corresponding to these respective eigenvalues.
Choose any~\(\con_k\in\complexes\) such that \(\sum_{k=1}^n\abs{\con_k}^2=1\), and consider the states~\(\fket[\measurement{A}]\coloneqq\sum_{k=1}^n\con_k\eigket[k]\in\hilbertspace_1\) and \(\fket[{\measurement{B}}]\coloneqq\sum_{k=1}^n\con_{k}\alteigket[k]\in\hilbertspace_2\).
Then \(\utility{A}(\fket[\measurement{A}])=\utility{B}(\fket[{\measurement{B}}])\).\end{proposition}

\begin{proof}
We may assume without loss of generality that there are~\(r\in\naturals\) distinct eigenvalues~\(\alteigval_k\in\reals\) with corresponding multiplicities~\(m_k\in\naturals\), so we have for the eigenvalues of~\(\measurement{A}\) that \(\eigval_\ell=\alteigval_k\) for all~\(\ell\in\set{M_{k-1}+1,\dots,M_{k}}\) and \(k\in\set{1,\dots,r}\), \(M_0\coloneqq0\) and \(M_k\coloneqq\sum_{\ell=1}^km_\ell\).
The eigenspace for~\(\measurement{A}\) that corresponds to the eigenvalue~\(\mu_k\) is then given by \(\eigspace[k]\coloneqq\linspanof{\cset{\eigket[\ell]}{\ell\in\set{M_{k-1}+1,\dots,M_k}}}\).
Correspondingly, we see that the linear space~\(\linspanof{\cset{\alteigket[\ell]}{\ell\in\set{M_{k-1}+1,\dots,M_k}}}\) is a subspace of the eigenspace~\(\alteigspace[k]\) for~\(\measurement{B}\) corresponding to the eigenvalue~\(\alteigval_k\).
The \(r\) orthogonal eigenspaces~\(\eigspace[k]\) span~\(\hilbertspace_1\), and the \(r\) orthogonal eigenspaces~\(\alteigspace[k]\) span~\(\hilbertspace_2\), by \cref{cor:basis:and:eigenvalues}.

For any~\(k\in\set{1,\dots,r}\), choose the complex number~\(\altcon_k\) in such a way that \(\abs{\altcon_k}^2=\sum_{\ell=M_{k-1}+1}^{M_k}\abs{\con_\ell}^2\).
If \(\beta_k\neq0\), then let
\[
\ket{c_k}
\coloneqq\frac1{\beta_k}\sum_{\ell=M_{k-1}+1}^{M_k}\con_\ell\eigket[\ell]
\text{ and }
\ket{d_k}
\coloneqq\frac1{\beta_k}\sum_{\ell=M_{k-1}+1}^{M_k}\con_\ell\alteigket[\ell].
\]
If \(\beta_k=0\), then let \(\ket{c_k}\) be any state in~\(\eigspace[k]\) and \(\ket{d_k}\) any state in~\(\alteigspace[k]\).
It follows that \(\ket{c_k}\) is always a state in~\(\eigspace[k]\) and that \(\ket{d_k}\) is always a state in~\(\alteigspace[k]\), and also that \(\sum_{k=1}^r\abs{\beta_k}^2=\sum_{\ell=1}^n\abs{\con_\ell}^2=1\), \(\fket[\measurement{A}]=\sum_{\ell=1}^n\con_{\ell}\eigket[\ell]=\sum_{k=1}^r\beta_k\ket{c_k}\) and \(\fket[{\measurement{B}}]=\sum_{\ell=1}^n\con_{\ell}\alteigket[\ell]=\sum_{k=1}^r\beta_k\ket{d_k}\).
Now apply \cref{post:dm:different:eigenspaces} [with \(\eigval_k\instantiateas\alteigval_k\), \(\con_k\instantiateas\altcon_k\), \(\eigket[k]\instantiateas\ket{c_k}\) and \(\alteigket[k]\instantiateas\ket{d_k}\)] to conclude that \(\utility{A}(\fket[\measurement{A}])=\utility{B}(\fket[{\measurement{B}}])\).
\end{proof}

To make the proof of \cref{thm:utility} easier to digest, we'll split the argument into several successive propositions that are stepping stones on our way to the main result, and which allow us to uncover the form of the reward function~\(\utility{A}(\gket)\) in increasingly more general types of arguments~\(\gket\in\statespace\).

In a first step, we manage to fix the value of the reward function in equal-amplitude superpositions of eigenstates of a measurement operator.

\begin{proposition}\label{prop:intermediate:step:uniform}
Consider any reward assignation system~\(\assignation\) that obeys Postulates~\cref{post:dm:eigenket,post:dm:different:eigenspaces,post:dm:additivity}; any Hilbert space~\(\hilbertspace\in\allhilbertspaces\), with corresponding reward assignation~\(\utilitymap\coloneqq\assignation(\hilbertspace)\); and any measurement~\(\measurement{A}\in\measurements\), with any orthogonal basis of eigenstates~\(\set{\eigket[1],\dots,\eigket[n]}\).
Then
\begin{equation*}
\utility{A}(\gket[m])=\gbra[m]\measurement{A}\gket[m]
\text{ for all~\(\gket[m]\coloneqq\frac{1}{\sqrt{m}}\sum_{k=1}^m\eigket[k]\), with~\(1\leq m\leq n\)}.
\end{equation*}
\end{proposition}

\begin{proof}
Before we start, infer from~\cref{lem:normalisation} that all~\(\gket[m]\in\statespace\).
Also, denote by \(\eigval_k\) the eigenvalue of~\(\measurement{A}\) corresponding to the eigenstate~\(\eigket[k]\), for all~\(k\in\set{1,\dots,n}\).

We begin with the simplest case that~\(m=1\), so with \(\gket[1]=\eigket[1]\).
\cref{post:dm:eigenket} then guarantees that, on the one hand, \(\utility{A}(\gket[1])=\utility{A}(\eigket[1])=\eigval_1\).
On the other hand, we infer from \cref{lem:quadratic:form:in:basis} that \(\eigval_1=\eigbra[1]\measurement{A}\eigket[1]=\gbra[1]\measurement{A}\gket[1]\).
Hence, indeed, \(\utility{A}(\gket[1])=\gbra[1]\measurement{A}\gket[1]\).

Next, we turn to the more involved case that~\(1<m\leq n\).
We denote by \(\permutation_m\colon\set{1,\dots,n}\to\set{1,\dots,n}\) the cyclic permutation of the first~\(m\) indices, defined by
\[
\permutation_m(k)
\coloneqq
\begin{cases}
k+1&\text{if }k+1\leq m\\
1&\text{if }k=m\\
k&\text{if }m<k\leq n
\end{cases}
\text{ for all }k\in\set{1,2,\dots,n}.
\]
For any~\(\ell\in\set{0,1,\dots,m}\), we let \(\permutation_m^\ell\) denote the result of applying \(\permutation_m\) \(\ell\) times, so \(\permutation_m^\ell\) satisfies
\[
\permutation_m^\ell(k)
=\begin{cases}
k+\ell&\text{if }k+\ell\leq m\\
k+\ell-m&\text{if }m<k+\ell\leq m+\ell\\
k&\text{if }m<k\leq n
\end{cases}
\text{ for all }k\in\set{1,2,\dots,n}.
\]
Also observe that
\begin{equation}\label{eq:cyclic:permutations:inverse}
\group{\permutation_m^\ell}^{-1}
=\permutation_m^{m-\ell}
\text{ for all }\ell\in\set{0,1,2,\dots,m}.
\end{equation}
Using these permutations, we can now construct \(m\) Hermitian operators~\(\indmeasurement{A}{\ell}\in\measurements\) as follows [see \cref{prop:basis:and:eigenvalues}]:
\[
\indmeasurement{A}{\ell}\coloneqq\sum_{k=1}^n\eigval_{k}\eigket[\permutation_m^\ell(k)]\eigbra[\permutation_m^\ell(k)],
\text{ for~\(\ell\in\set{0,\dots,m-1}\)},
\]
is the Hermitian operator with eigenstates~\(\eigket[\permutation_m^\ell(k)]\) corresponding to the respective real eigenvalues~\(\eigval_{k}\) for~\(k\in\set{1,\dots,n}\).
To use \cref{prop:different:eigenvectors}, with \(\hilbertspace_1\instantiateas\hilbertspace\), \(\hilbertspace_2\instantiateas\hilbertspace\), \(\measurement{A}\instantiateas\measurement{A}\), \(\measurement{B}\instantiateas\indmeasurement{A}{\ell}\) and \(\fket[{\measurement{A}}]\instantiateas\gket[m]\) for all \(\ell\in\set{0,1,\dots,m-1}\), we define
\[
\fket[{\measurement{B}}]
=\fket[{\indmeasurement{A}{\ell}}]
\coloneqq\frac{1}{\sqrt{m}}\sum_{k=1}^m\eigket[\permutation_m^\ell(k)]
=\frac{1}{\sqrt{m}}\sum_{k=1}^m\eigket[k]
=\gket[m]
\text{ for all~\(\ell\in\set{0,1,\dots,m-1}\)},
\]
and then \cref{prop:different:eigenvectors} implies that
\[
\utility{A}(\gket[m])
=\utility{A}(\fket[{\measurement{A}}])
=\utility{B}(\fket[{\measurement{B}}])
=\indutility{A}{\ell}(\gket[m])
\text{ for all~\(\ell\in\set{0,1,\dots,m-1}\)},
\]
and therefore also that
\begin{equation}\label{eq:intermediate:step:uniform:one}
\sum_{\ell=0}^{m-1}\indutility{A}{\ell}(\gket[m])=m\utility{A}(\gket[m]).
\end{equation}
Furthermore, because we infer from \Cref{prop:basis:and:eigenvalues} that the elements of the orthonormal basis~\(\set{\eigket[1],\dots,\eigket[n]}\) are the eigenstates for each of the Hermitian operators~\(\indmeasurement{A}{0}\), \(\indmeasurement{A}{1}\), \dots, \(\indmeasurement{A}{m-1}\), we infer from \cref{prop:commuting:operators} that these operators commute: \(\commutator{\indmeasurement{A}{r}}{\indmeasurement{A}{s}}=0\) for all~\(r,s\in\set{0,1,\dots,m-1}\).
But then clearly also \(\commutator{\indmeasurement{A}{r}}{\sum_{k=0}^{r-1}\indmeasurement{A}{k}}=0\) for all~\(r\in\set{1,\dots,m-1}\).
Successively applying \cref{prop:equi_additivity_post} \(m-1\) times, and letting \(\measurement{C}\coloneqq\sum_{\ell=0}^{m-1}\indmeasurement{A}{\ell}\), we then find for the left-hand side of \cref{eq:intermediate:step:uniform:one} that
\begin{equation}\label{eq:intermediate:step:uniform:two}
\sum_{\ell=0}^{m-1}\indutility{A}{\ell}(\gket[m])
=\utility{C}(\gket[m]).
\end{equation}
Moreover, we recall from the definition of~\(\indmeasurement{A}{0}\), \dots, \(\indmeasurement{A}{m-1}\) that
\begin{align*}
\measurement{C}
&=\smashoperator{\sum_{\ell=0}^{m-1}}\indmeasurement{A}{\ell}
=\smashoperator[l]{\sum_{\ell=0}^{m-1}}\smashoperator[r]{\sum_{k=1}^n}\eigval_{k}\eigket[\permutation_m^\ell(k)]\eigbra[\permutation_m^\ell(k)]
\overset{\textrm{\eqref{eq:cyclic:permutations:inverse}}}{=}\smashoperator[l]{\sum_{\ell=0}^{m-1}}\smashoperator[r]{\sum_{r=1}^n}\eigval_{\permutation_m^{m-\ell}(r)}\eigket[r]\eigbra[r]
=\smashoperator[l]{\sum_{s=1}^{m}}\smashoperator[r]{\sum_{r=1}^n}\eigval_{\permutation_m^{s}(r)}\eigket[r]\eigbra[r]\\
&=\smashoperator[l]{\sum_{r=1}^m}\smashoperator[r]{\sum_{s=1}^{m}}\eigval_{\permutation_m^s(r)}\eigket[r]\eigbra[r]
+\smashoperator[l]{\sum_{r=m+1}^n}\smashoperator[r]{\sum_{s=1}^{m}}\eigval_{\permutation_m^s(r)}\eigket[r]\eigbra[r]
=\smashoperator[l]{\sum_{r=1}^m}\group[\Big]{\smashoperator[r]{\sum_{\ell=1}^{m}}\eigval_\ell}\eigket[r]\eigbra[r]
+\smashoperator{\sum_{r=m+1}^n}(m\eigval_r)\eigket[r]\eigbra[r].
\end{align*}
This, together with \cref{prop:basis:and:eigenvalues}, tells us that the Hermitian operator~\(\measurement{C}\) has the same eigenvalue~\(\sum_{\ell=1}^m\eigval_\ell\) corresponding to each of the eigenstates~\(\eigket[k]\) for~\(k\in\set{1,2,\dots,m}\).
Therefore, the linear combination~\(\gket[m]=\frac{1}{\sqrt{m}}\sum_{k=1}^m\eigket[k]\) of these~\(m\) eigenstates will also be an eigenstate of~\(\measurement{C}\) with this same eigenvalue~\(\sum_{\ell=1}^m\eigval_\ell\).
\cref{post:dm:eigenket} then guarantees that \(\utility{C}(\gket[m])=\sum_{\ell=1}^m\eigval_\ell\), so ~\cref{eq:intermediate:step:uniform:two,eq:intermediate:step:uniform:one} then yield that \(\utility{A}(\gket[m])=\frac1m\sum_{\ell=1}^m\eigval_\ell\).
Since, on the other hand, \cref{lem:quadratic:form:in:basis} guarantees that \(\gbra[m]\measurement{A}\gket[m]=\sum_{\ell=1}^m\frac{1}{m}\eigval_\ell\), we're done.
\end{proof}

In a second step, we succeed in fixing the value of the reward function in real linear combinations \(\sum_{k=1}^n\con_k\eigket[k]\) of the eigenstates~\(\eigket[k]\) of a measurement operator~\(\measurement{A}\), provided that the squares \(\alpha_k^2\) of their real amplitudes~\(\alpha_k\) are rational.

\begin{proposition}\label{prop:intermediate:step:rational}
Consider any reward assignation system~\(\assignation\) that obeys Postulates~\cref{post:dm:eigenket,post:dm:different:eigenspaces,post:dm:additivity}, any~\(n\)-dimen\-sional Hilbert space~\(\hilbertspace\in\allhilbertspaces\) with corresponding reward assignation~\(\utilitymap\coloneqq\assignation(\hilbertspace)\) and any measurement~\(\measurement{A}\in\measurements\) with any orthogonal basis of eigenstates~\(\set{\eigket[1],\dots,\eigket[n]}\).
Then
\begin{multline*}
\utility{A}(\gket[m])=\gbra[m]\measurement{A}\gket[m]
\text{ for all~\(\gket[m]\coloneqq\sum_{k=1}^m\sqrt{q_k}\eigket[k]\),}\\
\text{ with~\(1\leq m\leq n\) and~\(q_1,\dots,q_m\in\posrationals\) such that \(\smash{\sum_{k=1}^m}q_k=1\)}.
\end{multline*}
\end{proposition}

\begin{proof}
Before we start, infer from~\cref{lem:normalisation} that all~\(\gket[m]\in\statespace\).
Due to \cref{prop:basis:and:eigenvalues}, we know that \(\measurement{A}=\sum_{k=1}^n\eigval_n\eigket[k]\eigbra[k]\), with \(\eigval_k\) the eigenvalue corresponding to eigenstate~\(\eigket[k]\).

First, since \(q_1,\dots,q_m\) are positive rational numbers, there are natural numbers~\(r\) and \(p_1,\dots,p_m\) such that \(p_k=q_kr\) for all~\(k\in\set{1,\dots,m}\).
Since \(\sum_{k=1}^mq_k=1\), we also have that \(\sum_{k=1}^mp_k=r\), and then clearly \(\gket[m]=\frac{1}{\sqrt{r}}\sum_{k=1}^m\sqrt{p_k}\eigket[k]\).
We'll also consider an arbitrary \(r\)-dimensional Hilbert space, which we denote by~\(\hilbertspace'\).

Define the Hermitian operator~\(\measurement{B}\coloneqq\measurement{A}\tensortimes\identity=\sum_{k=1}^n\eigval_k\eigket[k]\eigbra[k]\tensortimes\identity\) on the tensor product space~\(\hilbertspace\tensortimes\hilbertspace'\).
For the remainder of the argument, we'll now consider an arbitrary but fixed~\(\fket\in\statespace'\).
Then \(\eigket[k]\tensortimes\fket\) is an eigenstate of~\(\measurement{B}=\measurement{A}\tensortimes\identity\) with corresponding eigenvalue~\(\eigval_k\), and all these \(n\) eigenstates~\(\eigket[k]\tensortimes\fket\) are mutually orthogonal.

We now want to make use of \cref{prop:different:eigenvectors}, with \(\hilbertspace_1=\hilbertspace\), \(\hilbertspace_2=\hilbertspace\tensortimes\hilbertspace'\), \(\measurement{A}\instantiateas\measurement{A}\), \(\measurement{B}\instantiateas\measurement{B}\), \(\fket[{\measurement{A}}]\instantiateas\gket[m]\).
We then also let \(\fket[{\measurement{B}}]\instantiateas\gket[m]\tensortimes\fket=\frac{1}{\sqrt{r}}\sum_{k=1}^m\sqrt{p_k}\eigket[k]\tensortimes\fket\).
Since each \(\eigket[k]\tensortimes\fket\) is an eigenstate of~\(\measurement{B}\) corresponding to the eigenvalue~\(\eigval_k\), it follows from \cref{prop:different:eigenvectors} that
\begin{equation}\label{eq:intermediate:step:rational:one}
\utility{A}(\gket[m])=\utility{B}(\gket[m]\tensortimes\fket).
\end{equation}
Now, use \cref{lem:uniform:expansion:in:basis} to construct from the chosen and fixed~\(\fket\in\statespace'\), for any~\(k\in\set{1,\dots,m}\), the \(p_k\)~mutually orthogonal states~\(\indfket{k}{1},\dots,\indfket{k}{p_k}\in\statespace'\) such that \(\fket=\frac{1}{\sqrt{p_k}}\sum_{\ell=1}^{p_k}\indfket{k}{\ell}\) [which is possible since \(0<p_k\leq r\)].
Due to the bi-linearity of the tensor product, we then find that
\begin{equation}\label{eq:intermediate:step:rational:two}
\gket[m]\tensortimes\fket
=\group[\Bigg]{\frac{1}{\sqrt{r}}\sum_{k=1}^m\sqrt{p_k}\eigket[k]}\tensortimes\group[\bigg]{\frac{1}{\sqrt{p_k}}\sum_{\ell=1}^{p_k}\ket{\phi^k_\ell}}
=\frac{1}{\sqrt{r}}\sum_{k=1}^m\sum_{\ell=1}^{p_k}\eigket[k]\tensortimes\ket{\phi^k_\ell}.
\end{equation}
This is now an equal-amplitude superposition of the~\(r=\sum_{k=1}^mp_k\) mutually orthogonal eigenstates~\(\eigket[k]\tensortimes\ket{\phi^k_\ell}\)  of~\(\measurement{B}\) with corresponding eigenvalue~\(\eigval_k\), where \(\ell\in\set{1,\dots,p_k}\) for~\(k\in\set{1,\dots,m}\).
For all \(k\in\set{1,2,\dots,n}\), we define the one-dimensional linear subspaces~\(\subspace_k\coloneqq\linspanof{\set{\eigket[k]}}\).
For any fixed~\(k\in\set{1,2,\dots,m}\), we can use the Gram\textendash Schmidt procedure to extend the orthonormal collection~\(\eigket[k]\tensortimes\ket{\phi^k_1},\eigket[k]\tensortimes\ket{\phi^k_2},\dots,\eigket[k]\tensortimes\ket{\phi^k_{p_k}}\) --- with~\(p_k\) elements --- to an orthonormal basis~\(\basis_k\) for~\(\subspace_k\tensortimes\hilbertspace'\) --- with~\(r\) elements.
Similarly, for any~\(k\in\set{m+1,\dots,n}\), we define an orthonormal basis~\(\basis_k\) for~\(\subspace_k\tensortimes\hilbertspace'\) --- with~\(r\) elements.
As the direct sum of these orthogonal linear subspaces~\(\subspace_k\tensortimes\hilbertspace'\) is the complete Hilbert space~\(\hilbertspace\tensortimes\hilbertspace'\), the union of the bases \(\bigcup^n_{k=1}\basis_k\) is a basis of eigenstates for the Hilbert space~\(\hilbertspace\tensortimes\hilbertspace'\) that contains all the orthogonal eigenstates~\(\eigket[k]\tensortimes\ket{\phi^k_\ell}\) in \cref{eq:intermediate:step:rational:two}.
Since \(\gket[m]\tensortimes\fket=\frac{1}{\sqrt{r}}\sum_{k=1}^m\sum_{\ell=1}^{p_k}\eigket[k]\tensortimes\ket{\phi^k_\ell}\), we're therefore in a position to apply \cref{prop:intermediate:step:uniform} and find that
\begin{align*}
\utility{B}\group{\gket[m]\tensortimes\fket}
&=\group{\gbra[m]\tensortimes\fbra}\measurement{B}\group{\gket[m]\tensortimes\fket}
=\group{\gbra[m]\tensortimes\fbra}\measurement{A}\tensortimes\identity\group{\gket[m]\tensortimes\fket}\\
&=\gbra[m]\measurement{A}\gket[m]\fbra\identity\fket
=\gbra[m]\measurement{A}\gket[m]\fbraket\\
&=\gbra[m]\measurement{A}\gket[m].
\end{align*}
This, combined with \cref{eq:intermediate:step:rational:one}, shows that \(\utility{A}(\gket[m])=\gbra[m]\measurement{A}\gket[m]\), so we're done.
\end{proof}

In the next and penultimate step, we move from amplitudes that are the square roots of rationals to all non-negative real amplitudes.

\begin{proposition}\label{prop:intermediate:step:real}
Consider any reward assignation system~\(\assignation\) that obeys Postulates~\cref{post:dm:eigenket,post:dm:different:eigenspaces,post:dm:additivity,post:dm:continuity}, any Hilbert space~\(\hilbertspace\in\allhilbertspaces\) with corresponding reward assignation~\(\utilitymap\coloneqq\assignation(\hilbertspace)\), and any measurement~\(\measurement{A}\in\measurements\) with any orthogonal basis of eigenstates~\(\set{\eigket[1],\dots,\eigket[n]}\).
Then
\begin{multline*}
\utility{A}(\gket)
=\gbra\measurement{A}\gket
\text{ for all~\(\gket\coloneqq\sum_{k=1}^n\con_k\eigket[k]\in\statespace\)},\\
\text{ with~\(\con_1,\dots,\con_n\in\nonnegreals\) and~\(\smash{\sum_{k=1}^n}\con_k^2=1\)}.
\end{multline*}
\end{proposition}

\begin{proof}
We'll denote by \(\eigval_k\) the eigenvalue of~\(\measurement{A}\) corresponding to the eigenstate~\(\eigket[k]\), for all \(k\in\set{1,\dots,n}\).
Since \(\rationals\) is dense in~\(\reals\), there are \(n\) sequences \(\prescript{}{m}q_k\in\rationals\), \(k\in\set{1,\dots,n}\) such that for all~\(m\in\naturals\): (i) \(\sum_{k=1}^n\prescript{}{m}q_k=1\); (ii) \(\prescript{}{m}q_k>0\) for all~\(k\in\set{1,\dots,n}\); and (iii) \(\lim_{m\to+\infty}\prescript{}{m}q_k=\con_k^2\) for all~\(k\in\set{1,\dots,n}\).
We define a sequence of states as follows: let \(\gket[m]\coloneqq\sum_{k=1}^n\sqrt{\prescript{}{m}q_k}\eigket[k]\) for all~\(m\in\naturals\) .
Then \(\lim_{m\to+\infty}\gket[m]=\gket\), because it follows from our assumptions that \(\lim_{m\to+\infty}\sqrt{\prescript{}{m}q_k}=\con_k\) for all~\(k\in\set{1,\dots,n}\).
Now invoke \cref{post:dm:continuity} to find that, on the one hand,
\begin{equation*}
\utility{A}(\gket)
=\lim_{m\to+\infty}\utility{A}(\gket[m])
=\lim_{m\to+\infty}\utility{A}\group[\bigg]{\sum_{k=1}^n\sqrt{\prescript{}{m}q_k}\eigket[k]}
=\lim_{m\to+\infty}\sum_{k=1}^n\prescript{}{m}q_k\eigval_k
=\sum_{k=1}^n\con_k^2\eigval_k,
\end{equation*}
where for the third equality, we used \cref{prop:intermediate:step:rational,lem:quadratic:form:in:basis}.
On the other hand, we also infer from \cref{lem:quadratic:form:in:basis} that \(\gbra\measurement{A}\gket=\sum_{k=1}^n\abs{\con_k}^2\eigval_k=\sum_{k=1}^n\con_k^2\eigval_k\), so we're done.
\end{proof}

With all these stepping stones in place, it's now fairly straightforward to prove the existence and uniqueness result in~\cref{thm:utility}.

\begin{proof}[Proof of \cref{thm:utility}]
We begin with the existence part: it's enough to show that the reward assignation system~\(\theassignation\) obeys all four postulates~\cref{post:dm:eigenket,post:dm:different:eigenspaces,post:dm:additivity,post:dm:continuity}.

\underline{\upshape{DT\labelcref{post:dm:eigenket}}}.
Fix any Hilbert space~\(\hilbertspace\) and any Hermitian operator~\(\measurement{A}\) on~\(\hilbertspace\).
If \(\eigket\) is an eigenstate associated with an eigenvalue~\(\lambda\) of~\(\measurement{A}\), then, indeed, \(\theutility{A}(\eigket)=\eigbra\measurement{A}\eigket=\braket{a}{\lambda a}=\lambda\eigbraket=\lambda\).

\underline{\upshape{DT\labelcref{post:dm:different:eigenspaces}}}.
Fix any two Hilbert spaces~\(\hilbertspace_1\) and \(\hilbertspace_2\), and consider the Hermitian operator~\(\measurement{A}\coloneqq\sum_{k=1}^r\eigval_k\projection[{\eigspace[k]}]\) on~\(\hilbertspace_1\), with (distinct) real eigenvalues~\(\eigval_1,\dots,\eigval_r\) corresponding to the respective mutually orthogonal eigenspaces~\(\eigspace[1],\dots,\eigspace[r]\) that span~\(\hilbertspace_1\).
Similarly, consider the Hermitian operator~\(\measurement{B}\coloneqq\sum_{k=1}^r\eigval_k\projection[{\alteigspace[k]}]\) on~\(\hilbertspace_2\), with the same (distinct) eigenvalues~\(\eigval_1,\dots,\eigval_r\), corresponding to the respective mutually orthogonal eigenspaces~\(\alteigspace[1],\dots,\alteigspace[r]\) that span~\(\hilbertspace_2\).
Fix any states~\(\eigket[k]\in\eigspace[k]\) and~\(\alteigket[k]\in\alteigspace[k]\), and any~\(\con_k\in\complexes\) such that \(\sum_{k=1}^r\abs{\con_k}^2=1\).
Consider the states~\(\fket[\measurement{A}]\coloneqq\sum_{k=1}^r\con_k\eigket[k]\in\hilbertspace_1\) and \(\fket[{\measurement{B}}]\coloneqq\sum_{k=1}^r\con_{k}\alteigket[k]\in\hilbertspace_2\).
Then it follows from~\cref{lem:quadratic:form:in:basis} that \(\theutility{A}(\fket[\measurement{A}])=\fbra[\measurement{A}]\measurement{A}\fket[\measurement{A}]=\sum_{k=1}^r\modulus{\alpha_k}^2\lambda_k\), and similarly, that \(\theutility{B}(\fket[\measurement{B}])=\fbra[\measurement{B}]\measurement{B}\fket[\measurement{B}]=\sum_{k=1}^r\modulus{\alpha_k}^2\lambda_k\).
Hence, indeed, \(\theutility{A}(\fket[\measurement{A}])=\theutility{B}(\fket[\measurement{B}])\).

\underline{\upshape{DT\labelcref{post:dm:additivity}}}.
Fix any Hilbert space~\(\hilbertspace\), then we find for any two Hermitian operators~\(\measurement{A}\) and~\(\measurement{B}\) on~\(\hilbertspace\) [regardless of whether they have the same eigenspaces or not] that
\[
\uval_{\measurement{A}+\measurement{B}}(\fket)
=\fbra(\measurement{A}+\measurement{B})\fket
=\fbra\measurement{A}\fket+\fbra\measurement{B}\fket
=\theutility{A}(\fket)+\theutility{B}(\fket).
\]

\underline{\upshape{DT\labelcref{post:dm:continuity}}}.
Fix any Hilbert space~\(\hilbertspace\) and any Hermitian operator~\(\measurement{A}\) on~\(\hilbertspace\).
The argument is a standard one, but we include it here for the sake of completeness.
Consider any sequence of states~\(\fket[n]\) with \(\fket=\lim_{n\to+\infty}\fket[n]\), or in other words, \(\norm{\fket-\fket[n]}\to0\).
If we also consider the sequence of kets~\(\gket[n]\coloneqq\measurement{A}\fket[n]\), then
\begin{equation}\label{eq:utility:one}
\norm{\measurement{A}\fket-\gket[n]}
=\norm{\measurement{A}\fket-\measurement{A}\fket[n]}
=\norm{\measurement{A}\group{\fket-\fket[n]}}
\overset{\textrm{\eqref{eq:opnorm:hermitian}}}{\leq}\opnorm{\measurement{A}}\cdot\norm{\fket-\fket[n]}\to0,
\end{equation}
so \(\measurement{A}\fket[n]\to\measurement{A}\fket\): any Hermitian operator is continuous.
Moreover,
\begin{align*}
\abs{\theutility{A}(\fket)-\theutility{A}(\fket[n])}
&=\abs{\fbra\measurement{A}\fket-\fbra[n]\measurement{A}\fket[n]}\\
&=\abs{\fbra\measurement{A}\fket-\fbra\measurement{A}\fket[n]+\fbra\measurement{A}\fket[n]-\fbra[n]\measurement{A}\fket[n]}\\
&\overset{\eqref{eq:hermitian:inner:product}}{=}\abs{\fbra\measurement{A}\fket-\fbra\measurement{A}\fket[n]+\fbra[n]\measurement{A}\fket^*-\fbra[n]\measurement{A}\fket[n]^*}\\
&\leq\abs{\fbra\measurement{A}\fket-\fbra\measurement{A}\fket[n]}+\abs{\fbra[n]\measurement{A}\fket-\fbra[n]\measurement{A}\fket[n]}\\
&=\abs{\fbra\measurement{A}\group{\fket-\fket[n]}}+\abs{\fbra[n]\measurement{A}\group{\fket-\fket[n]}}\\
&\leq2\norm{\measurement{A}\group{\fket-\fket[n]}}\overset{\eqref{eq:utility:one}}{\to}0,
\end{align*}
where the second inequality follows from the Cauchy--Schwartz inequality and the fact that \(\norm{\fket}=\norm{\fket[n]}=1\).
This latest result can also be seen as a direct consequence of the continuity of the inner product in the topology generated by the associated norm.
Hence, indeed \(\theutility{A}(\fket)=\lim_{n\to\infty}\theutility{A}(\fket[n])\).

We can now move to the unicity part.
Consider any reward assignation system~\(\assignation\) and assume that it satisfies postulates~\cref{post:dm:eigenket,post:dm:different:eigenspaces,post:dm:additivity,post:dm:continuity}, then we're going to show that then necessarily~\(\assignation=\theassignation\).
To do this, we're going to consider any Hilbert space~\(\hilbertspace\) and any Hermitian operator~\(\measurement{A}\) on~\(\hilbertspace\), and prove that \(\utility{A}=\theutility{A}\), where, of course, \(\utilitymap=\assignation(\hilbertspace)\).

By \cref{prop:basis:and:eigenvalues}, we know that there's an orthonormal basis for~\(\hilbertspace\) that consists of eigenstates~\(\eigket[1],\dots,\eigket[n]\) of~\(\measurement{A}\).
If we denote by \(\eigval_k\) the eigenvalue of~\(\measurement{A}\) that corresponds with the eigenstate~\(\eigket[k]\), for~\(k\in\set{1,\dots,n}\), then \(\measurement{A}=\sum_{k=1}^n\eigval_k\eigket[k]\eigbra[k]\).

Fix any state~\(\gket\in\statespace\).
Since \(\set{\eigket[1],\dots,\eigket[n]}\) is an orthonormal basis for~\(\hilbertspace\), we know from \cref{lem:normalisation} that there are complex numbers~\(\con_1,\dots,\con_n\) such that \(\sum_{k=1}^n\vert\con_k\vert^2=1\) and \(\gket=\sum_{k=1}^n\con_k\eigket[k]\).
For all~\(k\in\set{1,2,\dots,n}\), let \(\beta_k\coloneqq\modulus{\con_k}\in\nonnegreals\) and \(\theta_k=\arg\con_k\in[0,2\pi)\).
Then \(\con_k=\beta_k e^{\im\theta_k}\) and \(\fket[k]\coloneqq e^{\im\theta_k}\eigket[k]\) is then obviously also an eigenstate of~\(\measurement{A}\) with eigenvalue~\(\eigval_k\).
It's immediate that the~\(\fket[1],\dots,\fket[n]\) also constitute an orthonormal basis for~\(\hilbertspace\) and that
\[
\gket
=\sum_{k=1}^n\con_k\eigket[k]
=\sum_{k=1}^n\beta_ke^{i\theta_k}\eigket[k]
=\sum_{k=1}^n\beta_k\fket[k]
\text{ with }
\sum_{k=1}^n\beta_k^2=1.
\]
We then infer from \cref{prop:intermediate:step:real} that, indeed, \(\utility{A}(\gket)=\gbra\measurement{A}\gket=\theutility{A}(\gket)\).
\end{proof}

\section{Conclusion}\label{sec:conclusion}
We've formulated the problem of uncertainty in a quantum mechanical system as a decision problem between acts that are measurements You can perform on the system.
We've tried to address this decision problem under uncertainty using a number of postulates that fix the form of the reward function~\(\theutility{A}\) associated with each act (or measurement)~\(\measurement{A}\) as \(\theutility{A}=\bra{\bolleke}\measurement{A}\ket{\bolleke}\).
We also show in Appendix~\labelcref{app:POVMs} how to extend this result to the more general case of POVMs.

The upshot of this is that \emph{Born's rule turns out to be already incorporated in the reward or the utility aspect of the decision problem and is thereby to some extent freed from its purely probabilistic connotations:} we can, in a de Finetti-like approach, separate the utilities and the probabilities, and mathematically express Your preferences as coherent preference orderings on the uncertain rewards, or equivalently, on the measurements.
Such coherent preference orderings, as we have seen, can be represented using coherent lower previsions on the real Hilbert space of all measurements.
Coherent lower previsions can also be represented mathematically by closed convex sets of density operators; only in special cases do these sets reduce to single density operators and to the standard probabilistic models used in quantum mechanics.
In this sense, the decision problem is primary and probabilities are derivative.

To try and drive home this point, let's backtrack a bit and look at \emph{general} decision problems and the role that coherent previsions play there, as described in \cref{sec:decision:theoretic:background}.
Generally speaking, any coherent prevision on the real linear space~\(\gambles[\Omega]\) of \emph{all bounded real-valued functions} on a state space~\(\Omega\) can always be seen as an expectation operator~\(E\) with respect to a (finitely additive) probability measure on that space: the restrictions of the coherent prevision to the indicator functions associated with the subsets of the state space; see for instance Refs.~\cite{walley1991,troffaes2013:lp,augustin2013:itip}.
This is the way that such probability measures can be made to help in characterising preferences, as is made clear in \cref{eq:decision:via:expected:utility}.
But in this characterisation of preferences, the role of probabilities is --- to repeat what we said above --- merely derivative.

This derivative character of probabilities becomes all the more apparent in the \emph{more specific} present context, where the coherent previsions are defined on the real linear space~\(\measurements\) of \emph{all measurements}~\(\measurement{A}\), or equivalently, on the real linear space~\(\theutilities\) of all bounded real-valued uncertain rewards~\(\theutility{A}\) on the state space~\(\statespace\).
This~\(\theutilities\) is only a linear subspace of the linear space~\(\gambles\) of all \emph{bounded} real-valued maps on~\(\statespace\): the uncertain rewards~\(\theutility{A}\) are `quadratic' functions of the system state and can in no way be associated with indicators of subsets of the state space~\(\statespace\), as \(\theutilities\) contains no indicator functions.
The coherence conditions make sure that there are probability measures on the state space~\(\statespace\) such that a given coherent prevision~\(\linprev\) on~\(\measurements\) coincides  on the quadratic functions in~\(\theutilities\) with the expectation operator associated with that probability measure, but these probability measures are in no way necessary to deal with the decision problem: the coherent prevision~\(\linprev\) on~\(\measurements\) suffices and the role of probabilities is, it bears repeating, derivative.
Due to the nature of quantum-mechanical decision problems, probabilities have no foundational part in their treatment --- unless a confusing one --- but coherent (lower and upper) previsions on the real Hilbert space of all measurements, and equivalently, (sets of) density operators, do.
On this view, probabilities are mere artefacts that turn out to appear and be useful in some decision problems, but not in others.

Incidentally, some of the stranger problems associated with using classical probabilities in quantum mechanics, such as the violation of Bell's inequalities, or their incompatibility with the Tsirelson bound, can be simply explained (away) by taking into account the consequences of the geometry of the Hilbert space of measurements, as is elaborately discussed and argued in Ref.~\cite[Secs.~2.6 and~3.4]{janas2021:raffles}.
It turns out that we've been doing this here by considering specific price functionals --- coherent (lower and upper) previsions --- \emph{on this Hilbert space}.
Working with coherent (lower and upper) previsions on measurements, rather than classical probabilities on the state space, also fits in perfectly with the geometric focus in the work of de Finetti \cite{finetti1937:correlazione} and Fisher \cite{fisher1924} and solves the above-mentioned problems.
It seems to us one of the simplest ways of justifying, and working with, `quantum probabilities'.

Finally, we want to point out that when we use coherent sets of desirable \emph{measurements} to express Your uncertainty, we're essentially relying on the Heisenberg picture of quantum mechanics.
And, since density operators are mixed quantum \emph{states}, we see that sets of density operators fit well within the Schrödinger picture.
The (almost-)equivalence between working with \emph{sets of measurements} on the one hand in \cref{sec:sets:of:desirable:measurements}, and working with \emph{sets of density operators} on the other in \cref{sec:lower:upper:previsions,sec:coherent:previsions}, allows us to recover the well-known duality between the Heisenberg and Schrödinger pictures in our decision-theoretic approach as well.

Our argumentation also enables us to get rid of the totality requirement on Your preference ordering whenever You're uncertain about the state of the quantum system.
It therefore allows for a more general treatment of quantum-mechanical uncertainty, that still gets back to Born's rule under the right (and usual) circumstances; a treatment of uncertainty that fits in perfectly with the more recent developments in the field of imprecise probabilities, that allows us to deal with partial probability assessments and that looks at probabilistic inference as a special case of conservative (deductive) inference.
To put it very succinctly, on our way of looking at it in this paper, quantum mechanics doesn't enforce a probability model for the uncertainty, as all the peculiar quantum mechanical aspects of the associated decision problem are captured in the reward functions and Your preference orderings can therefore be allowed to be only partial.
This leaves us with a lot of freedom to accommodate for other uncertainty models than precise probability models, which could (but needn't) be seen as partial and less perfect approximations to these precise models.

We've already taken the opportunity to point out reasons why such more general models can be useful: they can be more realistic for a You with limited time, resources and information, and who therefore can't be expected to come up with a total ordering of the infinity of acts associated with the decision problem; and the conservative inference mechanism associated with them allows You to mathematically draw all the inferences that can be drawn from the preferences You've actually expressed; see, for instance, the various instalments of our qubit running example.

But, there are, next to these more fundamental reasons for using more general models, also practical ones.
Letting go of precision (or totality, or completeness) can make a problem that is computationally too complex or expensive more manageable.
One interesting case in point is lumping, where a reduction of the dimension of the state space reduces the computational complexity of making inferences about a quantum system, but introduces imprecision.
Working with imprecise probability models and techniques will then allow us to turn a computationally hard problem into a more manageable one, at the cost of our only being able to find conservative bounds on quantities of interest whose exact calculation is too expensive.
In this way, our discussion here provides the foundation for practical applications in quantum mechanical inference and for importing into quantum mechanics ideas that have proved useful in classical dynamical systems theory, such as for instance, lumping in Markov chains \cite{erreygers:2021:thesis,erreygers2019:lumping,erreygers2018:spectrum:allocation}.
Before an attempt can be made at using our approach for solving such problems, we must be able to import the essentially static considerations in this paper into a dynamical context, where the quantum state of a system evolves in time.
This is the subject of our current research.

There are a number of issues and questions that remain, and which we feel to be worthy of attention.
First, it might be argued that our approach still isn't general enough, in that we still, in our decision-theoretic background assumption~\cref{backass:dm:utility:function}, assume the existence of a reward function~\(\theutility{A}\), which \emph{in effect} still imposes a \emph{total} strict ordering on measurements when the system is in a known state~\(\uket=\gket\), because then \(\measurement{A}\betterthan\measurement{B}\ifandonlyif\theutility{A}(\gket)>\theutility{B}(\gket)\).
Would it be possible to let go of even this {\itshape a priori} assumption and to then try and recover it {\itshape a posteriori} on the basis of other decision-theoretic assumptions, in the spirit of our postulates~\cref{post:dm:eigenket,post:dm:different:eigenspaces,post:dm:additivity,post:dm:continuity}?
A second question is whether similar conclusions to ours, and similar justifications for using imprecise probabilities in quantum mechanics can be reached when working, as Savage (and Wallace) did, in a decision-theoretic context where the utilities aren't supposed to exist as extraneous to the decision problem, but have to be constructed from the (partial) preference relation on acts.
Third, it would be interesting to try and recover similar results without relying on the Hilbert space structure of the state space, and in doing so find a decision-theoretic foundation for quantum mechanics that doesn't rely on the Hilbert space structure as fundamental.
One possible approach could start with the algebraic formulation of quantum mechanics \cite{takesaki2003,kaduson1986}, where observables are the fundamental objects living in a \(C^{*}\)-algebra, and the Hilbert space structure is determined by the GNS representation of the algebra \cite[Theorem~4.5.2]{kaduson1986}.
A problem with this suggestion, however, is that to get this result, probabilities are again seen as fundamental.
Nevertheless, through the use of symmetry arguments and similar postulates, we might still be able to arrive at coherence axioms for sets of desirable measurements, and we are currently exploring these ideas with some determination.

To conclude, we've already pointed out on several occasions that our work here can be used to provide a decision-theoretic foundation to ideas about introducing imprecise probability models (desirable gambles, lower and upper previsions, sets of density operators) in quantum mechanics, as first proposed by \citeauthor{benavoli2016:quantum} in Ref.~\cite{benavoli2016:quantum}.
As pointed out in a number of relevant places in \cref{sec:desirability,sec:lower:upper:previsions,sec:coherent:previsions}, we recover a number of their results, even if our interpretation of them may differ.
They have developed their ideas further in a different direction in Ref.~\cite{benavoli2019:computational}, where they provide an alternative way of thinking about, and justifying, imprecise probability models in quantum mechanics.
In contrast to what we do here, they don't follow the practice, standard in quantum mechanics, of using the tensor product space~\(\tensortimes^m_{k=1}\hilbertspace_k\) of the particle state spaces~\(\hilbertspace_k\) to represent the system state.
Instead, they essentially use only a subset of this space, namely the Cartesian product~\(\times^m_{k=1}\hilbertspace_k\).
On this alternative but smaller set of quantum states of the type~\(x\coloneqq(\fket[1],\dots,\fket[m])\in\times^m_{k=1}\hilbertspace_k\) they consider all quadratic gambles, which are defined as functions of the form
\[
g_{\measurement{A}}(\fket[1],\dots,\fket[m])
\coloneqq\group[\big]{\tensortimes^m_{k=1}\fbra[k]}\measurement{A}\group[\big]{\tensortimes^m_{k=1}\fket[k]},
\]
also symbolically written as \(g_{\measurement{A}}(x)=x^\dagger\measurement{A}x\), corresponding to the Hermitian operators~\(\measurement{A}\) on~\(\tensortimes^m_{k=1}\hilbertspace_k\).
They then introduce a concept of algorithmic rationality that leads to coherence axioms and a framework of desirability that closely resembles ours in spirit, but is rather different in the mathematical details.
Their gambles \(g_{\measurement{A}}\) correspond to our reward functions, but essentially restricted to the smaller Cartesian product; it's easy to see that there's a one-to-one correspondence between them, as they're both isomorphic to the space of Hermitian operators.\footnote{When there's only one particle, the state space will be identical in both frameworks.}
They furthermore take the quadratic shape of their gambles as a given, whereas we derive it from a set of postulates.
We therefore feel justified in claiming that our approach here is rather different and deserving of separate consideration.

\section*{Author contributions}
This paper originated in many intensive discussions amongst the four of us and Natan T'Joens in the context of Keano De Vos's master dissertation, in the late autumn and winter of 2020 and early spring of 2021, during the second Covid-19 lockdown.
Keano wrote a first draft of the paper, which was later extensively revised, expanded, reorganised and added to by Gert, while Alexander and Jasper focused on thoroughly reviewing and commenting on the intermediate and final stages and showering Keano and Gert with constructive criticism.
Gert and Keano also worked on revising the paper in response to the reviews.

\section*{Acknowledgments}
Gert's research on this paper was partly supported by a sabbatical grant from Ghent University, and from the FWO, reference number K801523N.
He is grateful to Jason Konek for funding several short research stays, as well as a one-month sabbatical stay, at University of Bristol's Department of Philosophy, which allowed him to work on this topic and present results to various people in the department, in the context of Jason Konek's ERC Starting Grant ``Epistemic Utility for Imprecise Probability'' under the European Union's Horizon 2020 research and innovation programme (grant agreement no. 852677).
He also wishes to express his gratitude to Teddy Seidenfeld, whose generous funding helped realise a sabbatical stay at Carnegie Mellon University's Department of Philosophy, during which some of the revision work for this paper was done.

We'd like to express our gratitude to the reviewers for taking the time to read this lengthy paper.
Some of them responded with very detailed and constructive criticism, which allowed us to improve the readability and the flow of the argumentation, and to expand its scope to dealing with POVMs.

\printbibliography

@article{fisher1924,
  author  = {Fisher, Ronald A.},
  journal = {Metron},
  pages   = {329--332},
  title   = {The distribution of the partial correlation coefficient},
  volume  = {3},
  year    = {1924},
}

@book{helstrom1976,
  address   = {New York},
  author    = {Helstrom, Carl W.},
  publisher = {Academic Press},
  series    = {Mathematics in Science and Engineering},
  title     = {Quantum Detection and Estimation Theory},
  volume    = {123},
  year      = {1976},
}

@book{holevo2001,
  address   = {Berlin},
  author    = {Holevo, Alexander S.},
  publisher = {Springer},
  series    = {Lecture Notes in Physics Monographs},
  title     = {Statistical Structure of Quantum Theory},
  volume    = {67},
  year      = {2001},
}

@book{saunders2018:many:worlds,
  author    = {Saunders, Simon and Barrett, Jonathan and Kent, Adrian and Wallace, David},
  doi       = {10.1093/acprof:oso/9780199560561.001.0001},
  month     = {6},
  publisher = {Oxford University Press},
  title     = {Many Worlds? Everett, Quantum Theory, and Reality},
  year      = {2010},
}

@article{wallace2003:defending:deutsch,
  author   = {Wallace, David},
  doi      = {10.1016/s1355-2198(03)00036-4},
  issn     = {1355-2198},
  journal  = {Studies in History and Philosophy of Science Part B: Studies in History and Philosophy of Modern Physics},
  number   = {3},
  pages    = {415--439},
  title    = {Everettian rationality: defending Deutsch's approach to probability in the Everett interpretation},
  volume   = {34},
  year     = {2003},
}

@article{wallace2007:improving:deutsch,
  author   = {Wallace, David},
  doi      = {10.1016/j.shpsb.2006.04.008},
  issn     = {1355-2198},
  journal  = {Studies in History and Philosophy of Science Part B: Studies in History and Philosophy of Modern Physics},
  number   = {2},
  pages    = {311--332},
  title    = {Quantum probability from subjective likelihood: Improving on Deutsch's proof of the probability rule},
  volume   = {38},
  year     = {2007},
}

@article{wallace2009:born:arxiv,
  author   = {Wallace, David},
  doi      = {10.48550/arxiv.0906.2718},
  eprint   = {0906.2718},
  journal  = {arXiv},
  title    = {A formal proof of the {B}orn rule from decision-theoretic assumptions},
  year     = {2009},
}

@book{Walley1991,
  author    = {Walley, Peter},
  publisher = {Chapman \& Hall},
  title     = {Statistical Reasoning with Imprecise Probabilities},
  year      = {1991},
}

@article{anscombe1963,
  author  = {Anscombe, Francis J. and Aumann, Robert J.},
  journal = {The Annals of Mathematical Statistics},
  pages   = {199--205},
  title   = {A definition of subjective probability},
  volume  = {34},
  year    = {1963},
}

@book{augustin2013:itip,
  editor    = {Augustin, Thomas and Coolen, Frank P. A. and prefix=de, family=Cooman, given=Gert and
               Troffaes, Matthias C. M.},
  publisher = {John Wiley \& Sons},
  title     = {Introduction to Imprecise Probabilities},
  year      = {2014},
}

@article{aumann1962,
  author  = {Aumann, Robert J.},
  journal = {Econometrica},
  pages   = {445--462},
  title   = {Utility theory without the completeness axiom},
  volume  = {30},
  year    = {1962},
}

@article{aumann1964,
  author  = {Aumann, Robert J.},
  journal = {Econometrica},
  pages   = {210--212},
  title   = {Utility theory without the completeness axiom: a correction},
  volume  = {32},
  year    = {1964},
}

@article{barnum2000:quantum,
  author   = {Barnum, Howard and Caves, Carlton M. and Finkelstein, Joshua and Fuchs, Christopher A. and Schack, Rüdiger},
  doi      = {10.1098/rspa.2000.0557},
  journal  = {Proceedings of the Royal Society of London. Series A: Mathematical, Physical and Engineering Sciences},
  month    = {5},
  number   = {1997},
  pages    = {1175--1182},
  title    = {Quantum probability from decision theory?},
  volume   = {456},
  year     = {2000},
}

@article{barrett2007,
  author    = {Barrett, Jonathan},
  doi       = {10.1103/PhysRevA.75.032304},
  issue     = {3},
  journal   = {Phys. Rev. A},
  month     = {3},
  numpages  = {21},
  publisher = {American Physical Society},
  title     = {Information processing in generalized probabilistic theories},
  volume    = {75},
  year      = {2007},
}

@article{benavoli2016:quantum,
  archiveprefix = {arXiv},
  author        = {Benavoli, Alessio and Facchini, Alessandro and Zaffalon, Marco},
  doi           = {10.1103/PhysRevA.94.042106},
  eprint        = {1605.08177},
  journal       = {Physical Review A},
  month         = {10},
  number        = {4},
  primaryclass  = {quant-ph},
  shorttitle    = {Quantum mechanics},
  title         = {Quantum mechanics: the Bayesian theory generalised to the space of {Hermitian} matrices},
  volume        = {94},
  year          = {2016},
}

@unpublished{benavoli2019:computational,
  author   = {Benavoli, Alessio and Facchini, Alessandro and Zaffalon, Marco},
  journal  = {arXiv:1902.04569 [quant-ph]},
  note     = {arXiv: 1902.04569},
  title    = {Computational complexity and the nature of quantum mechanics},
  year     = {2019},
}

@article{boyd:2004:convex:optimization,
  author   = {Boyd, Stephen and Vandenberghe, Lieven},
  journal  = {Cambridge University Press},
  language = {English},
  title    = {Convex optimization},
  year     = {2004},
}

@article{Chiribella2011,
  author    = {Chiribella, Giulio and D'Ariano, Giacomo Mauro and Perinotti, Paolo},
  doi       = {10.1103/PhysRevA.84.012311},
  issue     = {1},
  journal   = {Phys. Rev. A},
  month     = {7},
  numpages  = {39},
  publisher = {American Physical Society},
  title     = {Informational derivation of quantum theory},
  volume    = {84},
  year      = {2011},
}

@book{cohen1977:quantum:1,
  author    = {Cohen-Tannoudji, Claude and Laloe, Frank and Diu, Bernard},
  publisher = {Wiley},
  title     = {Quantum Mechanics},
  volume    = {1},
  year      = {1977},
}

@article{cooman2003a,
  author  = {prefix=de, family=Cooman, given=Gert},
  doi     = {10.1007/s10472-005-9006-x},
  journal = {Annals of Mathematics and Artificial Intelligence},
  pages   = {5--34},
  title   = {Belief models: an order-theoretic investigation},
  volume  = {45},
  year    = {2005},
}

@article{cooman2010,
  author  = {prefix=de, family=Cooman, given=Gert and Quaeghebeur, Erik},
  journal = {International Journal of Approximate Reasoning},
  note    = {Special issue in honour of Henry E.~Kyburg, Jr.},
  number  = {3},
  pages   = {363--395},
  title   = {Exchangeability and sets of desirable gambles},
  volume  = {53},
  year    = {2012},
}

@article{cooman2021:archimedean:choice,
  author  = {prefix=de, family=Cooman, given=Gert},
  doi     = {10.1016/j.ijar.2021.09.005},
  journal = {International Journal of Approximate Reasoning},
  pages   = {255--281},
  title   = {Coherent and {A}rchimedean choice in general {B}anach spaces},
  volume  = {140},
  year    = {2022},
}

@article{couso2011:desirable,
  author   = {Couso, Ines and Moral, Serafín},
  doi      = {10.1016/j.ijar.2011.08.001},
  journal  = {International Journal of Approximate Reasoning},
  number   = {7},
  title    = {Sets of desirable gambles: Conditioning, representation, and precise probabilities},
  volume   = {52},
  year     = {2011},
}

@phdthesis{debock2015:thesis,
  author = {De Bock, Jasper},
  school = {Ghent University, Faculty of Engineering and Architecture},
  title  = {Credal Networks under Epistemic Irrelevance: Theory and Algorithms},
  year   = {2015},
}

@inproceedings{debock2018,
  archiveprefix = {Extended arXiv version},
  author        = {De Bock, Jasper and prefix=de, family=Cooman, given=Gert},
  booktitle     = {Uncertainty Modelling in Data Science (Proceedings of SMPS 2018)},
  eprint        = {1806.01044},
  pages         = {46--53},
  title         = {A desirability-based axiomatisation for coherent choice functions},
  year          = {2018},
}

@article{debock2019:interpretation,
  archiveprefix = {Extended arXiv version},
  author        = {De Bock, Jasper  and prefix=de, family=Cooman, given=Gert},
  eprint        = {1903.00336},
  journal       = {Proceedings of Machine Learning Research},
  pages         = {125--134},
  title         = {Interpreting, axiomatising and representing coherent choice functions in terms of desirability},
  volume        = {103},
  year          = {2019},
}

@article{decooman2015:coherent:predictive:inference,
  author  = {prefix=de, family=Cooman, given=Gert and De Bock, Jasper and Diniz, M\'arcio Alves},
  journal = {Journal of Artificial Intelligence Research},
  pages   = {1--95},
  title   = {Coherent predictive inference under exchangeability with imprecise probabilities},
  volume  = {52},
  year    = {2015},
}

@article{decooman2024:things:logic,
  author   = {prefix=de, family=Cooman, given=Gert and Van Camp, Arthur and De Bock, Jasper},
  doi      = {10.1016/j.ijar.2024.109241},
  journal  = {International Journal of Approximate Reasoning},
  pages    = {109241},
  title    = {Desirable sets of things and their logic},
  volume   = {172},
  year     = {2024},
}

@article{deutsch1999:quantum:decisions,
  author  = {Deutsch, David},
  doi     = {10.1098/rspa.1999.0443},
  journal = {Proceedings of the Royal Society of London. Series A: Mathematical, Physical and Engineering Sciences},
  month   = {8},
  number  = {1988},
  pages   = {3129--3137},
  title   = {Quantum theory of probability and decisions},
  volume  = {455},
  year    = {1999},
}

@book{dirac1981:principles,
  author    = {Dirac, Paul A. M.},
  publisher = {Oxford University Press},
  title     = {The Principles of Quantum Mechanics},
  year      = {1981},
}

@thesis{erreygers:2021:thesis,
  addendum = {Available at \url{https://users.ugent.be/~aerreyge/}},
  author   = {Erreygers, Alexander},
  school   = {Ghent University},
  title    = {Markovian Imprecise Jump Processes: Foundations, Algorithms and Applications},
  type     = {Doctoral dissertation},
  year     = {2021},
}

@article{erreygers2018:spectrum:allocation,
  author  = {Erreygers, Alexander and Rottondi, Cristina and Verticale, Giacomo and De Bock, Jasper},
  doi     = {10.1109/TCOMM.2018.2846235},
  journal = {IEEE Transactions on Communications},
  number  = {11},
  pages   = {5401--5414},
  title   = {Imprecise Markov Models for Scalable and Robust Performance Evaluation of Flexi-Grid Spectrum Allocation Policies},
  volume  = {66},
  year    = {2018},
}

@article{erreygers2019:lumping,
  author  = {Erreygers, Alexander and De Bock, Jasper},
  doi     = {10.1016/j.ijar.2019.09.003},
  journal = {International Journal of Approximate Reasoning},
  pages   = {96--133},
  title   = {Bounding inferences for large-scale continuous-time Markov chains: A new approach based on lumping and imprecise Markov chains},
  volume  = {115},
  year    = {2019},
}

@incollection{faye2019:copenhagen,
  author       = {Faye, Jan},
  booktitle    = {The {Stanford} Encyclopedia of Philosophy},
  edition      = {{W}inter 2019},
  howpublished = {\url{https://plato.stanford.edu/archives/win2019/entries/qm-copenhagen/}},
  publisher    = {Metaphysics Research Lab, Stanford University},
  title        = {Copenhagen interpretation of quantum mechanics},
  year         = {2019},
}

@article{finetti1937,
  author  = {{d}e Finetti, Bruno},
  journal = {Annales de l'Institut Henri Poincar\'e},
  note    = {{E}nglish translation in \cite{kyburg1964}},
  pages   = {1--68},
  title   = {La pr\'evision: ses lois logiques, ses sources subjectives},
  volume  = {7},
  year    = {1937},
}

@article{finetti1937:correlazione,
  author  = {{d}e Finetti, Bruno},
  journal = {Economia},
  number  = {1--2},
  title   = {A proposito di correlazione},
  volume  = {3},
  year    = {1937},
}

@book{finetti1970,
  address   = {Turin},
  author    = {{d}e Finetti, Bruno},
  publisher = {Einaudi},
  title     = {Teoria delle Probabilit\`a},
  year      = {1970},
}

@book{finetti19745,
  address   = {Chichester},
  author    = {{d}e Finetti, Bruno},
  note      = {{E}nglish translation of \cite{finetti1970}, two volumes},
  publisher = {John Wiley \& Sons},
  title     = {Theory of Probability: A Critical Introductory Treatment},
  year      = {1974--1975},
}

@article{fuchs2014:introduction,
  author  = {Fuchs, Christopher A. and Mermin, N. David and Schack, Rüdiger},
  doi     = {10.1119/1.4874855},
  journal = {American Journal of Physics},
  number  = {8},
  pages   = {749--754},
  title   = {An introduction to {QBism} with an application to the locality of quantum mechanics},
  volume  = {82},
  year    = {2014},
}

@article{gleason,
  author  = {Gleason, A. M.},
  journal = {Journal of Mathematical Mechanics},
  number  = {6},
  pages   = {885--893},
  title   = {Measures on the closed subspaces of a Hilbert space},
  volume  = {6},
  year    = {1957},
}

@misc{hardy2001,
  archiveprefix = {arXiv},
  author        = {Lucien Hardy},
  eprint        = {quant-ph/0101012},
  title         = {Quantum Theory From Five Reasonable Axioms},
  year          = {2001},
}

@inbook{hardy2002,
  author    = {Lucien Hardy},
  booktitle = {Non-locality and Modality},
  chapter   = {4},
  editor    = {T. Placek and J. Butterfield},
  pages     = {61--74},
  publisher = {Springer},
  title     = {Why Quantum Theory?},
  year      = {2002},
}

@book{janas2021:raffles,
  author    = {Janas, Michael and Cuffaro, Michael E. and Janssen, Michel},
  publisher = {Springer},
  series    = {Boston Studies in the Philosophy and History of Science},
  title     = {Understanding Quantum Raffles: Quantum Mechanics on an Informational Approach: Structure and Interpretation},
  vol       = {340},
  year      = {2021},
}

@book{kaduson1986,
  author    = {Kaduson, R. V.},
  publisher = {Academic Press},
  title     = {An Introduction to the Theory of Operator Algebras},
  volume    = {1: Elementary Theory},
  year      = {1986},
}

@book{kitaev2002:classical,
  author    = {Kitaev, Alexei Yu and Shen, Alexander and Vyalyi, Mikhail N. and Vyalyi, Mikhail N.},
  doi       = {10.1090/gsm/047},
  number    = {47},
  publisher = {American Mathematical Society},
  title     = {Classical and Quantum Computation},
  year      = {2002},
}

@book{levi1980a,
  address   = {London},
  author    = {Levi, Isaac},
  publisher = {MIT Press},
  title     = {The Enterprise of Knowledge},
  year      = {1980},
}

@article{moral2005b,
  author  = {Moral, Serafín},
  doi     = {10.1007/s10472-005-9011-0},
  journal = {Annals of Mathematics and Artificial Intelligence},
  pages   = {197--214},
  title   = {Epistemic irrelevance on sets of desirable gambles},
  volume  = {45},
  year    = {2005},
}

@article{naimark1943,
  author  = {Naimark, Mark Aronovich},
  journal = {C. R. (Doklady) Acad. Sci. URSS (N.S.)},
  pages   = {359--361},
  title   = {On a representation of additive operator set functions},
  volume  = {41},
  year    = {1943},
}

@article{nau2006,
  author  = {Nau, Robert},
  doi     = {10.1214/009053606000000740},
  journal = {Annals Of Statistics},
  number  = {5},
  pages   = {2430--2448},
  title   = {The shape of incomplete preferences},
  volume  = {34},
  year    = {2006},
}

@book{nielsen2010:quantum,
  author    = {Nielsen, Michael A. and Chuang, Isaac L.},
  doi       = {10.1017/CBO9780511976667},
  place     = {Cambridge},
  publisher = {Cambridge University Press},
  title     = {Quantum Computation and Quantum Information: 10th Anniversary Edition},
  year      = {2010},
}

@inbook{quaeghebeur2012:itip,
  author    = {Quaeghebeur, Erik},
  booktitle = {Introduction to Imprecise Probabilities},
  editor    = {Augustin, Thomas and Coolen, Frank P. A. and prefix=de, family=Cooman, given=Gert and
               Troffaes, Matthias C. M.},
  publisher = {John Wiley \& Sons},
  title     = {Desirability},
  year      = {2014},
}

@article{quaeghebeur2015:statement,
  author  = {Quaeghebeur, Erik and prefix=de, family=Cooman, given=Gert and Hermans, Filip},
  journal = {International Journal of Approximate Reasoning},
  pages   = {69--102},
  title   = {Accept {\&} reject statement-based uncertainty models},
  volume  = {57},
  year    = {2015},
}

@book{reed2003:methods:1,
  author    = {Reed, Michael and Simon, Barry},
  doi       = {10.1016/B978-0-12-585001-8.X5001-6},
  publisher = {Academic Press},
  title     = {Methods of Modern Mathematical Physics: Functional analysis},
  volume    = {I},
  year      = {2003},
}

@book{savage1972,
  address   = {New York},
  author    = {Savage, Leonard J.},
  note      = {Second revised edition, first published 1954},
  publisher = {Dover},
  title     = {The Foundations of Statistics},
  year      = {1972},
}

@book{schechter1997,
  address   = {San Diego, CA},
  author    = {Schechter, Eric},
  publisher = {Academic Press},
  title     = {Handbook of Analysis and Its Foundations},
  year      = {1997},
}

@article{seidenfeld1995,
  author  = {Seidenfeld, Teddy and Schervish, Mark J. and Kadane, Jay B.},
  journal = {The Annals of Statistics},
  note    = {Reprinted in \cite{seidenfeld1999}, pp.~69--129},
  pages   = {2168--2217},
  title   = {A representation of partially ordered preferences},
  volume  = {23},
  year    = {1995},
}

@book{takesaki2003,
  author    = {Takesaki, Masamichi},
  doi       = {10.1007/978-3-662-10453-8},
  publisher = {Springer},
  title     = {Theory of Operator Algebras},
  volume    = {II},
}

@book{troffaes2013:lp,
  author    = {Troffaes, Matthias C. M. and prefix=de, family=Cooman, given=Gert},
  publisher = {Wiley},
  title     = {Lower Previsions},
  year      = 2014,
}

@article{walley2000,
  author  = {Walley, Peter},
  journal = {International Journal of Approximate Reasoning},
  pages   = {125--148},
  title   = {Towards a unified theory of imprecise probability},
  volume  = {24},
  year    = {2000},
}

@article{zaffalon2017:incomplete:preferences,
  author  = {Zaffalon, Marco and Miranda, Enrique},
  journal = {Journal Of Artificial Intelligence Research},
  pages   = {1057--1126},
  title   = {Axiomatising Incomplete Preferences through Sets of Desirable Gambles},
  volume  = {60},
  year    = {2017},
}

\appendix
\section{Glossary of main notation and concepts}\label{app:glossary}
\begin{longtable}{p{2.75cm}p{9.25cm}}
\textbf{Notation} & \textbf{Definition}\\[.75ex]
\endfirsthead
\multicolumn{2}{r}{\textit{(continued from previous page)}} \\
\textbf{Notation} & \textbf{Definition} \\[.75ex]
\endhead
\multicolumn{2}{r}{\textit{(continued on next page)}} \\
\endfoot
\endlastfoot

\multicolumn{2}{l}{\textit{Hilbert-space language}}\\[.35ex]
\(\hilbertspace\) & Hilbert space associated with the quantum system\\
\(\statespace\) & state space: all normalised kets in~\(\hilbertspace\)\\
\(\gket,\fket\) & generic kets representing possible quantum states\\
\(\gbra,\fbra\) & dual bras associated with kets via the inner product\\
\(\braket{\psi}{\phi}=\inprod{\gket}{\fket}\) & inner product between two kets~\(\gket\) and~\(\fket\)\\
\(\linspanof{\assessment}\) & linear span of a set of kets~\(\assessment\)\\
\(\basisket\) & generic notation for a basis ket\\
\\

\multicolumn{2}{l}{\textit{Operators, measurements, and spectra}} \\[.35ex]
\(\measurement{A}\) & Hermitian operator; interpreted as a measurement (act)\\
\(\measurements\), \(\measurements(\hilbertspace)\) & set of all Hermitian operators on~\(\hilbertspace\)\\
\(\adjoint{\operator{O}}\) & Hermitian adjoint of linear operator~\(\operator{O}\)\\
\(\eigval\), \(\eigket\) & generic notation for an eigenvalue and corresponding eigenket of a linear operator\\
\(\spectrum{A}\) & spectrum of~\(\measurement{A}\): set of eigenvalues\\[.05ex]
\(\projection[\eigspace]\) & projection operator on the linear subspace~\(\eigspace\)\\
\(\trace{\measurement{A}}\) & trace of~\(\measurement{A}\)\\
\\

\multicolumn{2}{l}{\textit{States, rewards, and decision language}} \\[.35ex]
\(\uket\) & unknown quantum state\\
\(\utility{A}\) & reward function for act~\(\act{A}\)/measurement~\(\measurement{A}\), mapping states to utiles\\
\(\theutility{A}\) & uniquely determined reward function based on the reward postulates: \(\theutility{A}(\bolleke)=\bra{\bolleke}\measurement{A}\ket{\bolleke}\)\\
\\

\multicolumn{2}{l}{\textit{Preference and desirability}} \\[.35ex]
\(\betterthan\) & strict preference relation between uncertain rewards or measurements: uncertainty model\\
\(\alwaysbetterthan\), \(\alwaysstrictlybetterthan\) & weak and strict pointwise dominance relations\\[.1ex]
\(\measurement{A}\alwaysstrictlybetterthan\zero\) & measurement~\(\measurement{A}\) is positive definite, equivalent to \(\min\spectrum{A}>0\)\\
\(\measurement{A}\possemidef\zero\) & measurement~\(\measurement{A}\) is positive semidefinite, equivalent to \(\min\spectrum{A}\geq 0\)\\
\(\possemidefmeasurements\) & set of positive semidefinite Hermitian operators\\[.1ex]
\(\posmeasurements\) & set of non-zero positive semidefinite Hermitian operators\\[.1ex]
\(\strictlyposmeasurements\) & set of positive definite Hermitian operators\\
\(\desirset\) & coherent set of desirable measurements; partial preference information: uncertainty model\\
\(\posi\group{\assessment}\) & positive hull of~\(\assessment\): all positive linear combinations of elements in~\(\assessment\)\\
\(\ext\group{\assessment}\) & natural extension, \(\ext\group{\assessment} =\posi\group{\assessment\cup\posmeasurements}\), most conservative coherent set of desirable measurements that includes~\(\assessment\)\\
\\

\multicolumn{2}{l}{\textit{Previsions and density-operator representations}} \\[.35ex]
\(\lowprev\), \(\uppprev\) & coherent lower and upper previsions (buying and selling prices): uncertainty model\\
\(\linprev\) & coherent linear prevision, when \(\lowprev=\uppprev\) (self-conjugacy), corresponds to the precise probability case\\
\(\density\) & density operator; mixed state representation\\
\(\densities\) & sets of density operators: uncertainty model\\
\end{longtable}

\section{POVMs as acts in our framework}\label{app:POVMs}
In \cref{sec:quantum:mechanics:probabilistic}, measurements were introduced as Hermitian operators, also called \emph{projective measurements}.
In many a quantum information theory setting, the most general type of measurement is a \emph{positive operator-valued measure} (POVM).
A POVM consists of a set~\(\mathsf{E}=\{\measurement{E}_1,\dots,\measurement{E}_d\}\) of~\(d\) positive semidefinite Hermitian linear operators~\(\measurement{E}_k\in\possemidefmeasurements\) such that \(\sum_{k=1}^d\measurement{E}_k=\identity\).\footnote{See also the standard textbook discussion in \cite[Sections~2.2.6 and~2.2.8]{nielsen2010:quantum} for a more elaborate overview.}
Each such operator~\(\measurement{E}_k\) is associated with a possible outcome for the experiment, which we'll identify as outcome~\(k\).
If the epistemic uncertainty about the system state is represented by the density operator~\(\density\in\densities\), then the most general form of Born's rule gives the probability of outcome~\(k\) as
\begin{equation}\label{eq:born:rule:povm}
\prob{k}=\trace{\measurement{E}_k\density}.
\end{equation}
The main difference between a POVM and a projective measurement corresponding to a Hermitian operator, is then that POVMs allow for general positive semidefinite operators~\(\measurement{E}_k\), while projective measurements only allow for projection operators on the eigenspaces~\(\projection[k]\).

The natural question that now arises, is how we can make such POVMs fit into the decision-theoretic framework developed in this paper.
In this Appendix, we give a brief indication of how this could be achieved.

To get a fully analogous treatment of POVMs as \emph{acts} in our framework, we will assign to each outcome~\(k\) of the POVM some real value~\(\eigval_k\), which we interpret as the pay-off You obtain when outcome~\(k\) occurs.

We now aim to derive a reward function~\(\uval_{\mathsf{E}}\) for the POVM~\(\mathsf{E}\).
We can do this without the need of any additional postulates, through Naimark's theorem \cite{naimark1943}, which states that any POVM can be implemented as a projective measurement on a larger Hilbert space.

More specifically, consider a POVM~\(\mathsf{E}=\{\measurement{E}_1,\ldots,\measurement{E}_d\}\) on~\(\statespace\), to which have been assigned the corresponding pay-offs~\(\eigval_1,\ldots,\eigval_d\).
In addition, fix some \(d\)-dimensional so-called \emph{ancilla} space~\(\statespace_\mathrm{a}\) with orthonormal basis \(\set{\basisket[1],\dots,\basisket[d]}\) and also fix some state~\(\ket{c}\in\statespace_\mathrm{a}\).
Then it's possible \cite[Section~2.2.8]{nielsen2010:quantum} to construct a unitary operator~\(\unitary\) on the enlarged state space~\(\statespace\otimes\statespace_\mathrm{a}\) for which
\begin{equation*}
\unitary\group{\gket\otimes\ket{c}}
=\sum_{k=1}^d\group[\Big]{\sqrt{\measurement{E}_k}\gket}\otimes\basisket[k]
\text{ for all }\gket\in\statespace,
\end{equation*}
and where \(\sqrt{\measurement{E}_k}\) is the unique positive semidefinite square root of the positive semidefinite~\(\measurement{E}_k\).
On this enlarged state space, define the Hermitian operator~\(\measurement{A}\) by letting
\begin{equation*}
\measurement{A}
\coloneqq\adjoint{\unitary}\group[\bigg]{\sum_{k=1}^d\eigval_k\identity\otimes\basisket[k]\basisbra[k]}\unitary.
\end{equation*}
By construction \cite{naimark1943}, the projective measurement~\(\measurement{A}\) on states~\(\gket\otimes\ket{c}\) is equivalent to the POVM measurement on states~\(\gket\) alone.

Given the equivalence of the POVM~\(\mathsf{E}=\{\measurement{E}_1,\ldots,\measurement{E}_d\}\) on states~\(\gket\) to the projective measurement~\(\measurement{A}\) on states~\(\gket\otimes\ket{c}\), we can take its reward function to be identical to that of the corresponding projective measurement~\(\measurement{A}\) on the enlarged state space, so
\begin{align*}
\uval_{\mathsf{E}}(\gket)
\coloneqq&\theutility{A}(\gket\otimes\ket{c})
=\group{\gbra\otimes\bra{c}}\measurement{A}\group{\gket\otimes\ket{c}}\\
=&\group{\gbra\otimes\bra{c}}\adjoint{\unitary}\group[\bigg]{\sum_{k=1}^d\eigval_k\identity\otimes\basisket[k]\basisbra[k]}\unitary\group{\gket\otimes\ket{c}}\\
=&\group[\bigg]{\sum_{r=1}^d\gbra\adjoint{\sqrt{\measurement{E}_r}}\otimes\basisbra[r]}\group[\bigg]{\sum_{k=1}^d\eigval_k\identity\otimes\basisket[k]\basisbra[k]}\group[\bigg]{\sum_{s=1}^d{\sqrt{\measurement{E}_s}}\gket\otimes\basisket[s]}\\
=&\gbraketwithop[\bigg]{\sum_{k=1}^d\eigval_k\measurement{E}_k},
\text{ for all }\gket\in\statespace.
\end{align*}
So we see that
\begin{equation*}
\uval_{\mathsf{E}}
=\uval_{\measurement{A}_{\mathsf{E}}},
\text{where we let \(\measurement{A}_{\mathsf{E}}\coloneqq\sum_{k=1}^d\eigval_k\measurement{E}_k\).}
\end{equation*}
We conclude that the reward function for the POVM is identical to the reward function for the Hermitian operator~\(\measurement{A}_{\mathsf{E}}\coloneqq\sum_{k=1}^d\eigval_k\measurement{E}_k\) that `averages' the POVM elements~\(\measurement{E}_k\) with the `weights'~\(\lambda_k\) assigned to their outcomes~\(k\).

From an uncertainty modelling perspective, this means that if you deem \(\measurement{A}_{\mathsf{E}}\) to be a desirable measurement, then you also deem the POVM~\(\mathsf{E}=\{\measurement{E}_1,\ldots,\measurement{E}_d\}\) with the pay-offs~\(\eigval_k\) assigned to its outcomes~\(k\) to be desirable, and vice versa.
It's therefore enough to consider the reward functions for Hermitian operators, as we have done in this paper, to also capture the reward functions for POVMs.

We conclude that the representation result in \cref{thm:utility} and the desirability/lower-prevision machinery for Hermitian measurements extend to POVMs without any drastic changes to the set-up.
In particular, if Your uncertainty is represented by a coherent lower prevision~\(\lowprev\) on~\(\measurements\), with associated closed convex set of density operators~\(\densities(\lowprev)\), then the lower prevision of the POVM act~\(\mathsf{E}\) is given by
\begin{align}
\lowprev(\mathsf{E})
\coloneqq\lowprev\group{\measurement{A}_{\mathsf{E}}}
&=\min\cset{\trace{\measurement{A}_{\mathsf{E}}\density}}{\density\in\densities(\lowprev)}
\notag\\
&=\min\cset[\bigg]{\sum_{k=1}^d\eigval_k\trace{\measurement{E}_k\density}}{\density\in\densities(\lowprev)},
\label{eq:povm:lower:prevision}
\end{align}
which is the POVM counterpart of \cref{eq:from:set:of:densities:to:lowprev}.
In the precise case, where \(\densities(\lowprev)=\{\density\}\), \cref{eq:povm:lower:prevision} reduces to
\begin{equation*}
\linprev[\density](\measurement{A}_{\mathsf{E}})
=\sum_{k=1}^d\eigval_k\trace{\measurement{E}_k\density},
\end{equation*}
which can be reinterpreted as the standard expectation \(\expec[\density]{\measurement{A}_{\mathsf{E}}}=\sum_{k=1}^d\eigval_k\prob{k}\) for POVMs resulting from Born's rule in \cref{eq:born:rule:povm}.
\end{document}